\pdfoutput=1
\documentclass[a4paper,UKenglish,cleveref, autoref,thm-restate,numberwithinsect]{lipics-v2021}

\title{A Coinductive Version of Milner's Proof System for Regular Expressions Modulo Bisimilarity}

\titlerunning{A Coinductive Version of Milner's Proof System} 

\author{Clemens Grabmayer}{Gran Sasso Science Institute, Viale F.\ Crispi, 7, Italy \and \url{https://www.gssi.it/people/post-doc/post-doc-computer-science/item/4733-grabmayer-clemens} }{clemens.grabmayer@gssi.it}{https://orcid.org/0000-0002-1825-0097}
{}

\authorrunning{C. Grabmayer} 

\Copyright{Clemens Grabmayer} 



\category{} 

\relatedversion{} 

\funding{The completion of this research was supported by the PRIN project
         \emph{{\nf IT MATTERS} - Methods and Tools for Trustworthy Smart Systems}
         (with project ID: \text{2017FTXR7S\_005}).} 

\acknowledgements{I want to thank Luca Aceto for his comments on the introduction.
                  I am also very grateful to the anonymous reviewers for their careful reading,
                    for spotting distracting oversights, and 
                    for asking for additional clarifications, such as concerning the motivation of~coinductive~proofs. 
  }

\nolinenumbers 

\hideLIPIcs  

\EventEditors{Fabio Gadducci and Alexandra Silva}
\EventNoEds{2}
\EventLongTitle{9th Conference on Algebra and Coalgebra in Computer Science}
\EventShortTitle{CALCO 2021}
\EventAcronym{CALCO 2021}
\EventYear{2021}
\EventDate{August 23--September 3, 2021}
\EventLocation{Salzburg, Austria}
\EventLogo{}
\SeriesVolume{}
\ArticleNo{}

\usepackage{afterpage}
\usepackage{amsmath,amsthm}
\usepackage{amssymb}
\usepackage{babel}
\usepackage{bussproofs}
\usepackage{calc}
\usepackage{centernot}
\usepackage{chngcntr}
\usepackage[dmyyyy]{datetime}
\usepackage{esvect}
\usepackage{fancybox}
\usepackage{hyperref}
\usepackage{multirow}
\usepackage{mathtools}
\usepackage{placeins}
\usepackage{soul}   
\usepackage{stmaryrd}
\usepackage{url}
\usepackage{ushort}
\usepackage[normalem]{ulem}
\let\savebigtimes\bigtimes
\let\bigtimes\relax
\usepackage{mathabx}
\let\bigtimes\savebigtimes
\usepackage{graphicx}
\usepackage{tikz,fp}
\usetikzlibrary{arrows.meta,arrows,calc,automata,angles,quotes,matrix,
                positioning,shapes.geometric,shapes.symbols,shapes.multipart,through,trees,
                decorations.markings,decorations.text}

\pgfdeclaredecoration{funbisim}{final}{
  \state{final}[width=\pgfdecoratedpathlength]{
    \draw[->]
      (0,\pgfdecorationsegmentamplitude+0.1mm) -- +(\pgfdecoratedpathlength,0);
    \draw[-]
      (0,-\pgfdecorationsegmentamplitude-0.1mm) -- +(\pgfdecoratedpathlength,0);}}
\tikzset{
  funbisim/.style={
    decoration={funbisim, amplitude=0.25ex},
    decorate,
    funbisim options/.style={#1}    
  }}

\pgfdeclaredecoration{bisim}{final}{
  \state{final}[width=\pgfdecoratedpathlength]{
    \draw[<->]
      (0,\pgfdecorationsegmentamplitude+0.1mm) -- +(\pgfdecoratedpathlength,0);
    \draw[-]
      (0,-\pgfdecorationsegmentamplitude-0.1mm) -- +(\pgfdecoratedpathlength,0);}}
\tikzset{
  bisim/.style={
    decoration={bisim, amplitude=0.25ex},
    decorate,
    bisim options/.style={#1}    
  }}

\makeatletter
\def\calcLength(#1,#2)#3{%
\pgfpointdiff{\pgfpointanchor{#1}{center}}%
             {\pgfpointanchor{#2}{center}}%
\pgf@xa=\pgf@x%
\pgf@ya=\pgf@y%
\FPeval\@temp@a{\pgfmath@tonumber{\pgf@xa}}%
\FPeval\@temp@b{\pgfmath@tonumber{\pgf@ya}}%
\FPeval\@temp@sum{(\@temp@a*\@temp@a+\@temp@b*\@temp@b)}%
\FProot{\FPMathLen}{\@temp@sum}{2}%
\FPround\FPMathLen\FPMathLen5\relax
\global\expandafter\edef\csname #3\endcsname{\FPMathLen}
}
\makeatother

\tikzset{
  my dash/.style={dash pattern=on 5pt off 2pt}
         }

\newtheorem{thm}{Theorem}[section]{\bfseries}{\itshape}
\newtheorem{lem}[thm]{Lemma}{\bfseries}{\itshape}
{\normalfont\bfseries}{\itshape}
\newtheorem*{repeatedlem}{Lemma}{\normalfont\bfseries}{\itshape}
\newtheorem{prop}[thm]{Proposition}
{\normalfont\bfseries}{\itshape}

\theoremstyle{definition}
\newtheorem{defi}[thm]{Definition}{\bfseries}{\normalfont}
\newtheorem{question}[thm]{Question}
 
\newtheorem{exa}[thm]{Example}{\bfseries}{\normalfont}
\newtheorem{nonexa}[thm]{Non-Example}{\bfseries}{\normalfont}
\newtheorem*{exanonum}{Example}
\newtheorem{rem}[thm]{Remark}
\setul{1pt}{.4pt}  

\DeclareFontFamily{U}{mathx}{\hyphenchar\font45}   
\DeclareFontShape{U}{mathx}{m}{n}{
      <5> <6> <7> <8> <9> <10>
      <10.95> <12> <14.4> <17.28> <20.74> <24.88>
      mathx10
      }{}
\DeclareSymbolFont{mathx}{U}{mathx}{m}{n}
\DeclareFontSubstitution{U}{mathx}{m}{n}
\DeclareMathAccent{\widecheck}{0}{mathx}{"71}
\DeclareMathAccent{\wideparen}{0}{mathx}{"75}

\definecolor{azure}{rgb}{0.94,1.00,1.00}
\definecolor{brown}{rgb}{.75,.25,.25}
\definecolor{cyan}{rgb}{0.25,0.88,0.82}
\definecolor{chocolate}{rgb}{0.82,0.41,0.12}
\definecolor{darkcyan}{rgb}{0.5,0,1}
\definecolor{darkgreen}{rgb}{0,0.39,0}
\definecolor{darkmagenta}{rgb}{0.5,0,0.5}
\definecolor{darkgoldenrod}{RGB}{184,134,11}
\definecolor{firebrick}{RGB}{175,25,25}
\definecolor{forestgreen}{rgb}{0.13,0.55,0.13}
\definecolor{goldenrod}{RGB}{218,165,32}
\definecolor{lightcyan}{rgb}{0.88,1.00,1.00}
\definecolor{lightpink}{rgb}{1.00,0.71,0.76}
\definecolor{myyellow}{RGB}{235,235,0}
\definecolor{lightyellow}{rgb}{1.00,1.00,0.88}
\definecolor{lightgoldenrod}{rgb}{0.83,0.97,0.51}
\definecolor{lightgoldenrodyellow}{rgb}{0.98,0.98,0.82}
\definecolor{lightskyblue}{rgb}{0.53,0.81,0.98}
\definecolor{moccasin}{rgb}{1.00,0.89,0.71}
\definecolor{magenta}{rgb}{1,0,1}
\definecolor{navyblue}{rgb}{0,0,0.5}
\definecolor{orange}{rgb}{1.0,0.65,0.0}
\definecolor{orangered}{rgb}{1.0,0.27,0.0}
\definecolor{palegreen}{rgb}{0.60,0.98,0.60}
\definecolor{paleblue}{rgb}{0.678,0.847,0.902}
\definecolor{powderblue}{rgb}{0.69,0.88,0.90}
\definecolor{purple}{rgb}{1,0.5,1}
\definecolor{royalblue}{RGB}{65,105,225}
\definecolor{mediumblue}{RGB}{0,0,205}
\definecolor{cornflowerblue}{RGB}{100,149,237}
\definecolor{springgreen}{rgb}{0.0,1.0,0.5}
\definecolor{turquoise}{rgb}{0.25,0.88,0.82}
\definecolor{snow}{rgb}{1.00,0.98,0.98}
\definecolor{tan}{rgb}{0.82,0.71,0.55}
\definecolor{red}{rgb}{1,0,0}
\definecolor{violetred}{RGB}{208,32,144}

\newcommand{\colorin}[1]{\textcolor{#1}}

\newcommand{\black}{\colorin{black}}

\newcommand{\chocolate}{\colorin{chocolate}}
\newcommand{\colorred}{\colorin{red}}
\newcommand{\alert}{\colorred}

\newcommand{\darkcyan}{\colorin{darkcyan}}

\newcommand{\forestgreen}{\colorin{forestgreen}}
\newcommand{\firebrick}{\colorin{firebrick}}

\newcommand{\mediumblue}{\colorin{mediumblue}}

\newcommand{\nb}{\nobreakdash}
\newcommand{\punc}[1]{\ensuremath{\hspace*{1.5pt}{#1}}}

\newcommand{\nf}{\normalfont}

\newenvironment{new}{\color{chocolate}}{\color{black}}
\newenvironment{newer}{\color{firebrick}}{\color{black}}
\newenvironment{newest}{\color{red}}{\color{black}}
\newenvironment{revised}{\color{violetred}}{\color{black}}
\newenvironment{change}{\color{violetred}}{\color{black}}
\newcommand{\funin}{\mathrel{:}}
\newcommand{\fap}[2]{{#1}(\hspace*{-0.5pt}{#2}\hspace*{-0.5pt})}

\newcommand{\bfap}[3]{{#1}({#2},\hspace*{0.5pt}{#3})}

\newcommand{\iap}[2]{#1_{#2}}
\newcommand{\bap}[2]{#1_{#2}}
\newcommand{\pap}[2]{#1^{#2}}
\newcommand{\bpap}[3]{#1_{#2}^{#3}}
\newcommand{\pbap}[3]{#1_{#3}^{#2}}
\newcommand{\sidfun}{\textrm{\nf id}}

\newcommand{\sproj}{\pi}

\newcommand{\proj}{\fap{\sproj}}

\newcommand{\sdefdby}{{:=}}
\newcommand{\defdby}{\mathrel{\sdefdby}}

\newcommand\tuple[1]{\langle #1 \rangle}

\newcommand\tuplespace{\hspace*{0.5pt}}
\newcommand\pair[2]{\tuple{#1, \tuplespace #2}}

\newcommand{\nat}{\mathbb{N}}
\newcommand{\natplus}{\pap{\nat}{+}} 

\newcommand{\BNFor}{\:\mid\:}
\newcommand{\BNFdefdby}{\:{::=}\:}

\newcommand{\family}[2]{\setexp{#1}_{#2}}

\newcommand{\slexspo}{{<}_{\text{\nf lex}}}
\newcommand{\lexspo}{\mathrel{\slexspo}}

\newcommand{\ssynteq}{{\equiv}} 
\newcommand{\synteq}{\mathrel{\ssynteq}} 

\newcommand{\sformeq}{{=}}
\newcommand{\formeq}{\mathrel{\sformeq}}

\newcommand{\sred}{{\to}}

\newcommand{\sredi}[1]{{\iap{\sred}{#1}}}
\newcommand{\redi}[1]{\mathrel{\sredi{#1}}}

\newcommand{\sconvred}{{\leftarrow}}

\newcommand{\sconvredi}[1]{{\iap{\leftarrow}{#1}}}
\newcommand{\convredi}[1]{\mathrel{\sconvredi{#1}}}

\newcommand{\sconvredtci}[1]{{\pbap{\leftarrow}{+}{#1}}}
\newcommand{\convredtci}[1]{\mathrel{\sconvredtci{#1}}}
\newcommand{\sredrtc}{\sred^{*}}

\newcommand{\sredrtci}[1]{{\iap{\sredrtc}{#1}}}
\newcommand{\redrtci}[1]{\mathrel{\sredrtci{#1}}}

\newcommand{\sredb}[1]{{\bap{\sred}{#1}}}

\newcommand{\sredbp}[2]{{\bpap{\sred}{#1}{#2}}}
\newcommand{\sconvredb}[1]{{\bap{\sconvred}{#1}}}

  \newcommand{\subotr}{\hspace*{-1pt}{\scriptscriptstyle (\sone)}}

\newcommand{\sotelimred}{\sredb{\subotr}}
\newcommand{\otelimred}{\mathrel{\sotelimred}}

\newcommand{\sotelimredtc}{\sredbp{\subotr}{+}}

\newcommand{\sotelimconvred}{\sconvredb{\subotr}}

\newcommand{\bodylabcol}{\text{\nf\darkcyan{bo}}}
\newcommand{\bodylab}{\bodylabcol}

\newcommand{\descsetexpmid}{\mathrel{\vert}}
\newcommand{\descsetexpbigmid}{\mathrel{\big\vert}}
\newcommand{\descsetexpBigmid}{\mathrel{\Big\vert}}

\newcommand{\descsetexp}[2]{\left\{{#1}\descsetexpmid{#2}\right\}}
\newcommand{\descsetexpnormalsize}[2]{\{{#1}\descsetexpmid{#2}\}}

\newcommand{\descsetexpbig}[2]{\bigl\{{#1}\descsetexpbigmid{#2}\bigr\}}
\newcommand{\descsetexpBig}[2]{\Bigl\{{#1}\descsetexpBigmid{#2}\Bigr\}}

\newcommand{\scompfuns}[2]{{#1}\circ{#2}}
\newcommand{\compfuns}[2]{\fap{\scompfuns{#1}{#2}}}

\newcommand{\sabinrel}{{R}}
\newcommand{\abinrel}{\mathrel{\sabinrel}}

\newcommand{\setexp}[1]{\left\{{#1}\right\}}

\newcommand{\setexpbig}[1]{\bigl\{{#1}\bigr\}}

\renewcommand{\emptyset}{\varnothing}

\newcommand{\slogand}{\wedge}
\newcommand{\slogor}{\vee}
\newcommand{\logand}{\mathrel{\slogand}}
\newcommand{\logor}{\mathrel{\slogor}} 
\newcommand{\slognot}{\neg}
\newcommand{\lognot}[1]{\slognot{#1}}

\newcommand{\tightfbox}[1]{{\fboxsep=1.5pt\fbox{#1}}}

\newcommand{\actiontarget}{ac\-tion-tar\-get}

\newcommand{\eqlogicbased}{$\eqlogic$\nb-based}

\newcommand{\immediatetermination}{im\-me\-di\-ate-ter\-mi\-na\-tion}

\newcommand{\loopelimination}{loop\nb-elim\-i\-nat\-ion}

\newcommand{\loopentry}{loop-en\-try}
\newcommand{\loopbody}{loop-body}

\newcommand{\onederivative}{1\nb-de\-riv\-a\-tive}
\newcommand{\onederivatives}{\onederivative{s}}

\newcommand{\entrybodylabeling}{en\-try\discretionary{/}{}{/}body-la\-be\-ling}
\newcommand{\entrybodylabelings}{\entrybodylabeling{s}}

\newcommand{\LEEshape}{\LEE\nb-shape}

\newcommand{\LEEshaped}{\LEE\nb-shaped}

\newcommand{\LEEwitness}{$\LEE$\nb-wit\-ness}

\newcommand{\LEEwitnesses}{$\LEE$\hspace*{1.25pt}\nb-wit\-nes\-ses}

\newcommand{\LLEEwitness}{{\nf LLEE}\nb-wit\-ness}
\newcommand{\LLEEwitnesses}{{\nf LLEE}\nb-wit\-nesses}
\newcommand{\LLEEwitnessed}{{\nf LLEE}\nb-wit\-nessed}

\newcommand{\LLEEcharts}{{\nf LLEE}\nb-chart{s}}

\newcommand{\LLEEonechart}{{\nf LLEE}\hspace*{1pt}-\onechart}

\newcommand{\LLEEonecharts}{{LLEE}\hspace*{1pt}-\onechart{s}}

\newcommand{\nontrivial}{non-triv\-i\-al}

\newcommand{\onetransition}{$\sone$\nb-tran\-si\-tion}
\newcommand{\onetransitions}{\onetransition{s}}

\newcommand{\onefree}{$1$\nb-free}
\newcommand{\sonefree}{$\sone$\nb-free}

\newcommand{\onechart}{$\sone$\nb-chart}

\newcommand{\onecharts}{\onechart{s}}

\newcommand{\onebisimulation}{$\sone$\nb-bi\-si\-mu\-la\-tion}
\newcommand{\onebisimulations}{\onebisimulation{s}}

\newcommand{\onebisimilar}{$\sone$\nb-bi\-si\-mi\-lar}

\newcommand{\onebisimilarity}{$\sone$\nb-bi\-si\-mi\-la\-ri\-ty}

\newcommand{\perpetual}{per\-pet\-u\-al}
\newcommand{\perpetualloop}{\perpetual-loop}

\newcommand{\provablein}[1]{{$#1$}\nb-pro\-vable}

\newcommand{\provablyin}[1]{{$#1$}\nb-pro\-vably}

\newcommand{\sidecondition}{side-con\-di\-tion}

\newcommand{\subonechart}{sub-$\sone$\nb-chart}
\newcommand{\subonecharts}{\subonechart{s}}

\newcommand{\starexpression}{star-ex\-pres\-sion}

\newcommand{\theoremsubsumed}{the\-o\-rem-sub\-sumed}
\newcommand{\theoremequivalent}{the\-o\-rem-equiv\-a\-lent}

\newcommand{\transitionact}[1]{{${#1}$}\nb-tran\-si\-tion}
\newcommand{\transitionsact}[1]{\transitionact{#1}s}
\newcommand{\derivativeact}[1]{{${#1}$}\nb-deri\-va\-tive}
\newcommand{\derivativesact}[1]{\derivativeact{#1}s}

\newcommand{\weaklyguarded}{weak\-ly guard\-ed}
\newcommand{\wg}{w.g.}

\newcommand{\welldefined}{well-de\-fined}

\newcommand{\wellfounded}{well-found\-ed}

\newcommand{\overlinebar}[1]{\mathbf{\overline{\text{$#1$}}}}

\newcommand{\aspec}{\mathcal{S}}

\newcommand{\astexp}{e}
\newcommand{\bstexp}{f}
\newcommand{\cstexp}{g}
\newcommand{\dstexp}{h}

\newcommand{\astexpi}{\iap{\astexp}}
\newcommand{\bstexpi}{\iap{\bstexp}}
\newcommand{\cstexpi}{\iap{\cstexp}}
\newcommand{\dstexpi}{\iap{\dstexp}}

\newcommand{\asstexp}{E}

\newcommand{\asstexpi}{\iap{\asstexp}}

\newcommand{\asstexpacc}{\asstexp'}

\newcommand{\asstexpacci}{\iap{\asstexpacc}}

\newcommand{\snormedplus}{{\textit{nd}^+}}
\newcommand{\normedplus}{\fap{\snormedplus}}

\newcommand{\astexpacc}{\astexp'}

\newcommand{\astexpacci}{\iap{\astexpacc}}

  \newcommand{\StExps}{\textit{StExp}}
  \newcommand{\StExpsover}{\fap{\StExps}}

\newcommand{\StExp}{\textit{StExp}}
\newcommand{\StExpover}{\fap{\StExp}}

\newcommand{\StExpEq}{\textit{Eq}}
\newcommand{\StExpEqover}{\fap{\StExpEq}}

\newcommand{\acxt}{C}

\newcommand{\cxtap}[2]{{#1}[{#2}]}
\newcommand{\acxtap}{\cxtap{\acxt}}

\newcommand{\iets}{\hspace*{0.5pt}}

\newcommand{\stexpzero}{0}
\newcommand{\stexpone}{1}

\newcommand{\sstexpit}{\sstar}
\newcommand{\stexpit}[1]{{#1^{\sstexpit}}}

\newcommand{\sstexpprod}{{\cdot}}
\newcommand{\stexpprod}[2]{{#1}\mathrel{\sstexpprod}{#2}}
\newcommand{\sstexpsum}{+}
\newcommand{\stexpsum}[2]{{#1}\sstexpsum{#2}}
\newcommand{\sstexpstackprod}{\chocolate{\varstar}} 
\newcommand{\stexpstackprod}[2]{{#1}\mathrel{\sstackprod}{#2}}
  
\newcommand{\ssprod}{{\cdot}}
\renewcommand{\prod}{\,{\ssprod}\,}%
\newcommand{\sstackprod}{\chocolate{\varstar}}
\newcommand{\stackprod}{\,{\sstackprod}\,}%

\newcommand{\sth}[1]{|{#1}|_{\scalebox{0.8}{$\sstar$}}}

\newcommand{\sdescrelstexpit}{\sredi{\scriptscriptstyle(\sstar)}}

\newcommand{\descrelstexpit}[1]{\mathrel{\sdescrelstexpit}}

\newcommand{\sconvdescrelstexpit}{\sconvredi{\scriptscriptstyle(\sstar)}}

\newcommand{\convdescrelstexpit}[1]{\mathrel{\sconvdescrelstexpit}}

\newcommand{\sspartderivs}{\partial}
\newcommand{\spartderivs}{\iap{\partial}}
\newcommand{\partderivs}[1]{\fap{\spartderivs{#1}}}

\newcommand{\sitpartonederivs}{{\underline{\partial}^+}}
\newcommand{\itpartonederivs}{\fap{\sitpartonederivs}}

\newcommand{\sactderivs}{\textit{A\hspace*{-0.25pt}$\sspartderivs$}}
\newcommand{\actderivs}{\fap{\sactderivs}}

\newcommand{\soneactderivs}{\textit{\underline{A\hspace*{-0.25pt}$\sspartderivs$}}}
\newcommand{\oneactderivs}{\fap{\soneactderivs}}

\newcommand{\slt}[1]{{\xrightarrow{#1}}}
\newcommand{\slti}[2]{{\xrightarrow{#1}}{_{#2}}}

\newcommand{\lt}[1]{\mathrel{\slt{#1}}}
\newcommand{\lti}[2]{\mathrel{\slti{#1}{#2}}}

\newcommand{\sone}{\firebrick{1}}
\newcommand{\sstar}{*}

\newcommand{\looplab}[1]{{\darkcyan{[#1]}}}  
\newcommand{\loopsteplab}{\looplab}

\newcommand{\loopnsteplab}[1]{\darkcyan{[{#1}]}}
\newcommand{\sloopnstepto}[1]{{\iap{\rightarrow}{\loopnsteplab{#1}}}}
\newcommand{\loopnstepto}[1]{\mathrel{\sloopnstepto{#1}}}

\newcommand{\loopentrytransition}{loop-entry tran\-si\-tion}
\newcommand{\loopentrytransitions}{\loopentrytransition{s}}

\newcommand{\aLname}{n}

\newcommand{\sterminates}{{\downarrow}}
\newcommand{\terminates}[1]{{#1}{\sterminates}}
\newcommand{\sterminatesbig}{{\big\downarrow}}
\newcommand{\terminatesbig}[1]{{#1}{\sterminatesbig}}

\newcommand{\snotterminates}{\ndownarrow}  
\newcommand{\notterminates}[1]{{#1}{\snotterminates}}

\newcommand{\onebrackscript}{\scalebox{0.75}{$\scriptstyle (1)$}}

\newcommand{\soneterminates}{{\pap{\downarrow}{\hspace*{-1.5pt}\firebrick{\onebrackscript}}}}
\newcommand{\oneterminates}[1]{{#1}{\soneterminates}}
\newcommand{\soneterminatesi}[1]{{\pbap{\downarrow}{\hspace*{-1.25pt}\firebrick{\onebrackscript}}{#1}}}
\newcommand{\oneterminatesi}[2]{{#2}{\soneterminatesi{#1}}}

\newcommand{\achart}{\mathcal{C}}
\newcommand{\acharti}{\iap{\achart}}
\newcommand{\bchart}{\mathcal{D}}
\newcommand{\bcharti}{\iap{\bchart}}

\newcommand{\acharthat}{\hspace*{0.75pt}\Hat{\hspace*{-0.75pt}\achart}\hspace*{-0pt}} 

\newcommand{\aonechart}{\underline{\mathcal{C}}}
\newcommand{\aonecharti}{\iap{\aonechart}}
\newcommand{\bonechart}{\underline{\mathcal{D}}}
\newcommand{\bonecharti}{\iap{\bonechart}}
\newcommand{\conechart}{\underline{\mathcal{E}}}

\newcommand{\aonechartacc}{\underline{\mathcal{C}}\hspace*{-0.4pt}'}
\newcommand{\bonechartacc}{\underline{\mathcal{D}}\hspace*{0pt}'}
\newcommand{\aonechartacci}{\iap{\aonechartacc}}
\newcommand{\bonechartacci}{\iap{\bonechartacc}}
\newcommand{\aonechartdacc}{\underline{\mathcal{C}}''}
\newcommand{\aonechartdacci}{\iap{\aonechartdacc}}
\newcommand{\aonecharthat}{\hspace*{0.2pt}\Hat{\hspace*{-0.75pt}\aonechart}\hspace*{-0.75pt}} 
\newcommand{\aonecharthati}[1]{\hspace*{0.2pt}\iap{\Hat{\hspace*{-0.75pt}\aonechart}}{#1}\hspace*{-0.75pt}}
\newcommand{\bonecharthati}[1]{\hspace*{0.2pt}\iap{\Hat{\hspace*{-0.75pt}\bonechart}}{#1}\hspace*{-0.75pt}} 
\newcommand{\conecharthati}[1]{\hspace*{0.2pt}\iap{\Hat{\hspace*{-0.75pt}\conechart}}{#1}\hspace*{-0.75pt}} 

\newcommand{\aonecharthatacci}[1]{\hspace*{0.2pt}\aonecharthat\hspace*{0.4pt}'_{#1}\hspace*{-0.75pt}}
\newcommand{\aonecharthatdacci}[1]{\hspace*{0.2pt}\aonecharthat\hspace*{0.4pt}''_{#1}\hspace*{-0.75pt}}

\newcommand{\bonecharthatsubscript}{\widehat{\smash{\bonechart}\rule{0pt}{4.75pt}}} 
\newcommand{\bonecharthataccsubscript}{\widehat{\smash{\bonechartacc}\rule{0pt}{4.75pt}}} 
\newcommand{\bonecharthatsubscripti}[1]{\widehat{\smash{\bonechart}\rule{0pt}{4.75pt}}_{#1}} 
\newcommand{\bonecharthataccsubscripti}[1]{\widehat{\smash{\bonechart}\rule{0pt}{4.75pt}}{}'_{#1}} 

\newcommand{\chartof}{\fap{\achart}}
\newcommand{\onechartof}{\fap{\aonechart}}

\newcommand{\onecharthatof}{\fap{\aonecharthat}}

\newcommand{\indscchartof}[1]{\iap{#1}{\subscriptindtrans}} 

\newcommand{\indchart}{\indscchartof}

\newcommand{\LTS}{LTS}

\newcommand{\saTSS}{{\cal{T}}}
\newcommand{\StExpTSS}{\text{$\saTSS$}}
\newcommand{\StExpTSSover}[1]{\text{$\fap{\StExpTSS}{#1}$}}

  \newcommand{\subscriptindtrans}{\mediumblue{\scriptscriptstyle\pmb{\otind{\cdot}}}}

\newcommand{\aoneloop}{\underline{\mathcal{L}}}

\newcommand{\indsubonechartinat}[1]{\fap{\aonecharti{#1}}}

\newcommand{\otind}[1]{({#1}]}

\newcommand{\iact}[1]{\mediumblue{{\scriptscriptstyle\pmb (}\hspace*{-0pt}{\black{#1}}\hspace*{0.4pt}{\scriptscriptstyle\pmb ]}}}  
\newcommand{\silt}[1]{\slt{\iact{#1}}}

\newcommand{\silti}[2]{{\xrightarrow{\iact{#1}}}{_{#2}}}
\newcommand{\ilt}[1]{\mathrel{\silt{#1}}}

\newcommand{\ilti}[2]{\mathrel{\silti{#1}{#2}}}

\newcommand{\sasol}{s}
\newcommand{\sasoli}{\iap{\sasol}}
\newcommand{\asol}{\fap{\sasol}}
\newcommand{\asoli}[1]{\fap{\iap{\sasol}{#1}}}
\newcommand{\sbsol}{t}

\newcommand{\sextrsol}{\sasol}
\newcommand{\sextrsolof}{\iap{\sextrsol}}
\newcommand{\sextrsoluntil}{\sbsol}
\newcommand{\sextrsoluntilof}{\iap{\sextrsoluntil}}
\newcommand{\extrsol}{\fap{\sextrsol}}

\newcommand{\extrsolof}[1]{\fap{\sextrsolof{#1}}}
\newcommand{\extrsoluntilof}[1]{\bfap{\sextrsoluntilof{#1}}}

\newcommand{\sterminatesconst}{\tau}

\newcommand{\sterminatesconstof}{\iap{\sterminatesconst}}

\newcommand{\terminatesconstof}[1]{\fap{\sterminatesconstof{#1}}}
\newcommand{\actions}{\mathit{A}}

\newcommand{\oneactions}{\fap{\actions}{\sone}}

\newcommand{\aact}{a}
\newcommand{\bact}{b}
\newcommand{\cact}{c}
\newcommand{\dact}{d}

\newcommand{\aacti}{\iap{\aact}}

\newcommand{\cacti}{\iap{\cact}}

\newcommand{\aoneact}{\firebrick{\underline{a}}}
\newcommand{\boneact}{\firebrick{\underline{b}}}

\newcommand{\aoneacti}[1]{\firebrick{\iap{\aoneact}{#1}}}
\newcommand{\boneacti}[1]{\firebrick{\iap{\boneact}{#1}}}

\newcommand{\verts}{V}

\newcommand{\start}{\averti{\hspace*{-0.5pt}\text{\nf s}}}
\newcommand{\transs}{{\sred}}
\newcommand{\exts}{{\sterminates}}
\newcommand{\termexts}{{\sterminates}}

\newcommand{\ats}{{AT}}

\newcommand{\alab}{\darkcyan{l}}
\newcommand{\alabi}[1]{\iap{\alab}{\darkcyan{#1}}}

\newcommand{\vertsof}{\fap{\verts\hspace*{-1pt}}}

\newcommand{\atsiof}[1]{\fap{\atsi{#1}}}

\newcommand{\vertsi}[1]{\iap{\verts}{\hspace*{-0.25pt}{#1}}}

\newcommand{\starti}[1]{\averti{\text{\nf s},#1}}

\newcommand{\termextsi}[1]{\iap{\termexts}{{#1}}}
\newcommand{\atsi}[1]{\iap{\ats}{\hspace*{-1.5pt}{#1}}}

\newcommand{\stransitions}{T}
\newcommand{\transitionsinfrom}[1]{\fap{\iap{\stransitions}{#1}}}

\newcommand{\ainst}{\iota}

\newcommand{\asettranss}{U}

\newcommand{\avert}{v}
\newcommand{\bvert}{w}
\newcommand{\cvert}{u}

\newcommand{\averti}{\iap{\avert}}
\newcommand{\bverti}{\iap{\bvert}}
\newcommand{\cverti}{\iap{\cvert}}

\newcommand{\avertacc}{\avert'}

\newcommand{\avertacci}{\iap{\avertacc}}

\newcommand{\averttilde}{\tilde{\avert}}

\newcommand{\bvertbar}{\overlinebar{\bvert}}

\newcommand{\bvertbari}{\iap{\bvertbar}}

\newcommand{\atrans}{\tau}

\newcommand{\sfunbisim}{%
    \setbox0=\hbox{\kern-.1ex{$\rightarrow$}\kern-.1ex}
    \setbox1=\vbox{\hbox{\raise .1ex \box0}\hrule}%
    {\hbox{\kern.05ex\box1\kern.1ex}}
  }

\newcommand{\sconvfunbisim}{%
    \setbox0=\hbox{\kern-.1ex{$\leftarrow$}\kern-.1ex}
    \setbox1=\vbox{\hbox{\raise .1ex \box0}\hrule}%
    {\hbox{\kern.05ex\box1\kern.1ex}}
  }

\newcommand{\sbisim}{%
    \setbox0=\hbox{\kern-.1ex{$\leftrightarrow$}\kern-.1ex}
    \setbox1=\vbox{\hbox{\raise .1ex \box0}\hrule}%
    \hbox{\kern.1ex\box1\kern.1ex}
  }
\newcommand{\bisim}{\mathrel{\sbisim\hspace*{1pt}}}

\newcommand{\sfunbisimos}{%
    \setbox0=\hbox{\kern-.1ex{$\rightarrow$}\kern-.1ex}
    \setbox1=\vbox{\hbox{\raise .1ex \box0}\hrule}%
    {\pap{\hbox{\kern.05ex\box1\kern.1ex}}{\hspace*{0.5pt}\subotr}}
  }

\newcommand{\sconvfunbisimos}{%
    \setbox0=\hbox{\kern-.1ex{$\leftarrow$}\kern-.1ex}
    \setbox1=\vbox{\hbox{\raise .1ex \box0}\hrule}%
    {\pap{\hbox{\kern.05ex\box1\kern.1ex}}{\hspace*{0.5pt}\subotr}}
  }

\newcommand{\sbisimos}{%
    \setbox0=\hbox{\kern-.1ex{$\leftrightarrow$}\kern-.1ex}
    \setbox1=\vbox{\hbox{\raise .1ex \box0}\hrule}%
    \ensuremath{\pap{\mathrel{\hbox{\kern.1ex\box1\kern.1ex}}}{\hspace*{0.5pt}\subotr}}
  }

\newcommand{\sfunbisimosvia}[1]{%
    \setbox0=\hbox{\kern-.1ex{$\rightarrow$}\kern-.1ex}
    \setbox1=\vbox{\hbox{\raise .1ex \box0}\hrule}%
    {\bpap{\hbox{\kern.05ex\box1\kern.1ex}}{#1}{\hspace*{0.5pt}\subotr}}
  }

\newcommand{\sconvfunbisimosvia}[1]{%
    \setbox0=\hbox{\kern-.1ex{$\leftarrow$}\kern-.1ex}
    \setbox1=\vbox{\hbox{\raise .1ex \box0}\hrule}%
    {\bpap{\hbox{\kern.05ex\box1\kern.1ex}}{#1}{\hspace*{0.5pt}\subotr}}
  }

\newcommand{\onebisim}{\mathrel{\sbisimos}}

\newcommand{\abisim}{B}

\newcommand{\sbehinc}{{\sqsubseteq}}

\newcommand{\sbehinca}[1]{{\prescript{#1}{}{\sbehinc}}}
\newcommand{\behinca}[1]{\mathrel{\sbehinca}}

\newcommand{\sonebehinc}{{\pap{\sbehinc}{\subotr}}}

\newcommand{\sonebehinca}[1]{{{}_{#1}\sonebehinc}}
\newcommand{\onebehinca}[1]{\mathrel{\sonebehinca}}
\newcommand{\REFL}{\textrm{\nf Ref\/l}}
\newcommand{\SYMM}{\textrm{\nf Symm}}
\newcommand{\TRANS}{\textrm{\nf Trans}}

\newcommand{\CXT}{\textrm{\nf Cxt}}

\newcommand{\sRSP}{\textrm{\nf RSP}}
\newcommand{\RSPstar}{\text{$\sRSP^{*}\hspace*{-1pt}$}}

\newcommand{\snotRSP}{\textrm{\nf\sout{RSP}}}
\newcommand{\notRSPstar}{\text{$\snotRSP^{*}\hspace*{-1pt}$}}

\newcommand{\sUSP}{\textrm{\nf USP}}
\newcommand{\USPstar}{\text{$\pap{\sUSP}{*}$}}
\newcommand{\USP}{\sUSP}

\newcommand{\assocstexpsum}{\fap{\text{\nf assoc}}{\sstexpsum}}
\newcommand{\assocstexpprod}{\fap{\text{\nf assoc}}{\sstexpprod}}
\newcommand{\commstexpsum}{\fap{\text{\nf comm}}{\sstexpsum}}
\newcommand{\neutralstexpsum}{\fap{\text{\nf neutr}}{\sstexpsum}}  
\newcommand{\idempotstexpsum}{\fap{\text{\nf idempot}}{\sstexpsum}}

\newcommand{\rdistr}{\fap{\text{\nf r-distr}}{\sstexpsum,\sstexpprod}}

\newcommand{\stexpzerostexpprod}{\text{\nf deadlock}}
\newcommand{\rightidstexpprod}{\fap{\text{\nf id}_{\text{\nf r}}}{\sstexpprod}}
\newcommand{\leftidstexpprod}{\fap{\text{\nf id}_{\text{\nf l}}}{\sstexpprod}}
\newcommand{\recdefstexpit}{\fap{\text{\nf rec}}{{}^{*}}}
\newcommand{\termstexpit}{\fap{\text{\nf trm-body}}{{}^{*}}}

\newcommand{\ACI}{\text{\sf ACI}}
\newcommand{\ACIover}{\fap{\ACI}}

\newcommand{\seqin}[1]{{\iap{=}{#1}\hspace*{1pt}}}
\newcommand{\eqin}[1]{\mathrel{\seqin{#1}}}
\newcommand{\seqinsol}[1]{{\pbap{=}{\text{\scriptsize (sol)}}{#1}\hspace*{1pt}}}
\newcommand{\eqinsol}[1]{\mathrel{\seqinsol{#1}}}

\newcommand{\ACIeq}{\eqin{\ACI}}

\newcommand{\milnersyseq}{\eqin{\milnersys}}

\newcommand{\milnersysmineq}{\eqin{\milnersysmin}}

\newcommand{\sderivablein}[1]{\vdash_{#1}}
\newcommand{\derivablein}[1]{\sderivablein{#1}}
\newcommand{\thmequiv}{\mathrel{\sim}}
\newcommand{\sisthmsubsumedby}{{\precsim}}
\newcommand{\isthmsubsumedby}{\mathrel{\sisthmsubsumedby}}
\newcommand{\sthmsubsumes}{{\succsim}}
\newcommand{\thmsubsumes}{\mathrel{\sthmsubsumes}}
\newcommand{\ssubsystem}{{\subseteq}}
\newcommand{\subsystem}{\mathrel{\ssubsystem}}

\newcommand{\CoindProofi}[1]{\text{\nf CoProof$_{#1}$}}
\newcommand{\LCoindProofi}[1]{\text{\nf LCoProof$_{#1}$}}

\newcommand{\saeqfun}{L}
\newcommand{\aeqfun}{\fap{\saeqfun}}
\newcommand{\saeqfuni}{\iap{\saeqfun}}
\newcommand{\aeqfuni}[1]{\fap{\iap{\saeqfun}{#1}}}

\newcommand{\scoindproofeqin}[1]{{\bap{\overset{\text{\nf coind}}{=\joinrel=\joinrel=}}{#1}}} 
\newcommand{\sLLEEcoindproofeqin}[1]{{\bap{\overset{\text{\nf\scriptsize LLEE}}{=\joinrel=\joinrel=}}{#1}}} 
\newcommand{\coindproofeqin}[1]{\mathrel{\scoindproofeqin{#1}}}
\newcommand{\LLEEcoindproofeqin}[1]{\mathrel{\sLLEEcoindproofeqin{#1}}}

\newcommand{\aLLEECoProof}{\mathcal{L\hspace*{-0.5pt}C\hspace*{-0.75pt}P}}
\newcommand{\aLLEECoProofoverof}[1]{\fap{\iap{\aLLEECoProof}{#1}}}
\newcommand{\aLLEECoProofacc}{\mathcal{L\hspace*{-0.5pt}C\hspace*{-0.75pt}P'}}
\newcommand{\aLLEECoProofaccoverof}[1]{\fap{\iap{\aLLEECoProofacc}{#1}}}
\newcommand{\aCoProof}{\mathcal{C\hspace*{-0.75pt}P}}

\newcommand{\sdescendsinloopto}{{\curvearrowright}}  
\newcommand{\descendsinloopto}{\mathrel{\sdescendsinloopto}}

\newcommand{\sconvdescendsinloopto}{{\curvearrowleft}} 
\newcommand{\convdescendsinloopto}{\mathrel{\sconvdescendsinloopto}}

\newcommand{\sconvdescendsinlooptotc}{{\pap{\sconvdescendsinloopto}{\hspace*{-0.8pt}+}}}
\newcommand{\convdescendsinlooptotc}{\mathrel{\sconvdescendsinlooptotc}}

\newcommand{\sconvdescendsinlooptosc}{{\pap{\sconvdescendsinloopto}{\hspace*{-0.8pt}=}}}
\newcommand{\convdescendsinlooptosc}{\mathrel{\sconvdescendsinlooptosc}}
\newcommand{\sloopsbackto}{{\lefttorightarrow}} 
\newcommand{\loopsbackto}{\mathrel{\sloopsbackto}}

\newcommand{\eqlogic}{\text{$\mathcal{E\hspace*{-1pt}L}$}}

\newcommand{\eqlogicover}[1]{\fap{\eqlogic}{#1}}

\newcommand{\milnersys}{\text{$\text{\sf Mil}$}}

\newcommand{\milnersysof}[1]{\text{$\fap{\text{\sf Mil}}{#1}$}}
\newcommand{\milnersysover}{\milnersysof}

\newcommand{\milnersysacc}{\text{\sf Mil$^{\prime}$}}
\newcommand{\milnersysaccover}[1]{\text{$\fap{\milnersysacc}{#1}$}}
\newcommand{\milnersysaccbar}{\text{$\overline{\sf Mil}{}^{\prime}$}}
\newcommand{\milnersysaccbarover}[1]{\text{$\fap{\milnersysaccbar}{#1}$}}

\newcommand{\milnersysmin}{\text{\sf Mil$^{\boldsymbol{-}}$}}
\newcommand{\milnersysminover}[1]{\text{$\fap{\milnersysmin}{#1}$}}

\newcommand{\coindmilnersys}{\text{\sf cMil}}      
\newcommand{\coindmilnersysof}[1]{\text{$\fap{\coindmilnersys}{#1}$}}
\newcommand{\coindmilnersysover}{\coindmilnersysof}

\newcommand{\coindmilnersysbar}{\overline{\text{\sf cMil}}}
\newcommand{\coindmilnersysbarof}[1]{\text{$\fap{\coindmilnersysbar}{#1}$}}
\newcommand{\coindmilnersysbarover}{\coindmilnersysbarof}

\newcommand{\coindmilnersysone}{{\text{{\sf cMil}$_1$}}}
\newcommand{\coindmilnersysoneof}[1]{\text{$\fap{\coindmilnersysone}{#1}$}}
\newcommand{\coindmilnersysoneover}{\coindmilnersysoneof}

\newcommand{\CC}{\text{\sf C\hspace*{-0.75pt}C}}

\newcommand{\CCover}[1]{\text{$\fap{\CC}{#1}$}}

\newcommand{\CLC}{\text{\sf C\hspace*{-0.5pt}L\hspace*{-0.5pt}C}}     
\newcommand{\CLCover}[1]{\text{$\fap{\CLC}{#1}$}}

\newcommand{\BBP}{\text{$\text{\sf BBP}$}}

\newcommand{\asys}{{\cal{S}}}
\newcommand{\asysi}{\iap{\asys}}

\newcommand{\aDeriv}{\mathcal{D}}
\newcommand{\aDerivi}{\iap{\mathcal{D}}}
\newcommand{\aDerivacc}{\aDeriv'}

\newcommand{\sLEE}{\text{\nf LEE}}

\newcommand{\LEE}{\sLEE}

\newcommand{\sLLEE}{\text{\nf LLEE}}

\newcommand{\LLEE}{\sLLEE}

\newcommand{\thplus}[2]{{#1}{+}{#2}}

\newcommand{\aseteqs}{\Gamma}
\newcommand{\bseteqs}{\Delta}

\newcommand{\picarrowstart}{\raisebox{2pt}{\begin{tikzpicture}%
                                             \draw[<-,very thick,>=latex,chocolate,shorten <=2pt](0,0) -- ++ (180:{12pt});%
                                           \end{tikzpicture}}}
  
\newcommand{\pictermvert}{\begin{tikzpicture}%
                           \node[draw,chocolate,very thick,circle,minimum width=2.5pt,fill,inner sep=0pt,outer sep=2pt](v){};%
                           \draw[thick,chocolate] (v) circle (0.12cm);%
                         \end{tikzpicture}}

\numberwithin{equation}{section}

\begin{document}

\maketitle

\begin{abstract}
  By adapting Salomaa's complete proof system for equality of regular expressions under the language semantics,
  Milner (1984) formulated a sound proof system for bisimilarity of regular expressions under the process interpretation he introduced.
  He asked whether this system is complete.
  Proof-theoretic arguments attempting to show completeness of this equational system are complicated
  by the presence of a non-algebraic rule for solving fixed-point equations by using star iteration.  
  
  We characterize the derivational power that the fixed-point rule adds to the purely equational part $\milnersysmin$ of Milner's system $\milnersys$: 
    it corresponds to 
  the power of coinductive proofs over $\milnersysmin$ that have the form of finite process graphs with the loop existence and elimination property LEE.
  We define a variant system $\coindmilnersys$ by replacing the fixed-point rule in $\milnersys$
    with a rule that permits 
                             LEE-shaped circular derivations 
                                                                  in~$\milnersysmin$ from previously derived equations
    as a premise.                                                                               
  With this rule alone we also define the variant system \CLC\ for 
                                                                   combining LEE-shaped coinductive proofs over~$\milnersysmin$. 
  We show that both $\coindmilnersys$ and $\CLC$ have proof interpretations in $\milnersys$, and vice versa.
  As this correspondence links, in both directions, derivability in $\milnersys$ 
    with derivation trees of 
                             process graphs,
    it widens the space for graph-based approaches to finding a completeness proof of~Milner's~system.
    
  This report is the extended version of a paper \cite{grab:2021:CALCO} with the same title presented at CALCO~2021.  
\end{abstract}


\section{Introduction}
  \label{intro}  
  
Milner \cite{miln:1984} (1984) defined a process semantics for regular expressions as process graphs:
the interpretation of $\stexpzero$ is deadlock, of $\stexpone$ is successful termination, letters $a$ are atomic actions,
the operators $\sstexpsum$ and $\sstexpprod$ stand for choice and concatenation of processes,
and (unary) Kleene star $\stexpit{(\cdot)}$ represents iteration with the option to terminate successfully 
  before each execution of the iteration body.
To disambiguate the use of regular expressions for denoting processes and comparing them via bisimilarity, Milner called them `star expressions'.
Using bisimilarity to identify processes with the same behavior, he was interested in an axiomatization of equality of `star behaviors',
which are bisimilarity equivalence classes of star-expression processes. 
He adapted Salomaa's complete proof system \cite{salo:1966} for language equivalence on regular expressions
to a system $\milnersys$ that is sound for equality of denoted star behaviors.
He left completeness as a question, because he recognized that Salomaa's proof~route~cannot~be~followed~directly. 


Specifically, Milner gave an example showing 
  that systems of guarded equations with star expressions cannot be solved by star expressions in general.
  Even if such a system is solvable, 
    the absence from~$\milnersys$ of 
    the \emph{left}-distributivity law  $x \cdot (y + z) = x \cdot y + x \cdot z$ in Salomaa's system (it is not sound under bisimilarity)
    frequently prevents applications of the fixed-point rule $\RSPstar$ in $\milnersys$ like in an extraction procedure from Salomaa's proof.
But if $\RSPstar$ is replaced in $\milnersys$ by a general unique-solvability rule scheme for guarded systems of equations 
  (see Def.~\ref{def:milnersys}), 
then a complete system arises (noted in \cite{grab:2006}).
Therefore completeness of $\milnersys$ hinges on whether the fixed-point rule $\RSPstar$ enables to prove equal
  any two star-expression solutions of a given guarded system of equations, 
  on the basis of the purely equational part~$\milnersysmin$~of~$\milnersys$.

As a stepping stone for tackling this difficult question, we here characterize 
  the derivational power that the fixed-point rule~$\RSPstar$ 
  adds to the subsystem $\milnersysmin$~of~$\milnersys$.
We do so by means of `coinductive proofs' whose shapes have  
the `loop existence and elimination property' \LEE\ from \cite{grab:fokk:2020a}. 
This property stems from the interpretation of (\onefree) star expressions,
  which is defined by induction on syntax trees, creating a hierarchy of `loop subgraphs'. 
Crucially for our purpose, guarded systems of equations that correspond to finite process graphs with \LEE\
  \emph{are} uniquely solvable modulo provability~in~$\milnersysmin$.
The reason is that process graphs with \LEE,
  which need not be in the image of the process interpretation,
are amenable to applying \emph{right}-distributivity and the 
                                                      rule \RSPstar\ for an extraction procedure like in Salomaa's proof (see Section~\ref{coindmilnersys:2:milnersys}).
These graphs can be expressed modulo bisimilarity by some star expression,
  which can be used to show that any two solutions modulo $\milnersysmin$ of a specification of \LEEshape\ are \provablyin{\milnersys} equal. 
This is a crucial step 
                       in the completeness proof by Fokkink and myself in \cite{grab:fokk:2020a} 
  for the tailored restriction \BBP\ of Milner's system $\milnersys$ to~`\onefree'~star~expressions.

Thus motivated, we define a `\LLEEwitnessed\ coinductive proof' as
  a process graph $\mathcal{G}$ with `layered' \LEE\ (\LLEE) whose vertices are labeled by equations between star expressions.
  The left- and the right-hand sides of the equations have to form a solution vector of a specification corresponding to the process graph $\mathcal{G}$.
  However, that specification needs to be satisfied only up to provability in $\milnersysmin$ from sound assumptions. 
Such coinductive derivations are typically circular,  
like the one below of the semantically valid equation $(a + b)^* \cdot 0 \formeq (a \cdot (a + b) + b)^* \cdot 0\,$:
%
\begin{flushleft}\vspace*{-2ex}\label{inf:ex:1:coindproof}
  \hspace*{-0.5ex}%
  \begin{tikzpicture}
 
\matrix[anchor=center,row sep=1.25cm,column sep=2.65cm,
        every node/.style={draw=none}
        ] at (0,0) {
    \node(v1){}; &             & \node(v2){};
    \\
                 & \node(v){}; 
    \\
  };
  
\path (v1) node(v1-label){$ (1 \cdot g^*) \cdot 0
                              \formeq
                            ((1 \cdot (a + b)) \cdot h^*)  \cdot 0 $};       
\draw[->,shorten <=2.6cm,shorten >=1.75cm] (v1) to node[above,pos=0.5725]{$a, b$} (v2);

\path (v2) node(v2-label){$ (1 \cdot g^*) \cdot 0
                              \formeq
                             (1 \cdot h^*) \cdot 0 $}; 
\draw[->,thick,densely dotted,out=-45,in=0,distance=0.75cm,shorten <= 0.2cm] (v2) to node[right]{$\sone$} ($(v) + (2.775cm,0.3cm)$);

\path (v) node(v-label){$ \underbrace{(a + b)^*}_{g^*} \,\cdot\, 0
                            \formeq
                          \underbrace{(a \cdot (a + b) + b)^*}_{h^*} \,\cdot\, 0 $}; 
\draw[->,thick,darkcyan,shorten <= 1.1cm,shorten >= 0.5cm] (v) to node[above,pos=0.45]{$\loopnsteplab{1}$}
                                                                  node[below,pos=0.7]{$\black{a}$} (v1);
\draw[->,thick,darkcyan,shorten <= 1.1cm,shorten >= 0.5cm] (v) to node[above,pos=0.45]{$\loopnsteplab{1}$}
                                                                  node[below,pos=0.7]{$\black{b}$} (v2);

\matrix[anchor=center,row sep=1.25cm,column sep=1.1cm,
        every node/.style={draw,very thick,circle,minimum width=2.5pt,fill,inner sep=0pt,outer sep=2pt}
        ] at (6.5,0) {
    \node(v1--2){}; &             & \node(v2--2){};
    \\
                 & \node(v--2){}; 
    \\
  };  
  
\draw[->,thick,darkcyan] (v--2) to node[right,pos=0.7]{$\loopnsteplab{1}$}
                                   node[left,pos=0.45]{$\black{a}$} (v1--2);
\draw[->,thick,darkcyan] (v--2) to node[left,pos=0.7]{$\loopnsteplab{1}$}
                                   node[right,pos=0.45]{$\black{b}$} (v2--2);  

\draw[->] (v1--2) to node[above]{$a, b$} (v2--2);

\draw[<-,very thick,>=latex,chocolate,shorten <= 2pt](v--2) -- ++ (180:0.55cm); 
\draw[->,thick,densely dotted,out=-45,in=0,distance=0.75cm] 
  (v2--2) to node[right]{$\sone$} (v--2);

\path (v--2) ++ (-1.35cm,0.2cm) node(v--2-label) {\Large $\mathcal{G},\,\Hat{\mathcal{G}}$};

\end{tikzpicture}  \vspace*{-3ex}
\end{flushleft}
The process graph $\mathcal{G}$, 
  which is given together with a labeling $\Hat{\mathcal{G}}$ that is a `\LLEEwitness'~of~$\mathcal{G}$ 
  (the \darkcyan{colored} transitions with marking labels $\loopsteplab{n}$, for $n\in\natplus$,
   indicate LLEE-structure, see Section~\ref{LEE}),
underlies the coinductive proof on the left
(see Ex.~\ref{ex:inf:ex:1:coindproof} in the Appendix 
                                      for a justification). 
$\mathcal{G}$ is a `\onechart' that~is, 
  a process graph with \onetransitions\ that represent empty step processes.
We depict \onetransitions\ as dotted arrows.
For \onecharts, `\onebisimulation' is the adequate concept of bisimulation.
We showed in \cite{grab:2021:TERMGRAPH-postproceedings,grab:2020:scpgs-arxiv} 
  that the process (chart) interpretation $\chartof{\astexp}$ of a star expression $\astexp$ is
  the image of a \onechart\ $\onechartof{\astexp}$ with \LLEE\ under a functional \onebisimulation.
In this example, $\mathcal{G} = \onechartof{h^* \cdot 0}$ maps by a functional \onebisimulation\
  to interpretations of both expressions in the conclusion. 
The correctness conditions for such coinductive proofs are formed by the requirement that the left-, and respectively, the right-hand sides of formal equations
form `\provablein{\milnersysmin} solutions' of the underlying process graph: an expression at a vertex~$\avert$ 
can be reconstructed, provably in $\milnersysmin$, from the transitions to, and the expressions at, immediate successor vertices of $\avert$.
Crucially we establish in Section~\ref{coindmilnersys:2:milnersys},
by a generalization of arguments in \cite{grab:fokk:2020a,grab:fokk:2020b} using $\RSPstar$,
  that every \LLEEwitnessed\ coinductive proof over $\milnersysmin$ can be transformed into a derivation~in~$\milnersys$~with~the~same~conclusion.

This raises the question of whether the fixed-point rule \RSPstar\ of $\milnersys$ adds any derivational power to $\milnersysmin$
  that goes beyond those of \LLEEwitnessed\ coinductive proofs over $\milnersysmin$, and if so, how far precisely.
As our main result
we show in Section~\ref{milnersys:2:coindmilnersys} that every instance of the fixed-point rule \RSPstar\
  can be mimicked by a \LLEEwitnessed\ coinductive proof over $\milnersysmin$ in which 
  also the premise of the rule may be used.%
  \begin{figure}[tb!]
  \begin{center}
    \AxiomC{$ \overbrace{(a + b)^*}^{\chocolate{{e _0 ^*}}} \chocolate{\mathrel{\cdot} 0} 
                \:\formeq\:
              \overbrace{(a \cdot (a + b) + b)}^{\alert{f}} \cdot (\overbrace{(a + b)^*}^{\chocolate{{e _0 ^*}}} \chocolate{{} \cdot 0}) + \forestgreen{0} $}
    \RightLabel{$\ainst$, $\RSPstar$}
    \UnaryInfC{$ (a + b)^* \cdot 0
                   \:\formeq\:
                 (a \cdot (a + b) + b)^* \,\cdot\, \forestgreen{0} $}
    \DisplayProof
    %
    %
    \AxiomC{$ \chocolate{{e _0 ^*} \cdot 0} \:\formeq\: \alert{f} \cdot (\chocolate{{e _0 ^*} \cdot 0}) + \forestgreen{0} \rule{0pt}{23.5pt}$}
    \UnaryInfC{$ \chocolate{{e _0 ^*} \cdot 0} \:\formeq\: \alert{f^*} \,\cdot\, \forestgreen{0}$}
    \DisplayProof 
  \end{center}
  \vspace*{-4.5ex}
  \begin{center}  
    \begin{tikzpicture}
      \matrix[anchor=center,row sep=1.25cm,column sep=3.5cm,
              every node/.style={draw=none}
              ] at (0,0) {
          \node(v1){}; &[1.5cm]             &[-1.5cm] \node(v2){};
          \\
                       & \node(v){}; 
          \\
        };
        
      \path (v1) node(v1-label){$ (1 \cdot (a + b)) \cdot (\chocolate{{e _0 ^*} \cdot 0}) 
                                    \formeq
                                  ((1 \cdot (a + b)) \cdot \alert{f^*})  \cdot \forestgreen{0} $};       
      \draw[->,shorten <=3.35cm,shorten >=1.75cm] (v1) to node[above,pos=0.6]{$a, b$} (v2);

      \path (v2) node(v2-label){$ 1 \cdot (\chocolate{{e _0 ^*} \cdot 0})
                                    \formeq
                                   (1 \cdot \alert{f^*}) \cdot \forestgreen{0} $}; 
      \draw[->,thick,densely dotted,out=-45,in=0,distance=0.75cm,shorten <= 0.2cm] (v2) to node[right]{$\sone$} ($(v) + (2.35cm,0.125cm)$);

      \path (v) node[xshift=-3.5cm,yshift=-0.15cm](v-label)
                             {$ \underbrace{\chocolate{{e _0 ^*} \cdot 0}}
                                           _{\text{(by the premise of $\ainst$)} \;\;
                                               (a \cdot (a + b) + b) \cdot (\chocolate{{e _0 ^*} \cdot 0}) + \forestgreen{0}
                                                 \:\synteq\:
                                               \alert{f} \cdot (\chocolate{{e _0 ^*} \cdot 0}) + \forestgreen{0} 
                                                 \: = \:
                                             \phantom{
                                                 \: = \:
                                               \alert{f} \cdot (\chocolate{{e _0 ^*} \cdot 0}) + \forestgreen{0}   
                                                 \:\synteq\: 
                                               (a \cdot (a + b) + b) \cdot (\chocolate{{e _0 ^*} \cdot 0}) + \forestgreen{0} 
                                               \text{(by the premise of $\ainst$)} }}
                                    \hspace*{-45ex}
                                  \formeq
                                \underbrace{(a \cdot (a + b) + b)^*}_{\alert{f^*}} \,\cdot\, \forestgreen{0} $}; 
      \draw[->,thick,darkcyan,shorten <= 1.1cm,shorten >= 1cm] (v) to node[above,pos=0.45]{$\loopnsteplab{1}$}
                                                                        node[below,pos=0.525]{$\black{a}$} (v1);
      \draw[->,thick,darkcyan,shorten <= 0.6cm,shorten >= 0.5cm] (v) to node[above,pos=0.45]{$\loopnsteplab{1}$}
                                                                        node[below,pos=0.575]{$\black{b}$} (v2); 
                                                                        
      \path (v) ++ (-5.5cm,0.5cm) node{\Large $ {\onecharthatof{\rule{0pt}{8pt}\alert{f^*}\cdot \forestgreen{0}}} $};                                                                     
    \end{tikzpicture}
  \end{center}
  \vspace*{-4ex}
\caption{\protect\label{fig:ex:1:RSPstar:to:coindproof}%
         Mimicking an instance $\protect\ainst$ of the fixed-point rule \protect\RSPstar\ (above) in Milner's system $\protect\milnersys = \protect\thplus{\protect\milnersysmin}{\protect\RSPstar}$
           by a coinductive proof (below) over $\protect\thplus{\protect\milnersysmin}{\protect\setexp{\protect\text{premise of $\protect\ainst$}}}$ with \protect\LLEEwitness~$\protect\onecharthatof{\protect\alert{f^*}\protect\cdot \protect\forestgreen{0}}$.} 
\end{figure}%
It follows that the derivational power that \RSPstar\ adds to $\milnersysmin$ within $\milnersys$ consists of iterating such \LLEEwitnessed\ coinductive proofs
along finite (meta-)prooftrees.    
The example in Fig.~\ref{fig:ex:1:RSPstar:to:coindproof} (see Ex.~\ref{ex:fig:ex:1:RSPstar:to:coindproof} in the Appendix for a justification)
  can give a first impression 
  of the construction that we will use (in the proof of Lem.~\ref{lem:lem:mimic:RSPstar})
  to mimic instances of \RSPstar.
  Here this construction results in a coinductive proof that only differs slightly from the one with the same underlying \LLEEonechart~we~saw~earlier.


Based on these two proof transformations we obtain a \theoremequivalent, coinductive variant $\coindmilnersys$ of $\milnersys$
by replacing $\RSPstar$ 
with a rule that as one premise permits a \LLEEwitnessed\ coinductive proof over $\milnersysmin$
plus the equations of other premises. 
We also define a \theoremequivalent\ system \CLC\ (`combining \LLEEwitnessed\ coinductive proofs') with this rule alone. 
While \CLC\ only has LEE-shaped coinductive proofs over $\milnersysmin$ as formulas,
we use a hybrid concept of formula in $\coindmilnersys$ that also permits equations between star expressions.

  
Additionally, we formulate proof systems $\coindmilnersysbar$ and \CC\
  that arise from $\coindmilnersys$ and \CLC\ by dropping `\LLEE-witnessed' as a requirement for coinductive proofs.
These systems are (obviously) complete for bisimilarity of process interpretations,
  because they can mimic the unique solvability rule scheme for guarded systems of specifications mentioned before.     
  
Our transformations are inspired by proof-theoretic interpretations in \cite{grab:2005:prfschr} between
proof systems for recursive type equality by Amadio and Cardelli~\cite{amad:card:1993},
  and by Brandt and Henglein~\cite{bran:heng:1998}. 
The transformation from $\coindmilnersys$ back to $\milnersys$ is similar in kind to
  one we described in \cite{grab:2005:calco} 
  from derivations in a coinductively motivated proof system for language equivalence between regular expressions
  to derivations in Salomaa's system \cite{salo:1966} with a fixed-point rule similar~to~$\RSPstar$.    
      
\medskip 

\section{Process semantics for star expressions, and Milner's proof system}%
  \label{prelims}

Here we fix terminology concerning star expressions, \onecharts, \onebisimulations,
  we exhibit Milner's system (and a few variants), and recall the chart interpretation of star expressions.

\smallskip

Let $\actions$ be a set 
                        of \emph{actions}.
The set $\StExpover{\actions}$ of \emph{star expressions over actions in $\actions$} 
are strings that are defined by the following grammar: 
\begin{center}
  $
  \astexp, \astexpi{1}, \astexpi{2}
    \:\BNFdefdby\:
  \stexpzero
    \BNFor
  \stexpone
    \BNFor  
  \aact
    \BNFor
  ( \stexpsum{\astexpi{1}}{\astexpi{2}} )
    \BNFor
  ( \stexpprod{\astexpi{1}}{\astexpi{2}} )
    \BNFor
  ( \stexpit{\astexp} ) \qquad\quad \text{(where $\aact\in\actions$)}
  $
\end{center} 
We will drop outermost brackets.
We use $\astexp,\bstexp,\cstexp,\dstexp$, possibly indexed and/or decorated, as syntactical variables for star expressions.
We write $\ssynteq$ for syntactic equality between star expressions denoted by such syntactical variables,
  and values of star expression functions, in a given context,
but we permit $\sformeq$ in formal equations between star expressions.
%
We denote by $\StExpEqover{\actions}$ the set of formal equations $\astexp \formeq \bstexp$ 
between two star expressions $\astexp,\bstexp\in\StExpover{\actions}$. 

We define sum expressions $\sum_{i=1}^{n} \astexpi{i}$ inductively
  as $\stexpzero$ if $n=0$,
  as $\astexpi{1}$ if $n=1$,
  and as $\stexpsum{(\sum_{i=1}^{n-1} \astexpi{i})}{\astexpi{n}}$ if $n>0$,
  for $n\in\nat = \setexp{0,1,2,\ldots}$.
The \emph{(syntactic) star height} $\sth{\astexp}$ of a star expression $\astexp\in\StExpover{\actions}$
is the maximal nesting depth of stars in $\astexp$, defined inductively 
by: $\sth{\stexpzero} \defdby \sth{\stexpone}  \defdby \sth{\aact} \defdby 0$, 
     $\sth{\stexpsum{\astexpi{1}}{\astexpi{2}}} \defdby \sth{\stexpprod{\astexpi{1}}{\astexpi{2}}}
                                                \defdby \max\setexp{\sth{\astexpi{1}}, \sth{\astexpi{2}}}$, 
     and $\sth{\stexpit{\astexp}} \defdby 1 + \sth{\astexp}$.

\smallskip

A \emph{\onechart}
is a 6\nb-tuple $\tuple{\verts,\actions,\sone,\start,\transs,\sterminates}$ 
  where $\verts$ is a \ul{finite} set of \emph{vertices},
  $\actions$ is a set of \emph{(proper)} $\emph{action labels}$,
  $\sone\notin\actions$ is the specified \emph{empty step label},
  $\start\in\verts$ is the \emph{start vertex} (hence $\verts \neq \emptyset$),
  $\transs \subseteq \verts\times\oneactions\times\verts$ is the \emph{labeled transition relation},
  where $\oneactions \defdby \actions \cup \setexp{\sone}$ is the set of action labels including $\sone$, 
  and $\sterminates \subseteq \verts$ is a set of \emph{vertices with immediate termination}. 
In such a \onechart, 
we call a transition in $\transs\cap(\verts\times\actions\times\verts)$ (labeled by a \emph{proper action} in $\actions$) 
        a \emph{proper transition},
and a transition in $\transs\cap(\verts\times\setexp{\sone}\times\verts)$ (labeled by the \emph{empty-step symbol}~$\sone$)
        a \emph{\onetransition}.  
Reserving non-underlined action labels like $\aact,\bact,\ldots$ for proper actions,
we use underlined action label symbols like $\aoneact$ for actions labels in the set $\oneactions$ 
that includes~the~label~$\sone$. 
We highlight in red transition labels that may involve $\sone$. 

We say that a \onechart\ is \emph{weakly guarded} if it does not contain cycles of \onetransitions.
By a \emph{chart} we mean a \onechart\ that is \emph{\sonefree} in the sense that it does not contain \onetransitions.

\smallskip

Below we define the process semantics of regular (star) expressions as (\sonefree) charts, 
  and hence as finite, rooted labeled transition systems,
  which will be compared with ($\sone$-)bisimilarity. 
The charts obtained correspond to non-deterministic finite-state automata that are obtained
  by iterating partial~derivatives \cite{anti:1996} of Antimirov (who did not aim at a process semantics).  
  
\begin{defi}
    \label{def:StExpTSS}\label{def:chartof}
  The \emph{chart interpretation of} a star expression~$\astexp\in\StExpover{\actions}$
  is the $\sone$\nb-tran\-si\-tion free chart $\chartof{\astexp} = \tuple{\vertsof{\astexp},\actions,\sone,\astexp,\transs\cap\vertsof{\astexp},\termexts\cap\vertsof{\astexp}}$,
  where $\vertsof{\astexp}$ consists of all star expressions that are reachable from $\astexp$
  via the labeled transition relation $ \transs \subseteq \StExpover{\actions}\times\actions\times\StExpsover{\actions}$
  that is defined, together with the \immediatetermination\ relation $\sterminates \subseteq \StExpsover{\actions}$,
  via derivability in the transition system specification (TSS)~$\StExpTSSover{\actions}$,
  for $\aact\in\actions$, $\astexp,\astexpi{1},\astexpi{2},\astexpacc,\astexpacci{1},\astexpacci{2}\in\StExpover{\actions}$:
  \begin{center}
    $
    \begin{aligned}
      &
      \AxiomC{\phantom{$\terminates{\stexpone}$}}
      \UnaryInfC{$\terminates{\stexpone}$}
      \DisplayProof
      & \hspace*{0.5ex} &
      \AxiomC{$ \terminates{\astexpi{1}} $}
      \UnaryInfC{$ \terminates{(\stexpsum{\astexpi{1}}{\astexpi{2}})} $}
      & \hspace*{4ex} &
      \AxiomC{$ \terminates{\astexpi{i}} $}
      \UnaryInfC{$ \terminates{(\stexpsum{\astexpi{1}}{\astexpi{2}})} $}
      \DisplayProof
      & \hspace*{4ex} &
      \AxiomC{$\terminates{\astexpi{1}}$}
      \AxiomC{$\terminates{\astexpi{2}}$}
      \BinaryInfC{$\terminates{(\stexpprod{\astexpi{1}}{\astexpi{2}})}$}
      \DisplayProof
      & \hspace*{4ex} &
      \AxiomC{$\phantom{\terminates{\stexpit{\astexp}}}$}
      \UnaryInfC{$\terminates{(\stexpit{\astexp})}$}
      \DisplayProof
    \end{aligned} 
    $
    \\[1ex]
    $
    \begin{aligned}
      &  
      \AxiomC{$\phantom{a \:\lt{a}\: \stexpone \rule{0pt}{12.4pt}}$}
      \UnaryInfC{$a \:\lt{a}\: \stexpone$}
      \DisplayProof
      & & \:
      \AxiomC{$ \astexpi{i} \:\lt{a}\: \astexpacci{i} $}
      \UnaryInfC{$ \stexpsum{\astexpi{1}}{\astexpi{2}} \:\lt{a}\: \astexpacci{i} $}
      \DisplayProof 
      & & \:
      \AxiomC{$ \astexpi{1} \:\lt{a}\: \astexpacci{1} $}
      \UnaryInfC{$ \stexpprod{\astexpi{1}}{\astexpi{2}} \:\lt{a}\: \stexpprod{\astexpacci{1}}{\astexpi{2}} $}
      \DisplayProof
      & &
      \AxiomC{$\terminates{\astexpi{1}}$}
      \AxiomC{$ \astexpi{2} \:\lt{a}\: \astexpacci{2} $}
      \BinaryInfC{$ \stexpprod{\astexpi{1}}{\astexpi{2}} \:\lt{a}\: \astexpacci{2} $}
      \DisplayProof
      & & \:
      \AxiomC{$\astexp \:\lt{a}\: \astexpacc \rule{0pt}{11.5pt}$}
      \UnaryInfC{$\stexpit{\astexp} \:\lt{a}\: \stexpprod{\astexpacc}{\stexpit{\astexp}}$}
      \DisplayProof
    \end{aligned}
    $
  \end{center}
  If $\astexp \lt{\aact} \astexpacc$ is derivable in $\StExpTSSover{\actions}$, for $\astexp,\astexpacc\in\StExpover{\actions}$,
  $\aact\in\actions$, then we say that $\astexpacc$ is a \emph{derivative} of $\astexp$.
  If $\terminates{\astexp}$ is derivable in $\StExpTSSover{\actions}$, 
  then we say that $\astexp$ \emph{permits immediate termination}.
\end{defi}

In Section~\ref{LEE} we define a refinement of this interpretation from \cite{grab:2021:TERMGRAPH-postproceedings} into a \onechart\ interpretation.
In both versions, ($\sone$-)charts obtained will be compared with respect to \onebisimilarity\
that relates the behavior of `induced transitions' of \onecharts.
By an \emph{induced \transitionact{\aact}} $\avert \ilt{\aact} \bvert$, for a proper action $\aact\in\actions$, in a \onechart~$\aonechart$ we mean 
  a path $\avert \lt{\sone} \cdots \lt{\sone} \cdot \lt{\aact} \bvert$ in $\aonechart$
  that consists of a finite number of \onetransitions\ that ends with a proper \transitionact{\aact}.
By \emph{induced termination} $\oneterminates{\avert}$, for $\avert\in\verts$ we mean that there is a path
  $\avert \lt{\sone} \cdots \lt{\sone} \averttilde$ with $\terminates{\averttilde}$~in~$\aonechart$.      

\begin{defi}[$\sone$-bisimulation]\label{def:onebisim}
  Let $\aonecharti{i} = \tuple{\vertsi{i},\actions,\sone,\starti{i},\sredi{i},\termextsi{i}}$
      be \onecharts, for $i\in\setexp{1,2}$. 
  
  By a \emph{\onebisimulation\ between $\aonecharti{1}$ and $\aonecharti{2}$}
  we mean a binary relation $\abisim \subseteq \vertsi{1}\times\vertsi{2}$ 
  such that $\pair{\starti{1}}{\starti{2}}\in\abisim$, 
  and for every $\pair{\averti{1}}{\averti{2}}\in\abisim$ the following three conditions hold:
  \begin{description}
    \item{(forth)} \quad
      $ \forall \avertacci{1}\in\vertsi{1}
          \forall \aact\in\actions
              \bigl(\,
                \averti{1} \ilti{\aact}{1} \avertacci{1}
                  \;\;\Longrightarrow\;\;
                    \exists \avertacci{2}\in\vertsi{2}
                      \bigl(\, \averti{2} \ilti{\aact}{2} \avertacci{2} 
                                 \logand
                               \pair{\avertacci{1}}{\avertacci{2}}\in\abisim \,)
            \,\bigr) \punc{,} $
      
    \item{(back)} \quad
      $ \forall \avertacci{2}\in\vertsi{2}
          \forall \aact\in\actions
            \bigr(\,
              \exists \avertacci{1}\in\vertsi{1}
                \bigl(\, \averti{1} \ilti{\aact}{1} \avertacci{1} 
                           \logand
                         \pair{\avertacci{1}}{\avertacci{2}}\in\abisim \,)
                  \;\;\Longleftarrow\;\;
                \averti{2} \ilti{\aact}{2} \avertacci{2}
              \,\bigr) \punc{,} $
      
    \item{(termination)} \quad
      $ \oneterminatesi{1}{\averti{1}}
          \;\;\Longleftrightarrow\;\;
            \oneterminatesi{2}{\averti{2}} \punc{.}$
  \end{description}
  \indent
  We denote by $\aonecharti{1} \onebisim \aonecharti{2}$, 
    and say that $\aonecharti{1}$ and $\aonecharti{2}$ are \emph{\onebisimilar},
  if there is a \onebisimulation\ between $\aonecharti{1}$ and $\aonecharti{2}$.
  We call \onebisimilar\ (\sonefree) charts $\acharti{1}$ and $\acharti{2}$ \emph{bisimilar},
     and write $\acharti{1} \bisim \acharti{2}$.
\end{defi}

Let $\actions$ be a set.
The \emph{basic proof system~$\eqlogicover{\actions}$ of equational logic}
  for star expressions
has as \emph{formulas} the formal equations between star expressions in $\StExpEqover{\actions}$,
and the following \emph{rules}$\,$:
%
%
\begin{gather*}\renewcommand{\fCenter}{\formeq}
\begin{alignedat}{3}
  &
\begin{aligned}[c]
  \AxiomC{$\phantom{\astexp   \fCenter   \astexp\rule{0pt}{7.5pt}}$}
  \RightLabel{\REFL}
  \UnaryInfC{$\astexp   \fCenter   \astexp\rule{0pt}{7.5pt}$} 
  \DisplayProof
\end{aligned}  
  & \qquad & 
\begin{aligned}[c]
  \Axiom$\astexp   \fCenter   \bstexp$
  \RightLabel{\SYMM}
  \UnaryInf$\bstexp   \fCenter   \astexp$
  \DisplayProof
\end{aligned}  
  & \qquad & 
\begin{aligned}[c]
  \AxiomC{$\astexp \formeq \bstexp$}
  \AxiomC{$\bstexp \formeq \cstexp$}
  \RightLabel{\TRANS}
  \BinaryInfC{$\astexp \formeq \cstexp\rule{0pt}{7.5pt} $}
  \DisplayProof
\end{aligned}  
  & \qquad
\begin{aligned}
  \Axiom$\astexp   \fCenter   \bstexp$
  \RightLabel{\CXT}
  \UnaryInf$\acxtap{\astexp}   \fCenter   \acxtap{\bstexp}$
  \DisplayProof
\end{aligned}
\end{alignedat}
\end{gather*}%
%
%
that is, the rules~$\REFL$ (for reflexivity), and the rules~$\SYMM$ (for symmetry), $\TRANS$ (for transitivity), and $\CXT$ (for filling a context),
where $\acxtap{\iets}$ is a 1-hole star expression context. 

By an \emph{\eqlogicbased\ system over $\StExpover{\actions}$} (and \emph{for star expressions over $\actions$}) 
we mean a proof system whose formulas are the formal equations in $\StExpEqover{\actions}$,
and whose rules include the rules of the basic system $\eqlogicover{\actions}$ of equational logic
(additionally, it may specify an arbitrary set of axioms). 
We will use $\asys$ as syntactical variable for \eqlogicbased\ proof systems. 

  
Let $\asys$ be an \eqlogicbased\ proof system over $\StExpover{\actions}$, and $\astexpi{1},\astexpi{2}\in\StExpover{\actions}$. 
We permit to write $\astexpi{1} \eqin{\asys} \astexpi{2}$ for $\derivablein{\asys} \astexpi{1} \formeq \astexpi{2}$,
  that is for the statement that there is a derivation without assumptions in $\asys$ that has conclusion $\astexpi{1} \formeq \astexpi{2}$.

\begin{defi}[sub-system, theorem equivalence, and theorem subsumption of proof systems]
  Let $\asysi{1}$ and $\asysi{2}$ be \eqlogicbased\ proof systems over $\StExpover{\actions}$. 
  We say that $\asysi{1}$ is a \emph{sub-system} of $\asysi{2}$, denoted by $\asysi{1} \subsystem \asysi{2}$,
    if every axiom of $\asysi{1}$ is an axiom of $\asysi{2}$, and every rule of $\asysi{1}$ is also a rule of $\asysi{2}$. 
  We say that \emph{$\asysi{1}$ is \theoremsubsumed\ by $\asysi{2}$}, denoted by $\asysi{1} \isthmsubsumedby \asysi{2}$,
    if whenever a formal equation $\astexpi{1} \formeq \astexpi{2}$ is derivable in $\asysi{1}$ 
         (without assumptions, by using only the rules and axioms of $\asysi{1}$),
       then $\astexpi{1} \formeq \astexpi{2}$ is also derivable in $\asysi{2}$.
  We say that \emph{$\asysi{1}$ and $\asysi{2}$ are \theoremequivalent}, denoted by $\asysi{1} \thmequiv \asysi{2}$,
    if they have the same derivable equations.
\end{defi}

\begin{defi}[Milner's system \milnersys, variants and subsystems]
            \label{def:milnersys}
  Let $\actions$ be a set of actions.
  
  By the proof system $\milnersysminover{\actions}$
    we mean  
  the \eqlogicbased\ proof system for star expressions over $\actions$
  with the following \emph{axiom schemes}:
  \begin{alignat*}{4}
    (\assocstexpsum) \quad & &
        \stexpsum{(\stexpsum{\astexp}{\bstexp})}{\cstexp}
          & \formeq
        \stexpsum{\astexp}{(\stexpsum{\bstexp}{\cstexp})}
    & \qquad
    (\leftidstexpprod) \quad & &
        \stexpprod{\stexpone}{\astexp}  
          & \formeq \astexp
    \displaybreak[0]\\ 
    (\neutralstexpsum) \quad & &
        \stexpsum{\astexp}{\stexpzero}
          & \formeq
        \astexp
    & \qquad
    (\rightidstexpprod) \quad & &   
        \stexpprod{\astexp}{\stexpone}
          & \formeq \astexp
    \displaybreak[0]\\ 
    (\commstexpsum) \quad & &
        \stexpsum{\astexp}{\bstexp} 
          & \formeq
        \stexpsum{\bstexp}{\astexp}
    & \qquad
    (\stexpzerostexpprod) \quad & &  
        \stexpprod{\stexpzero}{\astexp}
          & \formeq
        \stexpzero
    \displaybreak[0]\\ 
    (\idempotstexpsum) \quad & &
        \stexpsum{\astexp}{\astexp}
          & \formeq
        \astexp 
    & \qquad
    (\recdefstexpit) \quad & &
        \stexpit{\astexp}
          & \formeq \stexpsum{\stexpone}{\stexpprod{\astexp}{\stexpit{\astexp}}}
    \displaybreak[0]\\ 
    (\assocstexpprod) \quad & &
        \stexpprod{(\stexpprod{\astexp}{\bstexp})}{\cstexp}
          & \formeq
        \stexpprod{\astexp}{(\stexpprod{\bstexp}{\cstexp})}   
    & \qquad\quad
    (\termstexpit) \quad & & 
        \stexpit{\astexp}
          & \formeq \stexpit{(\stexpsum{\stexpone}{\astexp})}
    \displaybreak[0]\\ 
    (\rdistr) \quad & & 
        \stexpprod{(\stexpsum{\astexp}{\bstexp})}{\cstexp}
          & \formeq
        \stexpsum{\stexpprod{\astexp}{\cstexp}}{\stexpprod{\bstexp}{\cstexp}}
      %
  \end{alignat*}
  where $\astexp,\bstexp,\cstexp\in\StExpover{\actions}$, 
  and with the \emph{rules} of the system $\eqlogicover{\actions}$ of equational logic.

  The \emph{recursive specification principle for star iteration} $\RSPstar$,
  the \emph{unique solvability principle for star iteration} $\USPstar$,
  and the \emph{general unique solvability principle} $\USP$ are the rules:
  \begin{center}
    $
    \begin{aligned}\renewcommand{\fCenter}{\formeq}
      \Axiom$ \astexp  \fCenter  \stexpsum{\stexpprod{\bstexp}{\astexp}}{\cstexp} $
      \RightLabel{\RSPstar\ $\,${\small (if $\notterminates{\bstexp}$)}}
      \UnaryInf$ \astexp   \fCenter   \stexpprod{\stexpit{\bstexp}}{\cstexp} $
      \DisplayProof
      & \quad & 
      \AxiomC{$ \astexpi{1}   \formeq   \stexpsum{\stexpprod{\stexpit{\bstexp}}{\astexpi{1}}}
                                       {\cstexp} $}
      \AxiomC{$ \astexpi{2}   \formeq   \stexpsum{\stexpprod{\stexpit{\bstexp}}{\astexpi{2}}}
                                       {\cstexp} $}
      \RightLabel{\USPstar\ $\,${\small (if $\notterminates{\bstexp}$)}}
      \BinaryInfC{$ \astexpi{1}   \formeq   \astexpi{2} \rule{0pt}{7pt} $}                                 
      \DisplayProof
    \end{aligned}
    $
    \\[0.75ex]
    $
    \begin{aligned}  
      \AxiomC{$ \Bigl\{\,
                  \astexpi{i,1}   \formeq    \bigl(\, \sum_{j=1}^{n_i} \bstexpi{i,j} \prod \astexpi{j,1} \,\bigr) + \cstexpi{i} \,\bigr)
                    \quad
                  \astexpi{i,2}   \formeq    \bigl(\, \sum_{j=1}^{n_i} \bstexpi{i,j} \prod \astexpi{j,2} \,\bigr) + \cstexpi{i} \,\bigr) 
                \,\Bigr\}_{i=1,\ldots,n} $}
      \RightLabel{\USP\ $\,$ \parbox{\widthof{for all $i,j$)}}
                                    {\small (if $\notterminates{\bstexpi{i,j}}$
                                     \\\phantom{(}%
                                     for all $i,j$)}}
      \UnaryInfC{$ \astexpi{1,1}   \formeq    \astexpi{1,2} \rule{0pt}{7.5pt} $}
      \DisplayProof
    \end{aligned}
    $
  \end{center}
  
  \emph{Milner's proof system} $\milnersysover{\actions}$ is the extension of $\milnersysminover{\actions}$ 
  by adding the rule $\RSPstar$. Its variant systems $\milnersysaccover{\actions}$, and $\milnersysaccbarover{\actions}$,
  arise from $\milnersysminover{\actions}$ by adding (instead of $\RSPstar$) the rule $\USPstar$, and respectively, the rule $\USP$.
  $\ACIover{\actions}$ is the system with the axioms for \ul{a}ssociativity, \ul{c}ommutativity, and \ul{i}dempotency for $+$. 
  We will keep the action set $\actions$ implicit in the notation. 
  %
\end{defi}\enlargethispage{10ex}

\begin{prop}[Milner, \cite{miln:1984}]\label{prop:milnersys:sound}
  $\milnersys$ is sound for bisimilarity of chart interpretations.
  That is, for all $\astexp,\bstexp\in\StExpover{\actions}$ it holds: 
  $(\,
     \astexp \milnersyseq \bstexp
       \;\;\Longrightarrow\;\;
     \chartof{\astexp} \bisim \chartof{\bstexp}
   \,) \punc{.}$
\end{prop}

\begin{question}[Milner, \cite{miln:1984}]
  Is $\milnersys$ also complete for bisimilarity of process interpretations?
  That is, does for all $\astexp,\bstexp\in\StExpover{\actions}$ the implication
  $(\,
     \astexp \milnersyseq \bstexp
       \;\;\Longleftarrow\;\;
     \chartof{\astexp} \bisim \chartof{\bstexp}
   \,)$ hold?
\end{question}


\begin{defi}[provable solutions]\label{def:provable:solution}
  Let $\asys$ be an \eqlogicbased\ proof system for star expressions over $\actions$
    that extends \ACI. 
  Let $\aonechart = \tuple{\verts,\actions,\sone,\start,\transs,\exts}$ be a \onechart.\pagebreak[4]
  
  By a \emph{star expression function on $\aonechart$} we mean a function $\sasol \funin \verts \to \StExpover{\actions}$
    on the vertices of $\aonechart$. 
  Let $\avert\in\verts$.   
  We say that such a star expression function $\sasol$ on $\aonechart$ 
    is an \emph{\provablein{\asys} solution of $\aonechart$ at $\avert$} if it holds that 
    $
    \asol{\avert}
      \,\eqin{\asys}\,
        \stexpsum{\terminatesconstof{\aonechart}{\avert}}
                 {\sum_{i=1}^{n} \stexpprod{\aoneacti{i}}{\asol{\averti{i}}}} 
          \punc{,}
    $
    given the (possibly redundant) list representation
    $\transitionsinfrom{\aonechart}{\avert}
       =
     \descsetexpbig{ \avert \lt{\aoneacti{i}} \averti{i} }{ i \in\setexp{1,\ldots,n} }$,
    of transitions from $\avert$ in~$\aonechart$ 
    and where $\terminatesconstof{\aonechart}{\avert}$
    is the \emph{termination constant $\terminatesconstof{\aonechart}{\avert}$ of $\aonechart$ at $\avert$}
      defined as $\stexpzero$ if $\notterminates{\avert}$,
                      and as $\stexpone$ if $\terminates{\avert}$.
    This definition does not depend on the specifically chosen list representation of $\transitionsinfrom{\aonechart}{\avert}$, 
    because $\asys$ extends \ACI, and therefore it contains the associativity, commutativity, and idempotency axioms~for~$\sstexpsum$.
  
  By an \emph{\provablein{\asys} solution of $\aonechart$} 
    (with \emph{principal value $\asol{\start}$} at the start vertex $\start$)
    we mean a star expression function $\sasol$ on $\aonechart$ that is an \provablein{\asys} solution of $\aonechart$ at every vertex~of~$\aonechart$.
\end{defi}

\section{Layered loop existence and elimination, and LLEE-witnesses}
  \label{LEE}

In this subsection we recall principal definitions and statements from \cite{grab:fokk:2020a,grab:2021:TERMGRAPH-postproceedings}.
We keep formalities to a minimum as necessary for our purpose (in particular for `\LLEEwitnesses').  
  
\smallskip 
  
  A \onechart~$\aoneloop = \tuple{\verts,\actions,\sone,\start,\transs,\termexts}$ is called a \emph{loop \onechart} if
  it satisfies three conditions:\vspace*{-0.25ex}
  \begin{enumerate}[(L1)]
    \item{}\label{loop:1}
      There is an infinite path from the start vertex $\start$.
    \item{}\label{loop:2}  
      Every infinite path from $\start$ returns to $\start$ after a positive number of transitions.
    \item{}\label{loop:3}
      Immediate termination is only permitted at the start vertex, that is, $\termexts\subseteq\setexp{\start}$.
  \end{enumerate}\vspace*{-0.25ex}
  We call the transitions from $\start$ \emph{\loopentry\ transitions},
  and all other transitions \emph{\loopbody\ transitions}.
  A \emph{loop \subonechart\ of} a \onechart~$\aonechart$
    is a loop \onechart~$\aoneloop$
    that is a \subonechart\ of $\aonechart$ 
      with some vertex $\avert\in\verts$ of $\aonechart$ as start vertex,
    such that $\aoneloop$ is constructed, for a nonempty set $\asettranss$ of transitions of $\aonechart$ from $\avert$,
    by all paths that start with a transition in $\asettranss$ and continue onward until $\avert$ is reached again
  (so the transitions in $\asettranss$ are the \loopentrytransitions~of~$\aoneloop$).
%
%
%
%
%

The result of \emph{eliminating a loop \subonechart\ $\aoneloop$ from a \onechart\ $\aonechart$}
  arises by removing all \loopentrytransitions\ of $\aoneloop$ from $\aonechart$, 
  and then also removing all vertices and transitions that become unreachable. 
  We say that a \onechart\ $\aonechart$ has the \emph{loop existence and elimination property} (\LEE)
  if the procedure, started on~$\aonechart$, of repeated eliminations of loop \subonecharts\
  results in a \onechart\ without an infinite path.
  If, in a successful elimination process from a \onechart~$\aonechart$,
  \ul{loop-entry transitions} are \ul{never} removed from the body of a previously eliminated loop \subonechart,
  then we say that $\aonechart$ satisfies \emph{layered \LEE} (\LLEE),
  and is a \emph{\LLEEonechart}. 
While the property \LLEE\ leads to a formally easier concept of `witness', it is equivalent to 
                                                                                               \LEE. 
(For an example of a \LEEwitness\ that is not layered, see further below on page~\pageref{non-ex-LLEEw}.)                                                                                    
%
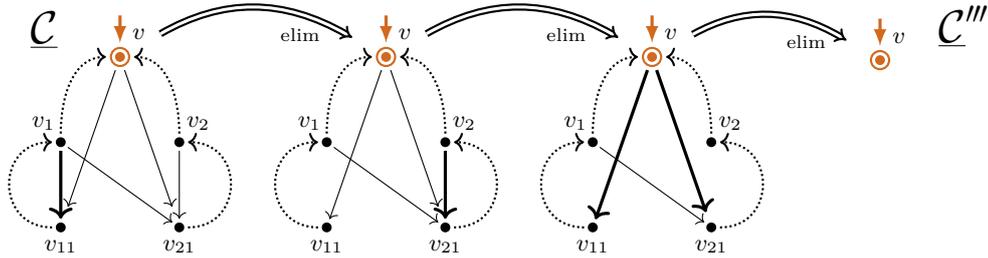
\begin{center}
  \begin{tikzpicture}

\matrix[anchor=center,row sep=1cm,column sep=0.65cm,
        every node/.style={draw,very thick,circle,minimum width=2.5pt,fill,inner sep=0pt,outer sep=2pt}] at (0,0) {
    & \node[chocolate](v--1){};
  \\           
  \node(v1){};
    & & \node(v2){};
  \\
  \node(v11){};
    & & \node(v21){};
  \\   
  };   
\path (v--1) ++ (-1cm,0.35cm) node(label){\LARGE $\aonechart$};

\draw[<-,very thick,>=latex,chocolate,shorten <= 2pt](v--1) -- ++ (90:0.55cm);   
\draw[thick,chocolate] (v--1) circle (0.12cm);
\path (v--1) ++ (0.25cm,0.3cm) node{\small $\avert$};
\draw[->,shorten >= 0.175cm,shorten <= 2pt] 
  (v--1) to 
         (v11);
\draw[->,shorten >= 0.175cm,shorten <= 2pt] 
  (v--1) to 
         (v21);

\path (v1) ++ (-0.225cm,0.25cm) node{\small $\averti{1}$};
\draw[->,very thick,shorten >= 0pt]
  (v1) to 
          (v11);
\draw[->,thick,densely dotted,out=90,in=180,distance=0.5cm,shorten >=2pt](v1) to (v--1);
\draw[->,shorten <= 0pt,shorten >= 0pt] (v1) to 
                                                (v21); 

\path (v11) ++ (0cm,-0.25cm) node{\small $\averti{11}$};
\draw[->,thick,densely dotted,out=180,in=180,distance=0.75cm](v11) to (v1);

\path (v2) ++ (0.25cm,0.25cm) node{\small $\averti{2}$};
\draw[->,shorten >= 0pt]
  (v2) to 
          (v21);
\draw[->,thick,densely dotted,out=90,in=0,distance=0.5cm,shorten >= 2pt](v2) to (v--1);

\path (v21) ++ (0cm,-0.25cm) node{\small $\averti{21}$};
\draw[->,thick,densely dotted,out=-0,in=0,distance=0.75cm](v21) to (v2);

\matrix[anchor=center,row sep=1cm,column sep=0.65cm,
        every node/.style={draw,very thick,circle,minimum width=2.5pt,fill,inner sep=0pt,outer sep=2pt}] at (3.5,0) {
    & \node[chocolate](v--2){};
  \\           
  \node(v1){};
    & & \node(v2){};
  \\
  \node(v11){};
    & & \node(v21){};
  \\   
  };   
  %

\draw[<-,very thick,>=latex,chocolate,shorten <= 2pt](v--2) -- ++ (90:0.55cm);   
\draw[thick,chocolate] (v--2) circle (0.12cm);
\path (v--2) ++ (0.25cm,0.3cm) node{\small $\avert$};
\draw[->,shorten <= 2pt] 
  (v--2) to 
         (v11);
\draw[->,shorten >= 0.175cm,shorten <= 2pt] 
  (v--2) to 
         (v21);

\path (v1) ++ (-0.225cm,0.25cm) node{\small $\averti{1}$};
\draw[->,thick,densely dotted,out=90,in=180,distance=0.5cm,shorten >=2pt](v1) to (v--2);
\draw[->,shorten <= 0pt,shorten >= 0pt] (v1) to 
                                                (v21); 

\path (v11) ++ (0cm,-0.25cm) node{\small $\averti{11}$};
\draw[->,thick,densely dotted,out=180,in=180,distance=0.75cm](v11) to (v1);

\path (v2) ++ (0.25cm,0.25cm) node{\small $\averti{2}$};
\draw[->,very thick,shorten >= 0pt]
  (v2) to 
          (v21);
\draw[->,thick,densely dotted,out=90,in=0,distance=0.5cm,shorten >= 2pt](v2) to (v--2);

\path (v21) ++ (0cm,-0.25cm) node{\small $\averti{21}$};
\draw[->,thick,densely dotted,out=-0,in=0,distance=0.75cm](v21) to (v2);

\draw[-implies,thick,double equal sign distance, bend left,distance=1.45cm,
               shorten <= 0.5cm,shorten >= 0.4cm
               ] (v--1) to node[below,pos=0.7]{\scriptsize elim} (v--2);

\matrix[anchor=center,row sep=1cm,column sep=0.65cm,
        every node/.style={draw,very thick,circle,minimum width=2.5pt,fill,inner sep=0pt,outer sep=2pt}] at (7,0) {
    & \node[chocolate](v--3){};
  \\           
  \node(v1){};
    & & \node(v2){};
  \\
  \node(v11){};
    & & \node(v21){};
  \\   
  };   
  %

\draw[<-,very thick,>=latex,chocolate,shorten <= 2pt](v--3) -- ++ (90:0.55cm);   
\draw[thick,chocolate] (v--3) circle (0.12cm);
\path (v--3) ++ (0.25cm,0.3cm) node{\small $\avert$};
\draw[->,very thick,shorten <= 2pt] 
  (v--3) to 
         (v11);
\draw[->,very thick,shorten >= 0.1cm,shorten <= 2pt] 
  (v--3) to 
         (v21);

\path (v1) ++ (-0.225cm,0.25cm) node{\small $\averti{1}$};
\draw[->,thick,densely dotted,out=90,in=180,distance=0.5cm,shorten >=2pt](v1) to (v--3);
\draw[->,shorten <= 0pt,shorten >= 0pt] (v1) to 
                                                (v21); 

\path (v11) ++ (0cm,-0.25cm) node{\small $\averti{11}$};
\draw[->,thick,densely dotted,out=180,in=180,distance=0.75cm](v11) to (v1);

\path (v2) ++ (0.25cm,0.25cm) node{\small $\averti{2}$};
\draw[->,thick,densely dotted,out=90,in=0,distance=0.5cm,shorten >= 2pt](v2) to (v--3);

\path (v21) ++ (0cm,-0.25cm) node{\small $\averti{21}$};
\draw[->,thick,densely dotted,out=-0,in=0,distance=0.75cm](v21) to (v2);

\draw[-implies,thick,double equal sign distance, bend left,distance=1.45cm,
               shorten <= 0.5cm,shorten >= 0.4cm
               ] (v--2) to node[below,pos=0.7]{\scriptsize elim} (v--3);

\matrix[anchor=center,row sep=1cm,column sep=0.65cm,
        every node/.style={draw,very thick,circle,minimum width=2.5pt,fill,inner sep=0pt,outer sep=2pt}] at (10,0) {
    & \node[chocolate](v--4){};
  \\           
  \node[draw=none,fill=none](v1){};
    & & \node[draw=none,fill=none](v2){};
  \\
  \node[draw=none,fill=none](v11){};
    & & \node[draw=none,fill=none](v21){};
  \\   
  };   
\path (v--4) ++ (1.1cm,0.45cm) node(label){\LARGE $\aonechart'''$};

\draw[<-,very thick,>=latex,chocolate,shorten <= 2pt](v--4) -- ++ (90:0.55cm);   
\draw[thick,chocolate] (v--4) circle (0.12cm);
\path (v--4) ++ (0.25cm,0.3cm) node{\small $\avert$};

\draw[-implies,thick,double equal sign distance, bend left,distance=1.25cm,
               shorten <= 0.5cm,shorten >= 0.4cm
               ] (v--3) to node[below,pos=0.7]{\scriptsize elim} (v--4);

\end{tikzpicture} %
\end{center}\vspace*{-1.25ex}

The picture above shows a successful run of the loop elimination procedure. 
In brown we highlight start vertices by \picarrowstart, and immediate termination with a boldface ring.
The \loopentry\ transitions of loop \subonecharts\ that are eliminated in the next step are marked in bold. 
We have neglected action labels here, except for indicating \onetransitions\ by~dotted~arrows.
Since the graph $\aonechart'''$ that is reached after three loop-subgraph elimination steps from the \onechart\ $\aonechart$ does not have 
an infinite path, and no \loopentry\ transitions have been removed from a previously eliminated loop \subonechart,
we conclude that $\aonechart$ satisfies \LEE\ and \LLEE. 
 
\begin{center}\vspace{-0.5ex}
  \scalebox{0.95}{\begin{tikzpicture}

\matrix[anchor=center,row sep=1cm,column sep=0.75cm,
        every node/.style={draw,very thick,circle,minimum width=2.5pt,fill,inner sep=0pt,outer sep=2pt}] at (0,0) {
    & \node[chocolate](v){};
  \\[-0.25ex]              
  \node(v1){};
    & & \node(v2){};
  \\[0.25cm]
  \node(v11){};
    & & \node(v21){};
  \\   
  };   
\path (v) ++ (-1.15cm,0.15cm) node(label){\LARGE $\aonecharthati{1}$};

\draw[<-,very thick,>=latex,chocolate,shorten <= 2pt](v) -- ++ (90:0.55cm);   
\draw[thick,chocolate] (v) circle (0.12cm);
\path (v) ++ (0.25cm,-0.3cm) node{\small $\avert$};
\draw[->,thick,darkcyan,shorten >= 0.175cm,shorten <= 2pt] 
  (v) to 
         node[right,pos=0.4,xshift=-0.075cm,yshift=1pt]{\small $\loopnsteplab{3}$}  (v11);
\draw[->,thick,darkcyan,shorten >= 0.175cm,shorten <= 2pt] 
  (v) to 
         node[left,pos=0.6,xshift=0.075cm,yshift=1pt]{\small $\loopnsteplab{3}$} (v21);

\path (v1) ++ (-0.225cm,0.25cm) node{\small $\averti{1}$};
\draw[->,thick,darkcyan,shorten >= 0pt]
  (v1) to 
          node[left,pos=0.6,xshift=0.075cm]{\small $\loopnsteplab{1}$} (v11);
\draw[->,thick,densely dotted,out=90,in=180,distance=0.5cm,shorten >=2pt](v1) to (v);
\draw[->,shorten <= 0pt,shorten >= 0pt] (v1) to 
                                                (v21); 

\path (v11) ++ (0cm,-0.25cm) node{\small $\averti{11}$};
\draw[->,thick,densely dotted,out=180,in=180,distance=0.75cm](v11) to (v1);

\path (v2) ++ (0.25cm,0.25cm) node{\small $\averti{2}$};
\draw[->,thick,darkcyan,shorten >= 0pt]
  (v2) to 
          node[right,pos=0.6,xshift=-0.075cm]{\small $\loopnsteplab{2}$} (v21);
\draw[->,thick,densely dotted,out=90,in=0,distance=0.5cm,shorten >= 2pt](v2) to (v);

\path (v21) ++ (0cm,-0.25cm) node{\small $\averti{21}$};
\draw[->,thick,densely dotted,out=0,in=0,distance=0.75cm](v21) to (v2);

\matrix[anchor=center,row sep=1cm,column sep=0.75cm,
        every node/.style={draw,very thick,circle,minimum width=2.5pt,fill,inner sep=0pt,outer sep=2pt}] at (4.5,0) {
    & \node[chocolate](v){};
  \\[-0.25ex]              
  \node(v1){};
    & & \node(v2){};
  \\[0.25cm]
  \node(v11){};
    & & \node(v21){};
  \\   
  };   
\path (v) ++ (-1.15cm,0.15cm) node(label){\LARGE $\aonecharthati{2}$};

\draw[<-,very thick,>=latex,chocolate,shorten <= 2pt](v) -- ++ (90:0.55cm);   
\draw[thick,chocolate] (v) circle (0.12cm);
\path (v) ++ (0.25cm,-0.3cm) node{\small $\avert$};
\draw[->,thick,darkcyan,shorten >= 0.175cm,shorten <= 2pt] 
  (v) to 
         node[right,pos=0.4,xshift=-0.075cm,yshift=1pt]{\small $\loopnsteplab{4}$}  (v11);
\draw[->,thick,darkcyan,shorten >= 0.175cm,shorten <= 2pt] 
  (v) to 
         node[left,pos=0.6,xshift=0.075cm,yshift=1pt]{\small $\loopnsteplab{3}$} (v21);

\path (v1) ++ (-0.225cm,0.25cm) node{\small $\averti{1}$};
\draw[->,thick,darkcyan,shorten >= 0pt]
  (v1) to 
          node[left,pos=0.6,xshift=0.075cm]{\small $\loopnsteplab{2}$} (v11);
\draw[->,thick,densely dotted,out=90,in=180,distance=0.5cm,shorten >=2pt](v1) to (v);
\draw[->,shorten <= 0pt,shorten >= 0pt] (v1) to 
                                                (v21); 

\path (v11) ++ (0cm,-0.25cm) node{\small $\averti{11}$};
\draw[->,thick,densely dotted,out=180,in=180,distance=0.75cm](v11) to (v1);

\path (v2) ++ (0.25cm,0.25cm) node{\small $\averti{2}$};
\draw[->,thick,darkcyan,shorten >= 0pt]
  (v2) to 
          node[right,pos=0.6,xshift=-0.075cm]{\small $\loopnsteplab{1}$} (v21);
\draw[->,thick,densely dotted,out=90,in=0,distance=0.5cm,shorten >= 2pt](v2) to (v);

\path (v21) ++ (0cm,-0.25cm) node{\small $\averti{21}$};
\draw[->,thick,densely dotted,out=-0,in=0,distance=0.75cm](v21) to (v2);

\matrix[anchor=center,row sep=1cm,column sep=0.75cm,
        every node/.style={draw,very thick,circle,minimum width=2.5pt,fill,inner sep=0pt,outer sep=2pt}] at (9,0) {
    & \node[chocolate](v){};
  \\[-0.25ex]              
  \node(v1){};
    & & \node(v2){};
  \\[0.25cm]
  \node(v11){};
    & & \node(v21){};
  \\   
  };   
\path (v) ++ (-1.15cm,0.15cm) node(label){\LARGE $\aonecharthati{3}$};

\draw[<-,very thick,>=latex,chocolate,shorten <= 2pt](v) -- ++ (90:0.55cm);   
\draw[thick,chocolate] (v) circle (0.12cm);
\path (v) ++ (0.25cm,-0.3cm) node{\small $\avert$};
\draw[->,thick,darkcyan,shorten >= 0.175cm,shorten <= 2pt] 
  (v) to 
         node[right,pos=0.4,xshift=-0.075cm,yshift=1pt]{\small $\loopnsteplab{2}$}  (v11);
\draw[->,thick,darkcyan,shorten >= 0.175cm,shorten <= 2pt] 
  (v) to 
         node[left,pos=0.6,xshift=0.075cm,yshift=1pt]{\small $\loopnsteplab{2}$} (v21);

\path (v1) ++ (-0.225cm,0.25cm) node{\small $\averti{1}$};
\draw[->,thick,darkcyan,shorten >= 0pt]
  (v1) to 
          node[left,pos=0.6,xshift=0.075cm]{\small $\loopnsteplab{1}$} (v11);
\draw[->,thick,densely dotted,out=90,in=180,distance=0.5cm,shorten >=2pt](v1) to (v);
\draw[->,shorten <= 0pt,shorten >= 0pt] (v1) to 
                                                (v21); 

\path (v11) ++ (0cm,-0.25cm) node{\small $\averti{11}$};
\draw[->,thick,densely dotted,out=180,in=180,distance=0.75cm](v11) to (v1);

\path (v2) ++ (0.25cm,0.25cm) node{\small $\averti{2}$};
\draw[->,thick,darkcyan,shorten >= 0pt]
  (v2) to 
          node[right,pos=0.6,xshift=-0.075cm]{\small $\loopnsteplab{1}$} (v21);
\draw[->,thick,densely dotted,out=90,in=0,distance=0.5cm,shorten >= 2pt](v2) to (v);

\path (v21) ++ (0cm,-0.25cm) node{\small $\averti{21}$};
\draw[->,thick,densely dotted,out=-0,in=0,distance=0.75cm](v21) to (v2);

\end{tikzpicture}}\vspace*{-1.5ex}
\end{center}
A \emph{\LLEEwitness\ $\aonecharthat$ of} a \onechart~$\aonechart$
is the recording of a 
                          successful run of the loop elimination procedure
by attaching to a transition $\atrans$ of $\aonechart$ the marking label $n$ for $n\in\natplus$ 
 (in pictures indicated as $\looplab{n}$, in steps as $\sredi{\looplab{n}}$) 
 forming a \emph{\loopentry\ transition}
if $\atrans$ is eliminated in the $n$\nb-th step,
and by attaching marking label $0$ to all other transitions of $\aonechart$
 (in pictures neglected, in steps indicated as $\sredi{\bodylab}$)
 forming a \emph{body transition}. 
Formally, \LLEEwitnesses\ arise as \emph{\entrybodylabeling{s}} from \onecharts,
 and are charts in which the transition labels are pairs of action labels over $\actions$,
 and marking labels in $\nat$.
We say that a \LLEEwitness~$\aonecharthat$ \emph{is guarded}
 if all \loopentrytransitions\ are proper, which means that they have a proper-action~transition~label.
  
The \entrybodylabeling\ $\aonecharthati{1}$ above of the \onechart~$\aonechart$ 
is a \LLEEwitness\ that arises from the run of the loop elimination procedure earlier above. 
The \entrybodylabelings\ $\aonecharthati{2}$ and $\aonecharthati{3}$ of $\aonechart$ record two other successful runs of the loop elimination procedure
of length 4 and 2, respectively, where for $\aonecharthati{3}$ we have permitted to eliminate two loop subcharts at different vertices
together in the first step. 
The \onechart~$\aonechart$ only has layered \LEEwitnesses.
But that is not the case for the \onechart~$\conechart$ below:

\begin{center}\vspace{-0.75ex}\label{non-ex-LLEEw}%
  \scalebox{0.95}{\begin{tikzpicture}  
\matrix[anchor=center,row sep=0.6cm,column sep=1.15cm,ampersand replacement=\&,
        every node/.style={draw,very thick,circle,minimum width=2.5pt,fill,inner sep=0pt,outer sep=2pt}] at (0,0) {
  \node(v){}; 
  \\
               \& \node(u){};
  \\
  \node(w1){};
  \\
  \\
  \node(w2){}; 
  \\   
  };     
   
\path (v) ++ (0.25cm,0.15cm) node{$\avert$}; 
\draw[<-,very thick,>=latex,chocolate](v) -- ++ (90:0.5cm);
\draw[->] (v) to (u);
\draw[->] (v) to (w1);
\path (v) ++ (-0.7cm,0.45cm) node(label){\LARGE $\conechart$};

\path (u) ++ (-0.05cm,-0.25cm) node{$\cvert$}; 
\draw[->] (u) to (w1);
\draw[->,out=-50,in=50,distance=1.2cm] (u) to (u);
\path (w1) ++ (-0.325cm,0cm) node{$\bverti{1}$};
\draw[->] (w1) to (w2);
\path (w2) ++ (0.35cm,-0.15cm) node{$\bverti{2}$};
\draw[->,out=180,in=180,distance=1.85cm] (w2) to (v);

\matrix[anchor=center,row sep=0.6cm,column sep=1.15cm,ampersand replacement=\&,
        every node/.style={draw,very thick,circle,minimum width=2.5pt,fill,inner sep=0pt,outer sep=2pt}] at (4.25,0) {
  \node(v){}; 
  \\
               \& \node(u){};
  \\
  \node(w1){};
  \\
  \\
  \node(w2){}; 
  \\   
  };     
   
\path (v) ++ (0.25cm,0.15cm) node{$\avert$}; 
\draw[<-,very thick,>=latex,chocolate](v) -- ++ (90:0.5cm);
\draw[->] (v) to (u);
\draw[->,very thick,darkcyan] (v) to node[right,xshift=-0.075cm,yshift=0.1cm]{$\loopnsteplab{1}$} (w1);
\path (v) ++ (-0.7cm,0.5cm) node(label){\LARGE $\conecharthati{1}$};
\path (u) ++ (-0.05cm,-0.25cm) node{$\cvert$}; 
\draw[->] (u) to (w1);
\draw[->,very thick,darkcyan,out=-50,in=50,distance=1.2cm] (u) to node[right,xshift=-0.075cm]{$\loopnsteplab{2}$} (u);
\path (w1) ++ (-0.325cm,0cm) node{$\bverti{1}$};
\draw[->,very thick,darkcyan] (w1) to node[right,xshift=-0.075cm,yshift=0.1cm]{$\loopnsteplab{3}$} (w2);
\path (w2) ++ (0.35cm,-0.15cm) node{$\bverti{2}$};
\draw[->,out=180,in=180,distance=1.85cm] (w2) to (v);

\matrix[anchor=center,row sep=0.6cm,column sep=1.15cm,ampersand replacement=\&,
        every node/.style={draw,very thick,circle,minimum width=2.5pt,fill,inner sep=0pt,outer sep=2pt}] at (8.75,0) {
  \node(v){}; 
  \\
               \& \node(u){};
  \\
  \node(w1){};
  \\
  \\
  \node(w2){}; 
  \\   
  };     
   
\path (v) ++ (0.25cm,0.15cm) node{$\avert$}; 
\draw[<-,very thick,>=latex,chocolate](v) -- ++ (90:0.5cm);
\draw[->,very thick,darkcyan] (v) to node[above,pos=0.6,xshift=0.025cm,yshift=0.025cm]{$\loopsteplab{3}$} (u);
\draw[->,very thick,darkcyan] (v) to node[right,xshift=-0.075cm,yshift=0.1cm]{$\loopnsteplab{1}$} (w1);
\path (v) ++ (-0.7cm,0.5cm) node(label){\LARGE $\conecharthati{2}$};
\path (u) ++ (-0.05cm,-0.25cm) node{$\cvert$}; 
\draw[->] (u) to (w1);
\draw[->,very thick,darkcyan,out=-50,in=50,distance=1.2cm] (u) to node[right,xshift=-0.075cm]{$\loopsteplab{2}$} (u);
\path (w1) ++ (-0.325cm,0cm) node{$\bverti{1}$};
\draw[->] (w1) to (w2);

\path (w2) ++ (0.35cm,-0.15cm) node{$\bverti{2}$};
\draw[->,out=180,in=180,distance=1.85cm] (w2) to (v);

\end{tikzpicture}  }\vspace*{-1.5ex}
\end{center}
The \entrybodylabeling~$\conecharthati{1}$ of $\conechart$ as above is a \LEEwitness\ that is not layered:
  in the third loop \subonechart\ elimination step that is recorded in $\conecharthati{1}$
  the \loopentrytransition\ from $\bverti{1}$ to $\bverti{2}$ is removed,
  which is in the body of the loop \subonechart\ at $\avert$ with \loopentrytransition\ from $\avert$ to $\bverti{1}$
    that (by that time) 
          has already been removed in the first \loopelimination\ step as recorded in $\conecharthati{1}$.
By contrast, the \entrybodylabeling~$\conecharthati{2}$ of $\conechart$ above is a layered \LEEwitness.
In general it can be shown that every \LEEwitness\ that is not layered can be transformed into a \LLEEwitness\ of the same underlying \onechart.    
Indeed, the step from $\conecharthati{1}$ to $\conecharthati{2}$ in the example above,
  which transfers the \loopentrytransition\ marking label $\loopsteplab{3}$
  from the transition from $\bverti{1}$ to $\bverti{2}$ over to the transition from $\avert$ to $\cvert$,
hints at the proof of this statement.
However, we do not need this result, because
  we will be able to use the guaranteed existence of \LLEEwitnesses\ (see Thm.~\ref{thm:onechart-int:LLEEw})
  for the \onechart\ interpretation~below~(see~Def.~\ref{def:onechartof}).

In a \LLEEwitness\ we denote by $\avert \descendsinloopto \bvert$, and by $\bvert \convdescendsinloopto \avert$, 
  that $\bvert$ is in the body of the loop \subonechart\ at $\avert$,
  which means that there is a path $\avert \redi{\looplab{n}} \avertacc \redrtci{\bodylab} \bvert$ 
    from $\avert$ via a \loopentrytransition\ and subsequent body transitions \ul{without} encountering $\avert$ again. 
  
\begin{lem}\label{lem:descsteps:bodysteps:wf}
  The relations $\sdescendsinloopto$ and $\sredi{\bodylab}$ defined by a \LLEEwitness~$\aonecharthat$
  of a \onechart~$\aonechart$ satisfy: 
  \begin{enumerate}[(i)]
    \item
      $\sconvdescendsinlooptotc$ is a well-founded, strict partial order on $\verts$.
    \item
      $\sconvredtci{\bodylab}$ is a well-founded strict partial order on $\verts$.  
  \end{enumerate}
\end{lem}

\begin{defi}[$\sone$-chart interpretation of star expressions]\label{def:onechartof}
  By the \emph{\onechart\ interpretation} $\onechartof{\astexp}$ of a star expression $\astexp$
  we mean the \onechart\ that arises together with the \entrybodylabeling\ $\onecharthatof{\astexp}$ 
  as the $\astexp$\nb-rooted sub-\LTS\ generated by $\setexp{\astexp}$
  according to the following TSS:\vspace*{-3ex}
  \begin{center}
    $
    \begin{gathered}
      \begin{aligned}
        & 
        \AxiomC{$\phantom{a_i \:\lti{a_i}{\darkcyan{\bodylabcol}}\: \stexpone}$}
        \UnaryInfC{$a \:\lti{a}{\bodylabcol}\: \stexpone$}
        \DisplayProof
        & & 
        \AxiomC{$ \astexpi{i} \:\lti{\aact}{\darkcyan{\alab}}\: \asstexpacci{i} $}
        \AxiomC{\scriptsize ($i\in\setexp{1,2}$)}
        \insertBetweenHyps{\hspace*{-0ex}}
        \BinaryInfC{$ \stexpsum{\astexpi{1}}{\astexpi{2}} \:\lti{\aact}{\bodylabcol}\: \asstexpacci{i} $}
        \DisplayProof 
        & &
        \AxiomC{$   \phantom{\astexpi{1}}
                 \astexp \:\lti{a}{\darkcyan{\alab}}\: \asstexpacc 
                   \phantom{\asstexpacci{1}} $}
        \AxiomC{\scriptsize (if $\fap{\textit{nd$^+$}}{\astexp}$)}   
        \insertBetweenHyps{\hspace*{-2ex}}        
        \BinaryInfC{$\stexpit{\astexp} \:\lti{\aact}{\darkcyan{\loopnsteplab{\sth{\stexpit{\astexp}}}}}\: \stexpstackprod{\asstexpacc}{\stexpit{\astexp}}$}
        \DisplayProof
        & &
        \AxiomC{$   \phantom{\astexpi{1}}
                 \astexp \:\lti{a}{\darkcyan{\alab}}\: \asstexpacc 
                   \phantom{\astexpacci{1}} $}
        \insertBetweenHyps{\hspace*{-2ex}}     
        \AxiomC{\scriptsize (if $\lognot{\fap{\textit{nd$^+$}}{\astexp}}$)} 
        \BinaryInfC{$\stexpit{\astexp} \:\lti{a}{\bodylabcol}\: \stexpstackprod{\asstexpacc}{\stexpit{\astexp}}$}
        \DisplayProof 
      \end{aligned}
      \\
      \begin{aligned}
        &   
        \AxiomC{$ \asstexpi{1} \:\lti{\alert{\aoneact}}{\darkcyan{\alab}}\: \asstexpacci{1} $}
        \UnaryInfC{$ \stexpprod{\asstexpi{1}}{\astexpi{2}} \:\lti{\alert{\aoneact}}{\darkcyan{\alab}}\: \stexpprod{\asstexpacci{1}}{\astexpi{2}} $}
        \DisplayProof
        & & 
        \AxiomC{$ \asstexpi{1} \:\lti{\alert{\aoneact}}{\darkcyan{\alab}}\: \asstexpacci{1} $}
        \UnaryInfC{$ \stexpstackprod{\asstexpi{1}}{\stexpit{\astexpi{2}}} \:\lti{\alert{\aoneact}}{\darkcyan{\alab}}\: \stexpstackprod{\asstexpacci{1}}{\stexpit{\astexpi{2}}} $}
        \DisplayProof
        & &
        \AxiomC{$\terminates{\astexpi{1}}$}
        \AxiomC{$ \astexpi{2} \:\lti{a}{\darkcyan{\alab}}\: \asstexpacci{2} $}
        \BinaryInfC{$ \stexpprod{\astexpi{1}}{\astexpi{2}} \:\lti{a}{\bodylabcol}\: \asstexpacci{2} $}
        \DisplayProof
        & &
        \AxiomC{$ \hspace*{5ex} \terminates{\astexpi{1}}\rule{0pt}{13pt} \hspace*{5ex} $}
        \UnaryInfC{$ \stexpstackprod{\astexpi{1}}{\stexpit{\astexpi{2}}} \:\lti{\sone}{\bodylabcol}\: \stexpit{\astexpi{2}} $}
        \DisplayProof
      \end{aligned}
    \end{gathered}
    $
  \end{center}
  where $\darkcyan{\alab} \in \setexp{\bodylabcol} \cup \descsetexp{ \darkcyan{\loopnsteplab{\aLname}} }{ \aLname\in\natplus }$,
  and star expressions using a `stacked' product operation $\sstackprod$ are permitted,
  which helps to record iterations from which derivatives originate.
  Immediate termination for expressions of $\onechartof{\astexp}$ is defined by the same rules 
  as in Def.~\ref{def:StExpTSS} (for star expressions only, preventing immediate termination for expressions with stacked product $\sstackprod$).
  The condition $\normedplus{e}$ means that $e$ permits a positive length path to an expression $f$~with~$\terminates{f}$. 
  We use a projection function $\sproj$ that 
    changes occurrences of stacked product $\sstackprod$ into product $\ssprod$.
\end{defi}  

\begin{thm}[\cite{grab:2020:scpgs-arxiv,grab:2021:TERMGRAPH-postproceedings}]\label{thm:onechart-int:LLEEw}
  For every $\astexp\in\StExpover{\actions}$,
    (a)~the \entrybodylabeling~$\onecharthatof{\astexp}$ of $\onechartof{\astexp}$ is a \LLEEwitness\ of $\onechartof{\astexp}$,
    and
    (b) the projection function $\sproj$ defines a \onebisimulation\ from the \onechart\ interpretation $\onechartof{\astexp}$ of $\astexp$
        to the chart interpretation $\chartof{\astexp}$ of $\astexp$,
        and hence $\onechartof{\astexp} \onebisim \chartof{\astexp}$.  
\end{thm}

Lem.~\ref{lem:onechart-int:milnersysmin:solvable} below follows
  from the next lemma, whose proof we sketch in the appendix. 

\begin{lem}\label{lem:FT:onechart-int}
  $ \proj{\asstexp}
     \milnersysmineq
   \terminatesconstof{\onechartof{\asstexp}}{\asstexp}  
     +  
   \sum_{i=1}^{n} \stexpprod{\aoneacti{i}}{\proj{\asstexpacci{i}}} $,
  given a list representation\vspace*{-1pt}
  $\transitionsinfrom{\onechartof{\asstexp}}{\bvert}
     =
   \descsetexpbig{ \asstexp \lt{\aoneacti{i}} \asstexpacci{i} }{ i\in\setexp{1,\ldots,n} }$
  of the transitions from $\asstexp$ in $\onechartof{E}$.
\end{lem}

\begin{lem}\label{lem:onechart-int:milnersysmin:solvable}
  For every star expression $\astexp\in\StExpover{\actions}$ 
    with \onechart\ interpretation $\onechartof{\astexp} = \tuple{\vertsof{\astexp},\actions,\sone,\astexp,\transs,\exts}$
  the \starexpression\ function $\sasol \funin \vertsof{\astexp} \to \StExpover{\actions}$, $\asstexp \mapsto \proj{\asstexp}$
    is a \provablein{\milnersysmin} solution of $\onechartof{\astexp}$ with principal value $\astexp$.  
\end{lem}

\section{Coinductive version of Milner's proof system}%
  \label{coindmilnersys}

In this section we motivate and define `coinductive proofs', introduce coinductive versions of Milner's system $\milnersys$,
and establish first interconnections between these proof systems.

A finite \onebisimulation\ between the \onechart\ interpretations of star expressions $\astexpi{1}$ and $\astexpi{2}$
  can be viewed as a proof of the statement that $\astexpi{1}$ and $\astexpi{2}$ 
  define the same star behavior. 
This can be generalized by permitting
  finite \emph{\onebisimulations\ up to provability in $\milnersys$}, 
  that is, finite relations $\abisim$ on star expressions  
  for which pairs $\pair{\averti{1}}{\averti{2}}\in\abisim$ progress, via the (forth) and (back) conditions in Def.~\ref{def:onebisim},
    to pairs $\pair{\avertacci{1}}{\avertacci{2}}$ in the (infinite) composed relation $\seqin{\milnersys} \cdot \abisim \cdot \seqin{\milnersys}$.
Now \onebisimilarity\ up to $\milnersyseq$ entails \onebisimilarity\ of the \onechart\ interpretations, and bisimilarity of chart interpretations,
  due to soundness of $\milnersys$ (see Prop.~\ref{prop:milnersys:sound}).
In order to link, later in Section~\ref{coindmilnersys:2:milnersys}, coinductive proofs with proofs in Milner's system $\milnersys$,
  we will be interested in \onebisimulations\ up to $\seqin{\asys}$ for systems $\asys$ with $\ACI \subsystem \asys \isthmsubsumedby \milnersys$
  that have the form of \LLEEonecharts,
  which will guarantee such a connection. 
First we introduce a `\LLEEwitnessed\ coinductive proof' 
  as an equation-labeled, \LLEEonechart~$\aonechart$
  that defines a \onebisimulation\ up to $\eqin{\asys}$ for $\asys$ with $\ACI \subsystem \asys$ 
  from the left-/right-hand sides of equations on the~vertices~of~$\aonechart$~(see~Rem.~\ref{rem:coindproof:defs:onebisim:up:to}).   

\begin{defi}[(\text{\nf LLEE}-witnessed) coinductive proofs]\label{def:coindproof}
  Let 
      $\astexpi{1},\astexpi{2}\in\StExpover{\actions}$ be star expressions,
  and let $\asys$ be an \eqlogicbased\ proof system for star expressions over $\actions$ with $\ACI \subsystem \asys$. 
  
  A \emph{coinductive proof over $\asys$ of $\astexpi{1} \formeq \astexpi{2}$}  
  is a pair $\aCoProof = \pair{\aonechart}{\saeqfun}$
    where $\aonechart = \tuple{\verts,\actions,\sone,\start,\transs,\termexts}$ is a weakly guarded \onechart, 
      and $\saeqfun \funin \verts \to \StExpEqover{\actions}$ a labeling function of vertices of $\aonechart$ by formal equations over $\actions$
      such that for the functions
      $\saeqfuni{1},\saeqfuni{2} \funin \verts \to \StExpover{\actions}$ 
      that denote the star expressions $\aeqfuni{1}{\avert}$, and $\aeqfuni{2}{\avert}$,
      on the left-, and on the right-hand side of the equation $\aeqfun{\avert}$, respectively, the following conditions hold: 
  \begin{enumerate}[(cp1)]
    \item
      $\saeqfuni{1}$ and $\saeqfuni{2}$ are \provablein{\asys} solutions of $\aonechart$,
    \item
      $\astexpi{1} \synteq \aeqfuni{1}{\start}$ and $\astexpi{2} \synteq \aeqfuni{2}{\start}$. 
  \end{enumerate}
  \smallskip\hspace*{2.75ex}
  By a \emph{\LLEEwitnessed\ coinductive proof} we mean a coinductive proof $\aCoProof = \pair{\aonechart}{\saeqfun}$\vspace*{-2pt}
    where $\aonechart$ is a \LLEEonechart.
  We denote by $\astexp \coindproofeqin{\asys} \bstexp$ 
    that there is a coinductive proof over $\asys$ of $\astexp \formeq \bstexp$,
  and by $\astexp \LLEEcoindproofeqin{\asys} \bstexp$
      that there is a \LLEEwitnessed\ coinductive proof over $\asys$ of $\astexp \formeq \bstexp$. 
\end{defi}

\begin{exa}\label{ex:2:LLEEcoindproof}
  The statement 
  $\stexpit{(\stexpprod{\stexpit{\aact}}
                       {\stexpit{\bact}})}
     \LLEEcoindproofeqin{\milnersysmin}
   \stexpit{(\stexpsum{\aact}{\bact})}$ 
  can be established by the following \LLEEwitnessed\ coinductive proof $\aCoProof = \pair{\aonechart}{\saeqfun}$ over $\milnersysmin$
    where $\aonechart = \onechartof{(a^* \prod b^*)^*}$ has the indicated \LLEEwitness~$\onecharthatof{(a^* \prod b^*)^*}$ (see Thm.~\ref{thm:onechart-int:LLEEw})
    where framed boxes contain~vertex~names:
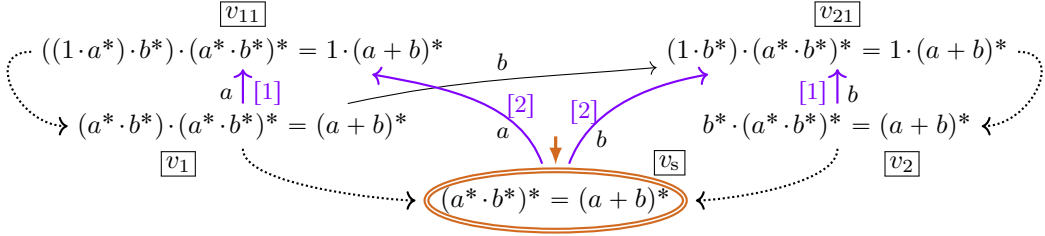
\vspace*{-3.75ex}
  \begin{flushleft}
    \hspace*{-0.8em}
    \begin{tikzpicture}\renewcommand{\stexpprod}[2]{{#1}\hspace*{1pt}{\sstexpprod}\hspace*{1pt}{#2}}
  
\matrix[anchor=center,row sep=0.4cm,column sep=-0.3cm] {
  \node(v11){$ \stexpprod{(\stexpprod{(\stexpprod{\stexpone}{\stexpit{\aact}})}{\stexpit{\bact}})}
                         {\stexpit{(\stexpprod{\stexpit{\aact}}{\stexpit{\bact}})}} 
                 \formeq
               \stexpprod{\stexpone}{\stexpit{(\stexpsum{\aact}{\bact})}} $};
    & & \node(v21){$ \stexpprod{(\stexpprod{\stexpone}{\stexpit{\bact}})}
                               {\stexpit{(\stexpprod{\stexpit{\aact}}{\stexpit{\bact}})}} 
                       \formeq
                     \stexpprod{\stexpone}{\stexpit{(\stexpsum{\aact}{\bact})}} $};
  \\                 
  \node(v1){$ \stexpprod{(\stexpprod{\stexpit{\aact}}{\stexpit{\bact}})}{\stexpit{(\stexpprod{\stexpit{\aact}}{\stexpit{\bact}})}} 
                \formeq
              \stexpit{(\stexpsum{\aact}{\bact})} $};
    & & \node(v2){$ \stexpprod{\stexpit{\bact}}{\stexpit{(\stexpprod{\stexpit{\aact}}{\stexpit{\bact}})}} 
                      \formeq
                    \stexpit{(\stexpsum{\aact}{\bact})} $};
  \\
    & \node(v){$\stexpit{(\stexpprod{\stexpit{\aact}}{\stexpit{\bact}})}
                  \formeq
                \stexpit{(\stexpsum{\aact}{\bact})}$};
  \\
  };   
 
\path (v11) ++ (0cm,0.5cm) node{\tightfbox{$\averti{11}$}};
\draw[->,thick,densely dotted,out=180,in=180,distance=0.65cm](v11) to (v1);

\path (v1) ++ (-0.85cm,-0.5cm) node{\tightfbox{$\averti{1}$}}; 
\draw[->,thick,darkcyan,shorten >= -2pt]
  (v1) to node[left,pos=0.35]{\small $\black{\aact}$} node[right,pos=0.35]{\small $\loopnsteplab{1}$} (v11);
\draw[->,thick,densely dotted,out=-90,in=180,distance=0.5cm,shorten >= 0.2cm](v1) to (v);
\draw[->,shorten <= 0pt,shorten >= 0pt,out=12.5,in=185] (v1) to node[above]{\small $\bact$} (v21); 

\path (v21) ++ (0cm,0.5cm) node{\tightfbox{$\averti{21}$}};
\draw[->,thick,densely dotted,out=0,in=0,distance=0.65cm](v21) to (v2);  
  
\path (v2) ++ (0.85cm,-0.5cm) node{\tightfbox{$\averti{2}$}}; 
\draw[->,thick,darkcyan,shorten >= -2pt]
  (v2) to node[right,pos=0.35]{\small $\black{\bact}$} node[left,pos=0.35]{\small $\loopnsteplab{1}$} (v21);
\draw[->,thick,densely dotted,out=-90,in=0,distance=0.5cm,shorten >= 0.2cm](v2) to (v);

\path (v) ++ (1.5cm,0.5cm) node{\tightfbox{$\start$}}; 
\draw[thick,chocolate,double] (v) ellipse (1.7cm and 0.4cm); 
\draw[<-,very thick,>=latex,chocolate,shorten <= 5pt](v) -- ++ (90:0.85cm); 
\draw[->,thick,darkcyan,out=110,in=-10,shorten <= 0.2cm] 
  (v) to node[below,pos=0.36]{\small $\black{\aact}$} node[above,pos=0.26,xshift=1pt,yshift=1pt]{\small $\loopnsteplab{2}$}  (v11);
\draw[->,thick,darkcyan,out=70,in=190,shorten <= 0.2cm] 
  (v) to node[below,pos=0.37]{\small $\black{\bact}$} node[above,pos=0.25,xshift=-1pt,yshift=1pt]{\small $\loopnsteplab{2}$} (v21);

\end{tikzpicture}    

  \end{flushleft}\vspace*{-0.5ex}
  Here we have drawn the \onechart~$\aonechart$ that carries the equations with its start vertex below
    in order to adhere to the prooftree intuition for the represented derivation, namely with the conclusion at the bottom. 
  We will do so repeatedly also below.
  Solution correctness conditions for the left-hand sides of the equations on $\aonechart$ 
    follows from Lem.~\ref{lem:onechart-int:milnersysmin:solvable}, due to $\aonechart = \onechartof{(a^* \prod b^*)^*}$
    as $(a^* \prod b^*)^*$ is the left-hand side of the conclusion.  
  However, we verify the correctness conditions for the left- and the right-hand sides for the (most involved) case of vertex $\averti{1}$
    together as follows
  (we usually neglect associative brackets, and combine some axiom applications): 
  \vspace{-0.75ex}
  \begin{align*}
    \hspace*{8ex} & \hspace*{-8ex}
      (a^* \cdot b^*) \cdot (a^* \cdot b^*)^*
        \milnersysmineq
    ((1 + a \cdot a^*) \cdot (1 + b \cdot b^*)) \cdot (a^* \cdot b^*)^*
    \\
      & {} \milnersysmineq
    (1 \cdot 1 + a \cdot a^* \cdot 1 + 1 \cdot b \cdot b^* + a \cdot a^* \cdot b \cdot b^*) \cdot (a^* \cdot b^*)^*
    \displaybreak[0]\\
      & {} \milnersysmineq
    (1 + a \cdot a^* + a \cdot a^* \cdot b \cdot b^* + b \cdot b^*) \cdot (a^* \cdot b^*)^*
    \displaybreak[0]\\
      & {} \milnersysmineq
    (1 + a \cdot a^* \cdot (1 + b \cdot b^*) + b \cdot b^*) \cdot (a^* \cdot b^*)^*
    \displaybreak[0]\\
      & {} \milnersysmineq
    (1 + a \cdot a^* \cdot b^* + b \cdot b^*) \cdot (a^* \cdot b^*)^*
    \\
      & {} \milnersysmineq
    1 \cdot (a^* \cdot b^*)^*  +  a \cdot (((1 \cdot a^*) \cdot b^*) \cdot (a^* \cdot b^*)^*)  +  b \cdot ((1 \cdot b^*) \cdot (a^* \cdot b^*)^*)
    \displaybreak[0]\\[0.75ex]
    (a + b)^* 
      & {} \milnersysmineq
    (a + b)^*  +  (a + b)^* 
      \milnersysmineq
    1 + (a + b) \cdot (a + b)^*  +  1 + (a + b) \cdot (a + b)^*
    \displaybreak[0]\\
      & {} \milnersysmineq
    1 + 1 + (a + b) \cdot (a + b)^*  + a \cdot (a + b)^* + b \cdot (a + b)^*
    \displaybreak[0]\\
      & {} \milnersysmineq
    1 + (a + b) \cdot (a + b)^*  + a \cdot (1 \cdot (a + b)^*) + b \cdot (1 \cdot (a + b)^*)
    \\
      & {} \milnersysmineq
    1 \cdot (a + b)^*  + a \cdot (1 \cdot (a + b)^*) + b \cdot (1 \cdot (a + b)^*) \vspace*{-1ex} 
  \end{align*}
  The solution conditions at the vertices $\avert$ and $\averti{2}$ can be verified analogously.
  At $\averti{11}$ and~at~$\averti{21}$ the solution conditions follow by uses of the axiom $\leftidstexpprod$ of $\milnersysmin$.
\end{exa}

\begin{lem}\label{lem:props:coindproofeq:LLEEcoindproofeq}
  Let $\sabinrel \in \setexpbig{ \scoindproofeqin{\asys}, \sLLEEcoindproofeqin{\asys} }$ 
    for some \eqlogicbased\ proof system $\asys$ with $\ACI \subsystem \asys$. 
    Then $\abinrel$ is reflexive, symmetric, and
    satisfies $ \seqin{\asys} \circ \sabinrel \subseteq \sabinrel $, 
              $ \sabinrel \circ \seqin{\asys} \subseteq \sabinrel $, and
              $ \seqin{\asys} \subseteq \sabinrel $. 
\end{lem}

%
  
\begin{rem}\label{rem:coindproof:defs:onebisim:up:to}
  For every coinductive proof $\aCoProof = \pair{\aonechart}{\saeqfun}$, whether $\aCoProof$ is \LLEEwitnessed\ or not, 
    over an \eqlogicbased\ proof system $\asys$ with $\ACI \subsystem \asys \isthmsubsumedby \milnersys$
  the finite relation defined by: 
  \begin{center}
    $
    \abisim
      \,\defdby\,
        \descsetexpBig{\displaystyle
          \Big\langle
            \terminatesconstof{\aonechart}{\avert}
              +
            \sum_{i=1}^{n} 
              \aoneacti{i} \prod \aeqfuni{1}{\averti{i}}
            ,\,
            \terminatesconstof{\aonechart}{\avert}
              +
            \sum_{i=1}^{n} 
              \aoneacti{i} \prod \aeqfuni{2}{\averti{i}}
          \Big\rangle
                       }{ \begin{gathered}
                             \transitionsinfrom{\aonechart}{\avert}
                               =
                             \descsetexpbig{ \avert \lt{\aoneacti{i}} \averti{i} }{ i \in\setexp{1,\ldots,n} },
                             \\[-0.75ex]
                             \avert\in\vertsof{\aonechart},\, \text{the set of vertices of $\aonechart$}
                          \end{gathered} }
    $
  \end{center}
  is a \onebisimulation\ up to $\seqin{\asys}$ 
  with respect to the labeled transition system on all star expressions that is defined by the TSS in Def.~\ref{def:chartof}. 
  This can be shown by using that the left-hand sides $\aeqfuni{1}{\avert}$, and respectively the right-hand sides $\aeqfuni{2}{\avert}$,
  of the equations $\aeqfun{\avert}$ in $\aCoProof$, for $\avert\in\vertsof{\aonechart}$,
    form \provablein{\asys} solutions of the \onechart\ $\aonechart$ that underlies $\aCoProof$.
\end{rem}  


\begin{defi}[proof systems \CLC, \CC\ for combining (\LLEE-witnessed) coinductive proofs]\label{def:CLC:CC}
  By the \emph{proof system \CLCover{\actions} for combining \LLEEwitnessed\ 
      coinductive proofs (over extensions of $\milnersysminover{\actions}\!$) between star expressions over $\actions$}
    we mean the Hilbert-style proof system whose \emph{formulas} are 
                                                                     equations between star expressions in $\StExpEqover{\actions}$
    or \LLEEwitnessed\ coinductive proofs over $\thplus{\milnersysminover{\actions}}{\bseteqs}$, where $\bseteqs\in\StExpEqover{\actions}$,  
    and whose \emph{rules} are those of~the~scheme: 
  \begin{equation*}
    \AxiomC{$ \cstexpi{1}   \formeq   \dstexpi{1}
                \; \ldots \;
              \cstexpi{n}   \formeq   \dstexpi{n} $}
    \AxiomC{$ \aLLEECoProofoverof{\thplus{\milnersysmin}{\aseteqs}}{\astexp \formeq \bstexp} $}
    \insertBetweenHyps{\hspace*{0em}}      
    \RightLabel{\LCoindProofi{n} %
                                 \parbox[c]{\widthof{(where $\Gamma = \setexp{ \cstexpi{1} \formeq \dstexpi{1}, \ldots, \cstexpi{n} \formeq \dstexpi{n} }$, and}}
                                           {\small 
                                            (where $\Gamma = \setexp{ \cstexpi{1} \formeq \dstexpi{1}, \ldots, \cstexpi{n} \formeq \dstexpi{n} }$, and 
                                            \\\phantom{(}%
                                            $\aLLEECoProofoverof{\thplus{\milnersysmin}{\aseteqs}}{\astexp \formeq \bstexp}$
                                            is a \LLEEwitnessed\ 
                                            \\\phantom{(}%
                                            coinductive proof of $\astexp \formeq \bstexp$ ov.\ $\thplus{\milnersysmin}{\aseteqs}$)}}
    \BinaryInfC{$ \astexp \formeq \bstexp $}           
    \DisplayProof
  \end{equation*}
  where $n\in\nat$ (including $n=0$),
  and the $(n+1)$-th premise of an instance of $\LCoindProofi{n}$ consists of a \LLEEwitnessed\ coinductive proof $\aCoProof$ of $\astexp \formeq \bstexp$ 
  over $\milnersysmin$ plus the formulas of the other premises. 
  By the \emph{proof system \CCover{\actions} for combining coinductive proofs (over extensions of $\milnersysminover{\actions}$) between star expressions over $\actions$}
  we mean the analogous system with a rule $\CoindProofi{n}$
  whose $(n+1)$-th premise is a coinductive proof $\aCoProof$ of $\astexp \formeq \bstexp$ 
  over $\milnersysminover{\actions}$ plus the formulas of the other premises
  that establishes $\astexp \coindproofeqin{(\thplus{\milnersysmin}{\aseteqs})} \bstexp$ 
  (thus here coinductive proofs do not need to be \LLEEwitnessed). 
  Keeping $\actions$ implicit, we write \CLC\ and \CC\ for \CLCover{\actions} and \CCover{\actions}, respectively. 
  Note that \CLC\ and \CC\ do not contain 
    the rules of $\eqlogic$ nor any axioms;
  instead, derivations have to start with 0-premise instances of $\LCoindProofi{0}$~or~$\CoindProofi{0}$.
\end{defi}

We now define a coinductively motivated variant $\coindmilnersys$ of Milner's proof system $\milnersys$.
In order to obtain $\coindmilnersys$ we drop the fixed-point rule $\RSPstar$ from $\milnersys$,
obtaining $\milnersysmin$, and then add
a rule each of whose instances use, as a premise, an entire \LLEEwitnessed\ coinductive proof over $\milnersysmin$ 
  and equations in other premises. 

\begin{samepage}
\begin{defi}[proof systems $\coindmilnersys$, $\coindmilnersysone$, $\coindmilnersysbar$]\label{def:coindmilnersys}
  Let $\actions$ be a set of actions.\nopagebreak[4]
  
  The proof system $\coindmilnersysover{\actions}$,
    the \emph{coinductive variant of $\milnersysover{\actions}$},
  has the same \emph{formulas} as $\CLCover{\actions}$ (formal equations and coinductive proofs), 
  its \emph{axioms} are those of $\milnersysminover{\actions}$,
  and its \emph{rules} are those of $\eqlogicover{\actions}$,
    plus the \emph{rule scheme} $\family{\LCoindProofi{n}}{n\in\nat}$ from $\CLCover{\actions}$.
  By $\coindmilnersysoneover{\actions}$ we mean
    \emph{the simple coinductive variant of $\milnersysover{\actions}$},
  in which only the rule $\LCoindProofi{1}$ of $\CLCover{\actions}$ is added to the rules and axioms of $\milnersysminover{\actions}$.
  By $\coindmilnersysbarover{\actions}$ we mean the variant of $\coindmilnersysover{\actions}$
    in which  the more general \emph{rule scheme} $\family{\CoindProofi{n}}{n\in\nat}$ from $\CCover{\actions}$
    is~used~instead.
    
  We again permit to write
             $\coindmilnersys$, $\coindmilnersysone$, $\coindmilnersysbar$
         for $\coindmilnersysover{\actions}$, $\coindmilnersysoneover{\actions}$, and $\coindmilnersysbarover{\actions}$, respectively.
\end{defi}
\end{samepage}

\begin{lem}\label{lem:easy:rels:coindproofsystems}
  The following theorem subsumption and equivalence statements hold:
  \begin{enumerate}[(i)]
    \item{}\label{it:1:lem:easy:rels:coindproofsystems}
      $ \coindmilnersysone
          \isthmsubsumedby
        \coindmilnersys $. 
    \item{}\label{it:2:lem:easy:rels:coindproofsystems}
      $ \CLC \thmequiv \coindmilnersys $. 
    \item{}\label{it:3:lem:easy:rels:coindproofsystems}
      $ \CC \thmequiv \coindmilnersysbar $.
  \end{enumerate}
\end{lem}

\begin{proof}
  Statement~(i) is due to $\coindmilnersysone \subsystem \coindmilnersys$,
  as $\coindmilnersysone$ is a subsystem of $\coindmilnersys$ that arises by restricting the rule scheme $\family{\LCoindProofi{i}}{i\in\nat}$
  to the single rule $\LCoindProofi{1}$.

  For \eqref{it:2:lem:easy:rels:coindproofsystems}, 
  $\CLC \isthmsubsumedby \coindmilnersys$ follows from $\CLC \subsystem \coindmilnersys$.
  For showing the converse implication, $\CLC \thmsubsumes \coindmilnersys$, 
    it suffices to transform an arbitrary derivation in \coindmilnersys\ into a derivation in \CLC.
  For this purpose, all instances of axioms and rules of $\milnersysmin$ have to be 
    eliminated from derivations in $\coindmilnersys$,
    keeping only instances of $\LCoindProofi{k}$ for $k\in\nat$.
  This can be done by extending the equation premises of rules $\LCoindProofi{k}$ if required.
  For example, a derivation with a bottommost instance of $\LCoindProofi{1}$ in which the subderivation $\aDeriv$ does not contain any instances  
  of $\LCoindProofi{k}$ for $k\in\nat$ below the conclusions $\cstexpi{1} \formeq \dstexpi{1}$, \ldots, $\cstexpi{m} \formeq \dstexpi{m}$
  of instances of $\LCoindProofi{k_1}$, \ldots, $\LCoindProofi{k_m}\,$:
  \begin{flushleft}
    \hspace*{-2.5ex}
    \mbox{%
      \AxiomC{$ \aDerivi{1} $}
      \noLine
      \UnaryInfC{$( \cstexpi{1} \formeq \dstexpi{1} )$}
      \AxiomC{$ \aDerivi{m} $}
      \noLine
      \UnaryInfC{$( \cstexpi{m} \formeq \dstexpi{m} )$}
      \insertBetweenHyps{\ldots}
      \noLine
      \BinaryInfC{$\aDeriv$}
      \noLine
      \UnaryInfC{$ \cstexp \formeq \dstexp $}
      \AxiomC{$ \aLLEECoProofoverof{}{\astexp \formeq \bstexp} $}
      \insertBetweenHyps{\hspace*{-8.5ex}}
      \RightLabel{$\LCoindProofi{1}$}
      \BinaryInfC{$ \astexp \formeq \bstexp $}
      \DisplayProof}
    \hspace*{-13ex}\raisebox{0.8ex}{$\Longrightarrow$}\hspace*{1ex}
    \raisebox{2.5ex}{
      \AxiomC{$ \aDerivi{1} $}
      \noLine
      \UnaryInfC{$( \cstexpi{1} \formeq \dstexpi{1} ) \: \ldots \hspace*{-1ex}$}
      \AxiomC{$ \aDerivi{m} $}
      \noLine
      \UnaryInfC{$( \cstexpi{m} \formeq \dstexpi{m} )$}
      \AxiomC{$ \aLLEECoProofaccoverof{}{ \astexp \formeq \bstexp } $}
      \insertBetweenHyps{}
      \RightLabel{$\LCoindProofi{m}$}
      \TrinaryInfC{$ \astexp \formeq \bstexp $}
      \DisplayProof}
  \end{flushleft}
  can be replaced, on the right, by a single instance of $\LCoindProofi{m}$,
  where $\aLLEECoProofaccoverof{}{ \astexp \formeq \bstexp }$ is, while formally the same as $\aLLEECoProofoverof{}{ \astexp \formeq \bstexp }$,
    now a \LLEEwitnessed\ coinductive proof over $\milnersysmin$ plus $\cstexpi{1} \formeq \dstexpi{1}$, \ldots, $\cstexpi{m} \formeq \dstexpi{m}$.
  The latter is possible because the derivation part $\aDeriv$ in $\milnersysmin$ implies that then $\cstexp \formeq \dstexp$ can be derived 
    from the assumptions in $\milnersysmin$ as well.
    
  Statement~\eqref{it:3:lem:easy:rels:coindproofsystems} can be shown entirely analogously as statement~\eqref{it:2:lem:easy:rels:coindproofsystems}.
\end{proof}

\begin{rem}[completeness of 
                            $\CC$, $\coindmilnersysbar$, $\milnersysaccbar$]
  The proof systems 
                    $\CC$ and $\coindmilnersysbar$,
    as well as the variant $\milnersysaccbar$ of Milner's system with the general unique solvability principle \USP\
  are complete for bisimilarity of star expressions.
  This can be established along Salomaa's completeness proof
    for his inference system for language equality of regular expressions~\cite{salo:1966}, 
    by an argument that we can outline as follows.
  Given star expressions $e$ and $f$ with $\chartof{e} \bisim \chartof{f}$, 
    $e$ and $f$ can be shown to be principal values of \provablein{\milnersysmin} solutions of $\chartof{e}$ and $\chartof{f}$, respectively
    (by a lemma for the chart interpretation similar to Lem.~\ref{lem:onechart-int:milnersysmin:solvable}). 
  These solutions can be transferred to the (\sonefree) product chart~$\achart$ of $\chartof{e}$ and $\chartof{f}$,
    with $e$ and $f$ as principal values
      of \provablein{\milnersysmin} solutions $\saeqfuni{1}$ and $\saeqfuni{2}$ of $\achart$, respectively. 
    From this we obtain a (not necessarily \LLEEwitnessed) coinductive proof $\pair{\achart}{\saeqfun}$ of $e = f$ over $\milnersysmin$. 
    It follows that $e = f$ is provable in 
                                           $\CC$, and in $\coindmilnersysbar$.
    Now since the correctness conditions for the \provablein{\milnersysmin} solutions $\saeqfuni{1}$ and $\saeqfuni{2}$ of $\achart$
      at each of the vertices of $\achart$
      together form a guarded system of linear equations to which the rule \USP\ can be applied (as $\achart$ is \sonefree),
      we obtain that $e \formeq f$ is also provable in $\milnersysaccbar$.  
\end{rem}

\section{From LLEE-witnessed coinductive proofs to Milner's system}
  \label{coindmilnersys:2:milnersys}

In this section we show that every \LLEEwitnessed\ coinductive proof over $\milnersysmin$ 
of an equation can also be established by a proof in Milner's system $\milnersys$. 
As a consequence we show that the coinductive version $\coindmilnersys$ of $\milnersys$ is \theoremsubsumed\ by $\milnersys$.
We obtain the statements in this section by adapting results in \cite[Sect.~5]{grab:fokk:2020a,grab:fokk:2020b}
 from \LLEEcharts\ to \LLEEonecharts. 



The hierarchical loop structure of a \onechart~$\aonecharthat$
  with \LLEEwitness~$\aonecharthat$ 
  facilitates the extraction of a \provablein{\milnersysmin} solution of $\aonechart$
  (see Lem.~\ref{lem:extrsol:is:sol}), for the following reason. 
The behaviour at every vertex $\bvert$ in $\aonechart$ can be split 
  into an iteration part that is induced via the \loopentry\ transitions from $\bvert$ in $\aonecharthat$
                                         (which induce loop \subonecharts\ with inner loop \subonecharts\ whose behaviour can be synthesized recursively),
  and an exit part that is induced via the body transitions from $\bvert$ in $\aonecharthat$. 
This idea permits to define (Def.~\ref{def:extrsoluntil:extrsol}) an `extraction function'\vspace*{-2.5pt} $\sextrsolof{\aonecharthat}$ of $\aonecharthat$
  from a `relative extraction function' $\sextrsoluntilof{\aonecharthat}$ of $\aonecharthat$
  whose values $\extrsoluntilof{\aonecharthat}{\bvert}{\avert}$ capture the behaviour at $\bvert$ 
  in a loop \subonechart\ at $\avert$ until $\avert$ is reached.
By the same idea the fixed-point rule $\RSPstar$ 
  can be used for showing that any two \provablein{\milnersys} solutions of $\aonechart$ are
  \provablyin{\milnersys} equal (Lem.~\ref{lem:sols:provably:equal:LLEE}).
We provide concrete examples for these statements that hint at their proofs 
  (Ex.~\ref{ex:extraction}, and Ex.~\ref{ex:extraction:proof}),
but we defer proofs to the appendix due to the similarity with \cite{grab:fokk:2020a,grab:fokk:2020b}.


\begin{defi}[(relative) extraction function]\label{def:extrsoluntil:extrsol}
  Let $\aonechart = \tuple{\verts,\actions,\sone,\start,\transs,\termexts}$ be a \onechart\ with guarded \LLEEwitness~$\aonecharthat$. 
  The \emph{extraction function 
            $\sextrsolof{\aonecharthat} \funin \verts \to \StExpsover{\actions}$ 
            of $\aonecharthat$}
  is defined from the 
  \emph{relative extraction function  
        $\sextrsoluntilof{\aonecharthat} \funin \descsetexp{\pair{\bvert}{\avert}}
                                                        {\bvert,\avert\in\vertsof{\aonechart},\,\bvert\convdescendsinlooptosc\avert} $
        of $\aonecharthat$}
  as follows, for $\bvert,\avert\in\verts$:\vspace{-1.5ex}
  \begin{align*}
    \extrsoluntilof{\aonecharthat}{\bvert}{\avert}
      & {} \defdby
        \begin{cases}
          \stexpone
            & \text{if $\bvert = \avert$,}
          \\[-0.5ex] 
          \stexpprod{\displaystyle
            \stexpit{\Bigl(
                       \sum_{i=1}^{n}
                         \stexpprod{\aacti{i}}{\extrsoluntilof{\aonecharthat}{\bverti{i}}{\bvert}}
                     \Bigr)}
                     }{\displaystyle
            \Bigl(
              {\sum_{i=1}^{m}
                 \stexpprod{\boneacti{i}}{\extrsoluntilof{\aonecharthat}{\cverti{i}}{\avert}}}
            \Bigr)              
                       }
            & \text{if $\bvert \convdescendsinloopto \avert$,}
        \end{cases}
      \displaybreak[0]\\[-0.5ex]  
      \extrsolof{\aonecharthat}{\bvert}
        & {} \defdby
          \stexpprod{%
            \stexpit{\Bigl(
                       \sum_{i=1}^{n}
                         \stexpprod{\aacti{i}}{\extrsoluntilof{\aonecharthat}{\bverti{i}}{\bvert}}
                     \Bigr)}
                     }{
            \Bigl(
              \stexpsum{\terminatesconstof{\aonechart}{\bvert}}
                       {\sum_{i=1}^{m}
                          \stexpprod{\boneacti{i}}{\extrsolof{\aonecharthat}{\cverti{i}}}}
            \Bigr)              
                       } \punc{,}
     \\[-0.5ex]
     & \hspace*{-8ex} \text{provided: } 
       \transitionsinfrom{\aonecharthat}{\bvert}
         =
       \begin{aligned}[t]
         &
         \descsetexpbig{ \bvert \lti{\aacti{i}}{\looplab{\alabi{i}}} \bverti{i} }
                       { \alabi{i}\in\natplus,\,  i\in\setexp{1,\ldots,n} }
           {\cup} 
         \descsetexpbig{ \bvert \lti{\boneacti{i}}{\looplab{\bodylab}} \cverti{i} }
                       { i\in\setexp{1,\ldots,m} } \text{,}
       \end{aligned}  
     \\[-0.5ex]
     & \hspace*{-8ex}\text{induction for $\sextrsoluntilof{\aonecharthat}$ on: }
     \begin{aligned}[t]  
       \pair{\bverti{1}}{\averti{1}}
          \lexspo
        \pair{\bverti{2}}{\averti{2}}
          \;\funin\: \Longleftrightarrow\:
            \averti{1} \convdescendsinlooptotc \averti{2}
              \logor
            (\, \averti{1} = \averti{2}
                  \logand 
                \bverti{1} \convredtci{\bodylab} \bverti{2} \,) \punc{,}
     \end{aligned}                        
     \\[-0.75ex]
     & \hspace*{-8ex}\text{induction for $\sextrsolof{\aonecharthat}$ on the strict partial order $\sconvredtci{\bodylab}$ (see Lem.~\ref{lem:descsteps:bodysteps:wf})}
         \punc{,}
  \end{align*}
  where $\slexspo$ is a well-founded strict partial order due to Lem.~\ref{lem:descsteps:bodysteps:wf}.
  The choice of the list representations of \actiontarget\ sets of $\aonecharthat$ 
  changes these definition 
                           only up to provability~in~\ACI.
\end{defi}

\begin{exa}\label{ex:extraction}\mbox{}%
  \begin{figure}[tbp]
\begin{center}
  \begin{equation*}
    \hspace*{-3ex}
    \begin{aligned}[c]
      \scalebox{1.3}{\begin{tikzpicture}\renewcommand{\stexpprod}[2]{{#1}\hspace*{1pt}{\sstexpprod}\hspace*{1pt}{#2}}
  
\matrix[anchor=center,row sep=1cm,column sep=0.75cm,
        every node/.style={draw,very thick,circle,minimum width=2.5pt,fill,inner sep=0pt,outer sep=2pt}] {
  \node(v11){};
    & & \node(v21){};
  \\                 
  \node(v1){};
    & & \node(v2){};
  \\
    & \node[chocolate](v){};
  \\
  };   
\path (v) ++ (0cm,-1cm) node(label){\Large $\aonechart$, $\aonecharthat$};

\path (v11) ++ (0cm,0.25cm) node{\small $\averti{11}$};
\draw[->,thick,densely dotted,out=180,in=180,distance=0.75cm](v11) to (v1);

\path (v1) ++ (-0.225cm,-0.25cm) node{\small $\averti{1}$};
\draw[->,thick,darkcyan,shorten >= 0pt]
  (v1) to node[left,pos=0.25,xshift=0.075cm]{\small $\black{\aact}$} node[left,pos=0.6,xshift=0.075cm]{\small $\loopnsteplab{1}$} (v11);
\draw[->,thick,densely dotted,out=-90,in=180,distance=0.5cm,shorten >=2pt](v1) to (v);
\draw[->,shorten <= 0pt,shorten >= 0pt] (v1) to node[above]{$\bact$} (v21); 

\path (v21) ++ (0cm,0.25cm) node{\small $\averti{21}$};
\draw[->,thick,densely dotted,out=-0,in=0,distance=0.75cm](v21) to (v2);  
  
\path (v2) ++ (0.25cm,-0.25cm) node{\small $\averti{2}$};
\draw[->,thick,darkcyan,shorten >= 0pt]
  (v2) to node[right,pos=0.25,xshift=-0.075cm]{\small $\black{\bact}$} node[right,pos=0.6,xshift=-0.075cm]{\small $\loopnsteplab{1}$} (v21);
\draw[->,thick,densely dotted,out=-90,in=0,distance=0.5cm,shorten >= 2pt](v2) to (v);

\draw[<-,very thick,>=latex,chocolate,shorten <= 2pt](v) -- ++ (-90:0.55cm);   
\draw[thick,chocolate] (v) circle (0.12cm);
\path (v) ++ (0.25cm,-0.3cm) node{\small $\start$};
\draw[->,thick,darkcyan,shorten >= 0.175cm,shorten <= 2pt] 
  (v) to node[left,pos=0.36,xshift=0.075cm]{\small $\black{\aact}$} node[right,pos=0.4,xshift=-0.075cm,yshift=1pt]{\small $\loopnsteplab{2}$}  (v11);
\draw[->,thick,darkcyan,shorten >= 0.175cm,shorten <= 2pt] 
  (v) to node[right,pos=0.36,xshift=-0.05cm]{\small $\black{\bact}$} node[left,pos=0.6,xshift=0.075cm,yshift=1pt]{\small $\loopnsteplab{2}$} (v21);
  
\end{tikzpicture} }
    \end{aligned}
    \hspace*{1ex}
    \begin{aligned}[c]
      \extrsoluntilof{\aonecharthat}{\averti{21}}{\averti{2}}
        & \;\parbox[t]{\widthof{$\milnersysmineq$}}{$\defdby$}\,
      \stexpprod{\stexpit{\stexpzero}}  
                {(\stexpprod{\stexpone}{\extrsoluntilof{\aonecharthat}{\averti{2}}{\averti{2}}})}
        \,\synteq\,
      \stexpprod{\stexpit{\stexpzero}}{(\stexpprod{\stexpone}{\stexpone})}
        \,\milnersysmineq\,
      \stexpone
      \\[-0.1ex]
      \extrsoluntilof{\aonecharthat}{\averti{2}}{\start}
        & \;\parbox[t]{\widthof{$\milnersysmineq$}}{$\defdby$}\,
      \stexpprod{\stexpit{(\stexpprod{\bact}{\extrsoluntilof{\aonecharthat}{\averti{21}}{\averti{2}}})}}
                {(\stexpprod{\stexpone}{\extrsoluntilof{\aonecharthat}{\start}{\start}})}
        \,\milnersysmineq\,  
      \stexpit{\bact}
      \\[-0.1ex]
      \extrsoluntilof{\aonecharthat}{\averti{21}}{\start}
        & \;\parbox[t]{\widthof{$\milnersysmineq$}}{$\defdby$}\,
      \stexpprod{\stexpit{\stexpzero}}  
                {(\stexpprod{\stexpone}{\extrsoluntilof{\aonecharthat}{\averti{2}}{\start}})}
        \,\milnersysmineq\,
      \stexpprod{\stexpone}{\stexpit{\bact}}
        \,\milnersysmineq\,
      \stexpit{\bact}
      \\[-0.1ex]
      \extrsoluntilof{\aonecharthat}{\averti{11}}{\averti{1}}
        & \;\parbox[t]{\widthof{$\milnersysmineq$}}{$\defdby$}\,
      \stexpprod{\stexpit{\stexpzero}}
                {(\stexpprod{\stexpone}{\extrsoluntilof{\aonecharthat}{\averti{1}}{\averti{1}}})}
        \,\synteq\,
      \stexpprod{\stexpit{\stexpzero}}
                {(\stexpprod{\stexpone}{\stexpone})}
        \,\milnersysmineq\,
      \stexpone
      \\[-0.1ex]
      \extrsoluntilof{\aonecharthat}{\averti{1}}{\start}
        & \;\parbox[t]{\widthof{$\milnersysmineq$}}{$\defdby$}\,
      \stexpprod{\stexpit{(\stexpprod{\aact}{\extrsoluntilof{\aonecharthat}{\averti{11}}{\start}})}}
                {(\stexpsum{\stexpprod{\bact}{\extrsoluntilof{\aonecharthat}{\averti{21}}{\start}}}
                           {\stexpprod{\stexpone}{\extrsoluntilof{\aonecharthat}{\start}{\start}}})}
      \\[-0.5ex] 
        & \;\parbox[t]{\widthof{$\milnersysmineq$}}{$\milnersysmineq$}\,
      \stexpprod{\stexpit{\aact}}
                {(\stexpsum{\stexpprod{\bact}{\stexpit{\bact}}}
                           {\stexpone})}
        \,\milnersysmineq\,
      \stexpprod{\stexpit{\aact}}{\stexpit{\bact}} 
      \\[-0.1ex]
      \extrsoluntilof{\aonecharthat}{\averti{11}}{\start}
        & \;\parbox[t]{\widthof{$\milnersysmineq$}}{$\defdby$}\,
      \stexpprod{\stexpit{\stexpzero}}  
                {(\stexpprod{\stexpone}{\extrsoluntilof{\aonecharthat}{\averti{1}}{\start}})}
        \,\milnersysmineq\,
      \stexpprod{\stexpit{\aact}}{\stexpit{\bact}}
      \\[-0.1ex]
      \extrsolof{\aonecharthat}{\start}
        & \;\parbox[t]{\widthof{$\milnersysmineq$}}{$\defdby$}\,
      \stexpprod{\stexpit{(\stexpsum{\stexpprod{\aact}{\extrsoluntilof{\aonecharthat}{\averti{11}}{\start}}}
                                    {\stexpprod{\bact}{\extrsoluntilof{\aonecharthat}{\averti{21}}{\start}}})}}
                {\stexpone}
      \\[-0.5ex]
        & \;\parbox[t]{\widthof{$\milnersysmineq$}}{$\milnersysmineq$}\,
     \stexpit{(\stexpsum{\stexpprod{\aact}{(\stexpprod{\stexpit{\aact}}{\stexpit{\bact}})}}
                        {\stexpprod{\bact}{\stexpit{\bact}}})}
    \end{aligned}
  \end{equation*}     
\end{center}
    \vspace*{-2.5ex}
  \caption{\protect\label{fig:ex:extraction}%
           Extraction of the principal value $\protect\extrsolof{\protect\aonecharthat}{\protect\start}$ of a \protect\provablein{\protect\milnersysmin} solution~$\protect\sextrsolof{\protect\aonecharthat}$ 
           from\protect\vspace*{-2pt} the \protect\LLEEwitness~$\protect\acharthat$ used for the coinductive proof 
           in Ex.~\protect\ref{ex:2:LLEEcoindproof}.
           We shorten solution values by using axioms~of~\protect\milnersysmin.
           }

\end{figure}      %
  We consider the \onechart~$\aonechart$,
  and the \LLEEwitness~$\aonecharthat$ of $\aonechart$,
  in the \LLEEwitnessed\ coinductive proof $\aCoProof = \pair{\aonechart}{\saeqfun}$ of 
  $\stexpit{(\stexpprod{\stexpit{\aact}}
                       {\stexpit{\bact}})}
     \formeq
   \stexpit{(\stexpsum{\aact}{\bact})}$
  in Ex.~\ref{ex:2:LLEEcoindproof}.
  We detail in Fig.~\ref{fig:ex:extraction} the process of computing
  the principal value $\extrsolof{\aonecharthat}{\start}$ of the
  extraction function $\sextrsolof{\aonecharthat}$ of $\aonecharthat$. 
  The statement of Lem.~\ref{lem:extrsol:is:sol} below will guarantee that $\sextrsolof{\aonecharthat}$
  is a \provablein{\milnersysmin} solution of $\aonechart$.
\end{exa}

\begin{lem}\label{lem:lem:extrsol:is:sol}
  Let $\aonechart$ be a weakly guarded \LLEEonechart\ with guarded \LLEEwitness~$\aonecharthat$. 
  
  Then 
  $\extrsolof{\aonecharthat}{\bvert}
     \milnersysmineq 
   \stexpprod{\extrsoluntilof{\aonecharthat}{\bvert}{\avert}}
             {\extrsolof{\aonecharthat}{\avert}}$
  holds 
  for all vertices $\bvert,\avert\in\vertsof{\aonechart}$ such that $\bvert \convdescendsinlooptosc \avert$.
\end{lem}

\begin{lem}[extracted function is provable solution]\label{lem:extrsol:is:sol}
  Let $\aonechart$ 
      be a w.g.\ \LLEEonechart\
  with guarded \LLEEwitness~$\aonecharthat$.
  Then the extraction function $\sextrsolof{\aonecharthat}$ of $\aonecharthat$ 
  is a \provablein{\milnersysmin} solution~of~$\aonechart$. 
\end{lem}

\begin{lem}\label{lem:lem:sols:provably:equal:LLEE}
  Let $\aonechart$ be a \onechart\ $\aonechart$ with guarded \LLEEwitness~$\aonecharthat$.
  Let $\asys$ be an \eqlogicbased\ proof system over $\StExpover{\actions}$ 
    such that $\ACI \subsystem \asys \isthmsubsumedby \milnersys$.
  
  Let $\sasol \funin \vertsof{\aonechart} \to \StExpover{\actions}$
      be an \provablein{\asys} solution of $\aonechart$.
  Then 
  $\asol{\bvert}
     \eqin{\milnersys}
   \stexpprod{\extrsoluntilof{\aonecharthat}{\bvert}{\avert}}
             {\asol{\avert}}$ \mbox{}
  holds for all vertices $\bvert,\avert\in\vertsof{\aonechart}$ with $\bvert \convdescendsinlooptosc \avert$.
\end{lem}

For an \eqlogicbased\ proof system $\asys$ over $\StExpover{\actions}$ we say
that two star expression functions $\sasoli{1},\sasoli{2} \funin \verts\to\StExpover{\actions}$
are \emph{\provablyin{\asys} equal} if $\asoli{1}{\avert} \eqin{\asys} \asoli{2}{\avert}$ holds for all $\avert\in\verts$. 

\begin{samepage}
\begin{lem}[provable equality of solutions of LLEE-1-charts]\label{lem:sols:provably:equal:LLEE}
  Let $\aonechart$ be a guarded \LLEEonechart,
  and let $\asys$ be an \eqlogicbased\ proof system over $\StExpover{\actions}$ such that $\ACI \subsystem \asys \isthmsubsumedby \milnersys$.\nopagebreak[4]
  
  Then any two \provablein{\asys} solutions of $\aonechart$ are \provablyin{\milnersys} equal.
\end{lem}
\end{samepage}

\begin{prop}\label{prop:LLEEcoindproofeq:impl:milnersyseq}
  For every \eqlogicbased\ proof system $\asys$ over $\StExpover{\actions}$ with $\ACI \subsystem \asys \isthmsubsumedby \milnersys$, 
  provability by \LLEEwitnessed\ coinductive proofs over $\asys$ implies derivability~in~$\milnersys$:\vspace{-1ex}
  \begin{equation}\label{eq:prop:LLEEcoindproofeq:impl:milnersyseq}
      \bigl(\;
        \astexpi{1}
          \LLEEcoindproofeqin{\asys}
        \astexpi{2}
            \quad\;\Longrightarrow\quad\;
        \astexpi{1}
          \milnersyseq
        \astexpi{2} 
      \;\bigr)
      \qquad \text{for all $\astexpi{1},\astexpi{2}\in\StExpover{\actions}$.} 
  \end{equation}
\end{prop}

\begin{proof}\vspace{-0.5ex}
  For showing \eqref{eq:prop:LLEEcoindproofeq:impl:milnersyseq}, 
  let $\astexp,\bstexp\in\StExpover{\actions}$ be such that 
  $\astexp \LLEEcoindproofeqin{\asys} \bstexp$.
  Then there is a \LLEEwitnessed\ coinductive proof $\aCoProof = \pair{\aonechart}{\saeqfun}$ of $\astexpi{1} \formeq \astexpi{2}$ over $\asys$.
  Then $\aonechart$ 
                    is a \LLEEonechart, and there are 
  \provablein{\asys} solutions $\saeqfuni{1},\saeqfuni{2} \funin \vertsof{\aonechart} \to \StExpover{\actions}$ of $\aonechart$
    such that $\astexpi{1} \synteq \aeqfuni{1}{\start}$ and $\astexpi{2} \synteq \aeqfuni{2}{\start}$. 
  Then $\saeqfuni{1}$ and $\saeqfuni{2}$ are \provablyin{\milnersys} equal by Lem.~\ref{lem:sols:provably:equal:LLEE}.
  As a consequence we find
    $\astexpi{1} 
       \synteq
     \aeqfuni{1}{\start}
       \milnersyseq
     \aeqfuni{2}{\start}
       \synteq
     \astexpi{2}$,
  and hence $\astexpi{1} \milnersyseq \astexpi{2}$.           
\end{proof}

\begin{exa}\label{ex:extraction:proof}\mbox{}%
  \begin{figure}[tbp]
\begin{center}
  $
  \hspace*{-3ex}%
  \begin{aligned}
    \scalebox{0.9}{\begin{tikzpicture}\renewcommand{\stexpprod}[2]{{#1}\hspace*{1pt}{\sstexpprod}\hspace*{1pt}{#2}}
  
\matrix[anchor=north,row sep=1.25cm,column sep=0cm,
        ] {
  \node(v11){$\aeqfuni{i}{\averti{11}}$};
    & & \node(v21){$\aeqfuni{i}{\averti{21}}$};
  \\                 
  \node(v1){$\aeqfuni{i}{\averti{1}}$};
    & & \node(v2){$\aeqfuni{i}{\averti{2}}$};
  \\
    & \node(v){$\aeqfuni{i}{\start}$};
  \\
  };   
\path (v) ++ (0cm,
                  +4.75cm) node(label){\Large $\aonechart$, $\aonecharthat$};

\draw[->,thick,densely dotted,out=180,in=165,distance=0.75cm](v11) to (v1);

\draw[->,thick,darkcyan,shorten >= 0pt]
  (v1) to node[left,pos=0.25,xshift=0.075cm]{\small $\black{\aact}$} node[left,pos=0.6,xshift=0.075cm]{\small $\loopnsteplab{1}$} (v11);
\draw[->,thick,densely dotted,out=-90,in=165,distance=0.5cm](v1) to (v);
\draw[->,shorten <= 0pt,shorten >= 0pt] (v1) to node[above]{$\bact$} (v21); 

\draw[->,thick,densely dotted,out=-0,in=15,distance=0.75cm](v21) to (v2);  
  
\draw[->,thick,darkcyan,shorten >= 0pt]
  (v2) to node[right,pos=0.25,xshift=-0.075cm]{\small $\black{\bact}$} node[right,pos=0.6,xshift=-0.075cm]{\small $\loopnsteplab{1}$} (v21);
\draw[->,thick,densely dotted,out=-90,in=15,distance=0.5cm](v2) to (v);

\draw[<-,very thick,>=latex,chocolate,shorten <=2.5pt](v) -- ++ (-90:0.8cm); 
\draw[thick,chocolate,double] (v) ellipse (0.475cm and 0.3cm);   
\draw[->,thick,darkcyan,shorten >= 0.175cm,shorten <= 3pt] 
  (v) to node[left,pos=0.36,xshift=0.075cm]{\small $\black{\aact}$} node[right,pos=0.4,xshift=-0.075cm,yshift=1pt]{\small $\loopnsteplab{2}$}  (v11);
\draw[->,thick,darkcyan,shorten >= 0.175cm,shorten <= 3pt] 
  (v) to node[right,pos=0.36,xshift=-0.05cm]{\small $\black{\bact}$} node[left,pos=0.6,xshift=0.075cm,yshift=1pt]{\small $\loopnsteplab{2}$} (v21);
  
\end{tikzpicture} }
  \end{aligned}
  \hspace*{1ex}
  \begin{aligned}[c]
    \aeqfuni{i}{\averti{21}}
      & \;\parbox[t]{\widthof{$\milnersysmineq$}}{$\eqinsol{\milnersysmin}$}\,
    \stexpprod{\stexpone}{\aeqfuni{i}{\averti{2}}} 
      \,\milnersysmineq\,
    \aeqfuni{i}{\averti{2}} 
    \\[-0.75ex]
      & \phantom{\;\parbox[t]{\widthof{$\milnersysmineq$}}{$\eqinsol{\milnersysmin}$}}\,
        \text{($\eqinsol{\milnersysmin}$ means use of 
               `is \provablein{\milnersysmin} solution')}
    \\
    \aeqfuni{i}{\averti{2}}
      & \;\parbox[t]{\widthof{$\milnersysmineq$}}{$\eqinsol{\milnersysmin}$}\,
    \stexpsum{\stexpprod{\bact}{\aeqfuni{i}{\averti{21}}}}
             {\stexpprod{\stexpone}{\aeqfuni{i}{\start}}} 
      \,\milnersysmineq\,
    \stexpsum{\stexpprod{\bact}{\aeqfuni{i}{\averti{2}}}}
             {\aeqfuni{i}{\start}}
    \\[-0.5ex]
    & \;\parbox[t]{\widthof{$\milnersysmineq$}}{$\Downarrow${\text{~~applying $\RSPstar$}}}
    \\[-0.5ex]
    \aeqfuni{i}{\averti{2}}
      & \;\parbox[t]{\widthof{$\milnersysmineq$}}{$\milnersyseq$}\,
    \stexpprod{\stexpit{\bact}}
              {\aeqfuni{i}{\start}}
    \\[0.5ex]
    \aeqfuni{i}{\averti{11}}
      & \;\parbox[t]{\widthof{$\milnersysmineq$}}{$\eqinsol{\milnersysmin}$}\,
    \stexpprod{\stexpone}{\aeqfuni{i}{\averti{1}}} 
      \,\milnersysmineq\,
    \aeqfuni{i}{\averti{1}} 
    \\
    \aeqfuni{i}{\averti{1}}
      & \;\parbox[t]{\widthof{$\milnersysmineq$}}{$\milnersysmineq$}\,
    \stexpsum{\stexpprod{\aact}{\aeqfuni{i}{\averti{11}}}}
             {\stexpsum{\stexpprod{\bact}{\aeqfuni{i}{\averti{21}}}}
                       {\stexpprod{\stexpone}{\aeqfuni{i}{\start}}}}
    \\[-0.25ex]      
      & \;\parbox[t]{\widthof{$\milnersysmineq$}}{$\milnersyseq$}\,
    \stexpsum{\stexpprod{\aact}{\aeqfuni{i}{\averti{1}}}}
             {\stexpprod{(\stexpsum{\stexpprod{\bact}{\stexpit{\bact}}}
                                   {\stexpone})}
                        {\aeqfuni{i}{\start}}}
    \\[-0.25ex]      
      & \;\parbox[t]{\widthof{$\milnersysmineq$}}{$\milnersysmineq$}\,
    \stexpsum{\stexpprod{\aact}{\aeqfuni{i}{\averti{1}}}}
             {\stexpprod{\stexpit{\bact}}
                        {\aeqfuni{i}{\start}}}
    \\[-0.5ex]
    & \;\parbox[t]{\widthof{$\milnersysmineq$}}{$\Downarrow${\text{~~applying $\RSPstar$}}}
    \\[-0.5ex]
    \aeqfuni{i}{\averti{1}}
      & \;\parbox[t]{\widthof{$\milnersysmineq$}}{$\milnersyseq$}\,
    \stexpprod{\stexpit{\aact}}
              {(\stexpprod{\stexpit{\bact}}{\aeqfuni{i}{\start}})}
      \,\milnersysmineq\,          
    \stexpprod{(\stexpprod{\stexpit{\aact}}{\stexpit{\bact}})}{\aeqfuni{i}{\start}}  
  \end{aligned}
  $
  \\[0.75ex]
  $
  \begin{aligned}
    \aeqfuni{i}{\start}
      & \;\parbox[t]{\widthof{$\milnersysmineq$}}{$\eqinsol{\milnersysmin}$}\,
    1 +  
    \stexpsum{\stexpprod{\aact}{\aeqfuni{i}{\averti{11}}}}
             {\stexpprod{\bact}{\aeqfuni{i}{\averti{21}}}} 
      \,\milnersysmineq\,
    1 +  
    \stexpsum{\stexpprod{\aact}{\aeqfuni{i}{\averti{1}}}}
             {\stexpprod{\bact}{\aeqfuni{i}{\averti{2}}}}   
    \\[-0.25ex]
      & \;\parbox[t]{\widthof{$\milnersysmineq$}}{$\milnersyseq$}\,
    \stexpprod{(\stexpsum{\stexpprod{\aact}{(\stexpprod{\stexpit{\aact}}{\stexpit{\bact}})}}
                         {\stexpprod{\bact}{\stexpit{\bact}}})}
              {\aeqfuni{i}{\start}}
      + 1          
    \\[-0.5ex]
    & \;\parbox[t]{\widthof{$\milnersysmineq$}}{$\Downarrow${\text{~~applying $\RSPstar$}}}
    \\[-0.5ex]
    \aeqfuni{i}{\start}
      & \;\parbox[t]{\widthof{$\milnersysmineq$}}{$\milnersyseq$}\,
    \stexpprod{\stexpit{(\stexpsum{\stexpprod{\aact}{(\stexpprod{\stexpit{\aact}}{\stexpit{\bact}})}}
                                  {\stexpprod{\bact}{\stexpit{\bact}}})}}
              {\stexpone}
      \,\milnersysmineq\,
    \stexpit{(\stexpsum{\stexpprod{\aact}{(\stexpprod{\stexpit{\aact}}{\stexpit{\bact}})}}
                       {\stexpprod{\bact}{\stexpit{\bact}}})}        
      \,\milnersysmineq\,
    \extrsolof{\aonecharthat}{\start}
  \end{aligned}
  $     
\end{center}
    \vspace*{-3ex}
\caption{\protect\label{fig:ex:extraction:proof}%
         Use of the \protect\LLEEwitness~$\protect\aonecharthat$ underlying the coinductive proof $\protect\pair{\protect\aonechart}{\protect\saeqfun}$ in Ex.~\protect\ref{ex:2:LLEEcoindproof}
         for showing that 
         the principal value $\protect\aeqfuni{i}{\protect\start}$ of the \protect\provablein{\protect\milnersysmin} solution $\protect\saeqfuni{i}$ for $i\protect\in\protect\setexp{1,2}$ 
         is \protect\provablyin{\protect\milnersys} is \protect\provablyin{\protect\milnersys} equal to the principal value $\protect\extrsolof{\protect\aonecharthat}{\protect\start}$ of the solution $\protect\sextrsolof{\protect\aonecharthat}$ 
         extracted from $\protect\aonecharthat$ (see Fig.~\protect\ref{fig:ex:extraction}). 
         }
\end{figure}      %
  We consider again the \LLEEwitnessed\ coinductive proof $\aCoProof = \pair{\aonechart}{\saeqfun}$ of 
  $\stexpit{(\stexpprod{\stexpit{\aact}}
                       {\stexpit{\bact}})}
     \formeq
   \stexpit{(\stexpsum{\aact}{\bact})}$
  in Ex.~\ref{ex:2:LLEEcoindproof}. 
  In Fig.~\ref{fig:ex:extraction:proof} we exhibit the extraction process of derivations in $\milnersys$
  of $\aeqfuni{1}{\start} \formeq \extrsolof{\aonechart}{\avert}$ 
  and $\aeqfuni{2}{\start} \formeq \extrsolof{\aonechart}{\avert}$
  from the \LLEEwitness~$\aonecharthat$ of $\aonechart$,
  which can be combined by $\eqlogic$ rules 
  to obtain a derivation in $\milnersys$ of
  $\stexpit{(\stexpprod{\stexpit{\aact}}{\stexpit{\bact}})} 
     \synteq       \aeqfuni{1}{\start} 
     \formeq \aeqfuni{2}{\start}
     \synteq       \stexpit{(\stexpsum{\aact}{\bact})}$.
\end{exa}


\begin{thm}\label{thm:coindmilnersys:isthmsubsumedby:milnersys}
  $\coindmilnersys \isthmsubsumedby \milnersys$. 
  Moreover, every derivation in $\coindmilnersys$ with conclusion $\astexp \formeq \bstexp$ 
    can be transformed effectively
  into a derivation in $\milnersys$ that has the same conclusion.
\end{thm}

\begin{proof}
  It suffices to show the transformation statement.
  This can be established by a straightforward induction on the depth of derivations in $\coindmilnersys$,
    in which the only \nontrivial\ case is the elimination of $\LCoindProofi{n}$ instances.
  For every instance of $\LCoindProofi{n}$, see Def.~\ref{def:coindproof},
  the induction hypothesis guarantees that the first $n$ premises $\cstexpi{1}\formeq\dstexpi{1}$, \ldots, $\cstexpi{n}\formeq\dstexpi{n}$ 
  are derivable in $\milnersys$.
  Then the $(n+1)$-th premise $\aLLEECoProofoverof{\thplus{\milnersysmin}{\aseteqs}}{\astexp\formeq\bstexp}$ 
  for $\aseteqs \defdby \descsetexp{ \cstexpi{i}\formeq\dstexpi{i} }{ i\in\setexp{1,\ldots,n} }$ 
  is also a \LLEEwitnessed\ coinductive proof of $\astexp\formeq\bstexp$ over $\milnersys$. 
  Therefore we can apply Prop.~\ref{prop:LLEEcoindproofeq:impl:milnersyseq} in order to obtain a derivation of $\astexp\formeq\bstexp$ in $\milnersys$,
  the conclusion of the $\LCoindProofi{n}$ instance.
\end{proof}


\section{From Milner's system to LLEE-witnessed coinductive proofs}
  \label{milnersys:2:coindmilnersys}

In this section we develop a proof-theoretic interpretation of $\milnersys$
  in $\coindmilnersysone$, and hence in $\coindmilnersys$. 
The crucial step hereby is to show that every instance $\ainst$ of the fixed-point rule \RSPstar\ of \milnersys\
  can be mimicked by a \LLEEwitnessed\ coinductive proof over \milnersysmin\ 
    in which also the premise of $\ainst$ may be used. 
Specifically, an \RSPstar-instance with premise $e = f \prod e + g$ and conclusion $e = f^* \prod g$   
  can be translated into a coinductive proof of $e = f^* \prod g$ over $\thplus{\milnersysmin}{\setexp{e = f \prod e + g}}$
  with underlying \onechart~$\onechartof{f^* \prod g}$ and \LLEEwitness~$\onecharthatof{f^* \prod g}$.
While we have illustrated this transformation already in Fig.~\ref{fig:ex:1:RSPstar:to:coindproof} in the Introduction,
  we detail it also for a larger example.

\begin{exa}\label{ex:1:RSPstar:2:LLEEcoindproof}
  We consider an instance of \RSPstar\ that corresponds, up to an application of $\rdistr$,
  to the instance of \RSPstar\ at the bottom in Fig.~\ref{fig:ex:extraction:proof}:\vspace*{-1.5ex} 
  \begin{equation}\renewcommand{\fCenter}{\formeq}\label{eq:ex:1:RSPstar:2:LLEEcoindproof}
    \mbox{%
    \Axiom$ \overbrace{(a + b)^*}^{\chocolate{e}}   \fCenter   \overbrace{((a \cdot a^* + b) \cdot b^*)}^{\alert{f}} \cdot \overbrace{(a + b)^*}^{\chocolate{e}} + \overbrace{1}^{ \forestgreen{g} } $
    \RightLabel{\RSPstar \hspace*{1.5ex} \text{\small (where $\notterminates{\alert{f}}$)}}
    \UnaryInf$ \underbrace{(a + b)^*}_{\smash{\chocolate{e}}}   \fCenter   \underbrace{((a \cdot a^* + b) \cdot b^*)^*}_{\smash{\alert{f^*}} } \cdot \underbrace{1}_{\smash{\forestgreen{g}}} $
    \DisplayProof}\vspace*{-1ex}
  \end{equation}%
  We want to mimic this instance by one of $\LCoindProofi{1}$ that uses a \LLEEwitnessed\ coinductive proof of 
  $\chocolate{e} \formeq \alert{f^*} \cdot \forestgreen{g}$ over $\milnersysmin$ plus the premise of the \RSPstar\ instance. 
  We first obtain the \onechart\ interpretation $\onechartof{\alert{\stexpit{\bstexp}}}$ of $\alert{\stexpit{\bstexp}}$
  according to Def.~\ref{def:onechartof}
  together with its \LLEEwitness~$\onecharthatof{\alert{\stexpit{\bstexp}}}\,$:
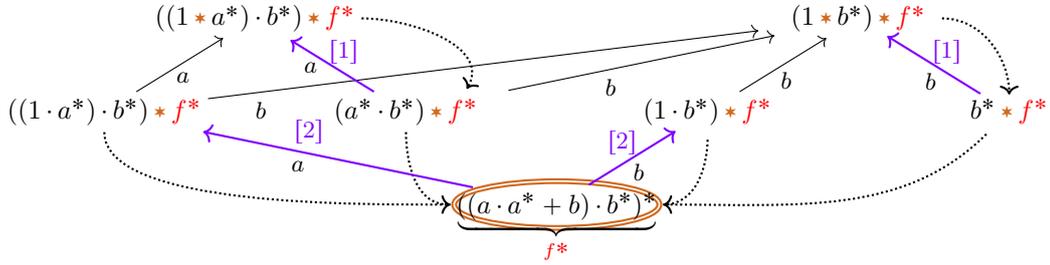
\vspace*{-0.5ex}
  \begin{center}
    \begin{tikzpicture}
  \renewcommand{\prod}{\,{\cdot}\,}
  \renewcommand{\stackprod}{\,{\chocolate{\varstar}}\,}

\matrix[anchor=center,row sep=1cm,column sep=1.75cm,
        every node/.style={draw=none}
        ] {
    & \node(v11){}; & & & & \node(v21){};
    \\
    \node(v1acc){}; & & \node(v1){}; & & \node(v2acc){}; & & \node(v2){};
    \\
    & & & \node(v){}; 
    \\ 
  };

\path (v11) ++ (0cm,0cm) node{$ ((1 \stackprod a^*) \prod b^*) \stackprod \alert{f^*} $}; 
\draw[->,thick,densely dotted,out=0,in=60] ($(v11) + (1.4cm,0cm)$) to ($(v1) + (0.8cm,0.3cm)$);

\path (v1acc) ++ (0cm,0cm) node{$ ((1 \prod a^*) \prod b^*) \stackprod \alert{f^*} $};
\draw[->,shorten <=0.35cm,shorten >=0.35cm] 
  (v1acc) to node[below,pos=0.525]{\small \black{$\aact$}} (v11);
\draw[->,shorten <=1.25cm,shorten >=1.2cm] (v1acc) to node[below,pos=0.2]{\small $\bact$} (v21);
\draw[-,thick,densely dotted,out=270,in=180,shorten <= 0.15cm,distance=1.15cm] (v1acc) to ($(v) + (-1.4cm,0cm)$);
 
\path (v1) ++ (0cm,0cm) node{$ (a^* \prod b^*) \stackprod \alert{f^*} $};
\draw[->,thick,densely dotted,out=270,in=180,shorten <= 0.15cm] (v1) to ($(v) + (-1.4cm,0cm)$);
\draw[->,thick,darkcyan,shorten <=0.35cm,shorten >=0.4cm] 
  (v1) to node[below,pos=0.65]{\small \black{$\aact$}} node[above,pos=0.4]{\small $\loopsteplab{1}$} (v11);
\draw[->,shorten <=1.25cm,shorten >=1cm] (v1) to node[below,pos=0.45]{\small $\bact$} (v21);

\path (v21) ++ (0cm,0cm) node{$ (1 \stackprod b^*) \stackprod \alert{f^*} $}; 
\draw[->,thick,densely dotted,out=0,in=90,shorten >=0.1cm] ($(v21) + (1.1cm,0cm)$) to (v2);

\path (v2acc) ++ (0cm,0cm) node{$ (1 \prod b^*) \stackprod \alert{f^*} $};
\draw[->,shorten <=0.35cm,shorten >=0.35cm] 
  (v2acc) to node[below,pos=0.525]{\small \black{$\bact$}} (v21); 
\draw[-,thick,densely dotted,out=270,in=0,shorten <= 0.25cm] (v2acc) to ($(v) + (1.4cm,0cm)$);  
 
\path (v2) ++ (0cm,0cm) node{$ b^* \stackprod \alert{f^*} $};
\draw[->,thick,densely dotted,out=225,in=0,shorten <= 0.25cm] (v2) to ($(v) + (1.4cm,0cm)$); 
\draw[->,thick,darkcyan,shorten <=0.3cm,shorten >=0.3cm] 
  (v2) to node[below,pos=0.525]{\small \black{$\bact$}} node[above,pos=0.4]{\small $\loopsteplab{1}$} (v21);

\draw[thick,chocolate,double] (v) ellipse (1.35cm and 0.31cm);
\path (v) ++ (0cm,-0.3cm) node{$ \underbrace{{((a \prod a^* + b) \prod b^*)^*}}_{\alert{f^*}} $}; 
\draw[->,thick,darkcyan,shorten <=1cm,shorten >=1.2cm] 
  (v) to node[below,pos=0.575]{\small \black{$\aact$}} node[above,pos=0.55]{\small $\loopsteplab{2}$} (v1acc);
\draw[->,thick,darkcyan,shorten <=0.35cm,shorten >=0.35cm] 
  (v) to node[below,pos=0.55]{\small \black{$\bact$}} node[above,pos=0.425]{\small $\loopsteplab{2}$} (v2acc);

\end{tikzpicture}
  \end{center}\vspace{-2.25ex}
  Due to Lem.~\ref{lem:onechart-int:milnersysmin:solvable} 
    the iterated partial \onederivatives\ as depicted define a \provablein{\milnersysmin} solution of $\onechartof{\alert{\stexpit{\bstexp}}}$
    when stacked products $\sstexpstackprod$ are replaced by products $\sstexpprod\,$. 
  From this \LLEEwitness\ that carries a \provablein{\milnersys} solution we now obtain a \LLEEwitnessed\ coinductive proof 
    of $ \alert{f} \prod \chocolate{e} + \forestgreen{g} \formeq \alert{f^*} \prod \forestgreen{g}$
    under the assumption of $\chocolate{e} \formeq \alert{f} \prod \chocolate{e} + \forestgreen{g}$, as follows.
      By replacing parts $({\ldots}) \stackprod \alert{f^*}$ by $\proj{({\ldots})} \prod \chocolate{e}$   
    in the \provablein{\milnersys} solution of $\onechartof{\alert{f^*}}$,
  and respectively, by replacing $({\ldots}) \stackprod \alert{f^*}$ by $(\proj{({\ldots})} \prod \alert{f^*}) \prod \forestgreen{g}$  
  we obtain the left- and the right-hand sides of the formal equations below:\vspace*{-2.5ex}  
  \begin{center}\enlargethispage{5ex}
    \begin{tikzpicture}
 
\matrix[anchor=center,row sep=1.75cm,column sep=1.35cm,
        every node/.style={draw=none}
        ] {
    & \node(v11){}; & & & & \node(v21){};
    \\ 
    \node(v1acc){}; & & \node(v1){}; & & \node(v2acc){}; & & \node(v2){};
    \\
    & & & \node(v){}; 
    \\
  };

\path (v11) ++ (0cm,0cm) node{$ ((1 \prod a^*) \prod b^*) \prod \chocolate{e}
                                  \formeq
                                (((1 \prod a^*) \prod b^*) \prod \alert{f^*}) \prod \forestgreen{g} $}; 
\draw[->,thick,densely dotted,out=-45,in=90] ($(v11) + (1cm,-0.35cm)$) to ($(v1) + (0.5cm,0.55cm)$);

\path (v1acc) ++ (0cm,-0.3cm) node{$ ((1 \prod a^*) \prod b^*) \prod \chocolate{e} 
                                    \formeq
                                   (((1 \prod a^*) \prod b^*) \prod \alert{f^*}) \prod \forestgreen{g} $};
\draw[->,shorten <=0cm,shorten >=0.3cm] 
  (v1acc) to node[above,pos=0.35]{\small \black{$\aact$}} (v11);
\draw[->,shorten <=0cm,shorten >=1cm,out=35,in=195] (v1acc) to node[above,pos=0.15]{\small $\bact$} (v21);
\draw[->,thick,densely dotted,out=270,in=120,shorten <= 0.25cm,shorten >= 0.325cm] ($(v1acc) + (-0.475cm,-0.3cm)$) to ($(v) + (-2.25cm,0cm)$);

\path (v1) ++ (0cm,0.3cm) node{$ (a^* \prod b^*) \prod \chocolate{e} 
                                 \formeq
                               ((a^* \prod b^*) \prod \alert{f^*}) \prod \forestgreen{g} $};
\draw[->,thick,densely dotted,out=270,in=90,shorten >= 0.2cm] ($(v1) + (0cm,0cm)$) to (v); 
\draw[->,thick,darkcyan,shorten <=0.5cm,shorten >=0.4cm] 
  (v1) to node[below,pos=0.45,xshift=-0.05cm]{\small \black{$\aact$}} node[above,pos=0.25,xshift=0.05cm]{\small $\loopsteplab{1}$} (v11);
\draw[->,shorten <=1.25cm,shorten >=1cm] (v1) to node[below,pos=0.45]{\small $\bact$} (v21);

\path (v21) ++ (0cm,0cm) node{$ (1 \prod b^*) \prod \chocolate{e}
                                  \formeq
                                ((1 \prod b^*) \prod \alert{f^*}) \prod \forestgreen{g} $}; 
\draw[->,thick,densely dotted,out=315,in=90,shorten <=0.5cm,shorten >=0.5cm] ($(v21) + (0.75cm,0cm)$) to (v2); 
 
\path (v2acc) ++ (0cm,-0.3cm) node{$ (1 \prod b^*) \prod \chocolate{e} 
                                    \formeq
                                  ((1 \prod b^*) \prod \alert{f^*}) \prod \forestgreen{g} $};
\draw[->,shorten <=0cm,shorten >=0.35cm] 
  (v2acc) to node[below,pos=0.525]{\small \black{$\bact$}} (v21);  
\draw[->,thick,densely dotted,out=290,in=70,shorten <= 0.2cm,shorten >= 0.2cm,distance=0.5cm] ($(v2acc) + (0cm,-0.3cm)$) to ($(v) + (1cm,0.15cm)$); 
 
\path (v2) ++ (0cm,0.3cm) node{$ b^* \prod \chocolate{e}
                                    \formeq
                                  (b^* \prod \alert{f^*}) \prod \forestgreen{g} $};
\draw[->,thick,densely dotted,out=265,in=0,shorten <= 0cm] (v2) to ($(v) + (4.4cm,0.15cm)$); 
\draw[->,thick,darkcyan,shorten <=0.5cm,shorten >=0.3cm] 
  (v2) to node[below,pos=0.625]{\small \black{$\bact$}} node[above,pos=0.425,xshift=0.05cm]{\small $\loopsteplab{1}$} (v21);

(v)
\path (v) ++ (0cm,-0.8cm) node{$ \underbrace{
                                   \underbrace{{((a \prod a^* + b) \prod b^*)}\rule[-6.5pt]{0pt}{7pt}}_{\alert{f}} \prod \underbrace{ (a + b)^* \rule[-6.5pt]{0pt}{7pt}}_{\chocolate{e}} + \underbrace{1\rule[-6.5pt]{0pt}{7pt}}_{\forestgreen{g}}
                                             }_{\sformeq\: \chocolate{e} \;\; \text{\small (by rule assumption)}} 
                                   \hspace*{-1.5ex} \formeq
                                 \underbrace{{(((a \prod a^* + b) \prod b^*)^* \rule[-6.5pt]{0pt}{7pt}}}_{\alert{f^*}} \prod \hspace*{-2ex} \underbrace{ 1 \rule[-6.5pt]{0pt}{7pt} }_{\forestgreen{g}} $}; 

\draw[->,thick,darkcyan,shorten <=0.75cm,shorten >=1.5cm] 
  (v) to node[below,pos=0.4]{\small \black{$\aact$}} node[above,pos=0.375]{\small $\loopsteplab{2}$} (v1acc);
\draw[->,thick,darkcyan,shorten <=0.25cm,shorten >=0.6cm] 
  (v) to node[below,pos=0.425]{\small \black{$\bact$}} node[above,pos=0.275]{\small $\loopsteplab{2}$} (v2acc);

\draw[thick,chocolate,double] ($(v) + (-0.15cm,-0.1cm)$) ellipse (4.75cm and 0.4cm);

%
%
%
%
%
%
%

\end{tikzpicture}
  \end{center}\vspace{-1.5ex}
  This is a \LLEEwitnessed\ coinductive proof $\aLLEECoProof$ of $ \alert{f} \cdot \chocolate{e} + \forestgreen{g} = \alert{f^*} \prod \forestgreen{g}$
    over $\thplus{\milnersysmin}{\setexp{\chocolate{\astexp}   \formeq   \stexpsum{\stexpprod{\alert{\bstexp}}{\chocolate{\astexp}}}{\forestgreen{\cstexp}}}}\,$:
  The right-hand sides form a \provablein{\milnersys} solution of $\chartof{\alert{f^*}\cdot\forestgreen{g}}$
  due to Lemma~\ref{lem:onechart-int:milnersysmin:solvable}
  (note that $\chartof{\alert{f^*}\cdot\forestgreen{g}}$ is isomorphic to $\chartof{\alert{f^*}}$ due to $\forestgreen{g} \synteq 1$).
  The left-hand sides also form a solution of $\chartof{\alert{f^*}\cdot\forestgreen{g}}$
  (see Lem.~\ref{lem:lem:mimic:RSPstar} below),
  noting that for the \onetransitions\ back to the conclusion the assumption $\chocolate{e} \formeq \alert{f} \prod \chocolate{e} + \forestgreen{g}$
  must be used in addition to $\milnersysmin$.
  By using this assumption again, the result $\aLLEECoProofacc$ of replacing
    $\alert{f} \prod \chocolate{e} + \forestgreen{g}$ in the conclusion of $\aLLEECoProof$
  by $\chocolate{e}$ is also a \LLEEwitnessed\ coinductive proof 
    over $\thplus{\milnersysmin}{\setexp{\chocolate{\astexp}   \formeq   \stexpsum{\stexpprod{\alert{\bstexp}}{\chocolate{\astexp}}}{\forestgreen{\cstexp}}}}$.
  Consequently:
  \begin{center}
    $
    \begin{aligned}
      \AxiomC{$ \chocolate{\astexp}   \formeq   \stexpsum{\stexpprod{\alert{\bstexp}}{\chocolate{\astexp}}}{\forestgreen{\cstexp}} $}
      \AxiomC{$ \aLLEECoProofoverof{\thplus{\milnersysmin}{\setexp{\chocolate{\astexp}   \formeq   \stexpsum{\stexpprod{\alert{\bstexp}}{\chocolate{\astexp}}}{\forestgreen{\cstexp}}}}}
                                   {\chocolate{\astexp}   \formeq   \stexpprod{\stexpit{\alert{\bstexp}}}{\forestgreen{\cstexp}}} $}
      \insertBetweenHyps{\hspace*{1ex}}                             
      \RightLabel{\LCoindProofi{1}}
      \BinaryInfC{$ \chocolate{\astexp}   \formeq   \stexpprod{\alert{\stexpit{\bstexp}}}{\forestgreen{\cstexp}} $}           
      \DisplayProof
    \end{aligned}
    $
  \end{center}
  is a rule instance of $\coindmilnersys$ and $\CLC$ by which we have mimicked the $\RSPstar$ instance in \eqref{eq:ex:1:RSPstar:2:LLEEcoindproof}.  
  %
\end{exa}


\begin{lem}\label{lem:lem:mimic:RSPstar}
  Let $\astexp,\bstexp,\cstexp\in\StExpover{\actions}$ with $\alert{\notterminates{\bstexp}}$,
  and let $\aseteqs \defdby \setexp{ \astexp = \stexpsum{\stexpprod{\bstexp}{\astexp}}{\cstexp} }$. 
  Then 
       $\astexp$ is the principal value of a
  \provablein{(\thplus{\milnersysmin}{\aseteqs})} solution of
    the \onechart\ interpretation $\onechartof{\stexpprod{\stexpit{\bstexp}}{\cstexp}}$~of~$\stexpit{\bstexp}{\prod}{\cstexp}$.
\end{lem}

\begin{proof}
  First, it can be verified that the vertices of $\onechartof{f^* \cdot g}$ are of either of three forms:
  %
  \begin{center}  
    $
    \vertsof{\onechartof{f^* \cdot g}}
      =
    \setexp{ f^* \cdot g }
      \cup
    \descsetexpbig{ (F \stackprod f^*) \cdot g }{ F \in \itpartonederivs{f} }
      \cup
    \descsetexp{ G }{ G \in \itpartonederivs{g} } \punc{,}  
    $
  \end{center}  
  %
  where $\itpartonederivs{f}$ means the set of iterated \onederivatives\ of $f$ according to the TSS in Def.~\ref{def:onechartof}.
  This facilitates to define a star expression function $\sasol \funin \vertsof{\onechartof{f^* \cdot g}} \to \StExpsover{\actions}$
    on~$\onechartof{f^* \cdot g}$~by: 
  \begin{center}
    $
    \begin{aligned}
      \asol{f^* \prod g}
        & {} \defdby
               e \punc{,}
      & \qquad
      \asol{(F \stackprod f^*) \prod g}
        & \defdby
               \proj{F} \prod e \punc{,}
      & \qquad
      \asol{G}
        & \defdby
               \proj{G} \punc{,}
    \end{aligned}
    $          
  \end{center}
  for all $F \in \itpartonederivs{f}$ and $G \in \itpartonederivs{g}$.
  We show that $\sasol$ is a \provablein{(\thplus{\milnersysmin}{\aseteqs})} solution of $\onechartof{f^* \cdot g}$.
  
  For this, we have to show that $\sasol$ 
  is \provablein{(\thplus{\milnersysmin}{\aseteqs})}
  at each of the three kinds of
    vertices of $\onechartof{f^* \cdot g}$, 
    namely at $f^* \prod g$,
           at $(F \stackprod f^*) \prod g$ with  $F\in\itpartonederivs{f}$,
           and at $G$ with $G\in\itpartonederivs{g}$. 
  %
  Here we only consider $f^* \prod g$,
    because it is the single case in which the assumption in $\aseteqs$ has to be used,
    and the two other forms can be argued similarly (please see in~the~appendix). 
  By $\oneactderivs{H} \defdby \descsetexpnormalsize{ \pair{\aoneact}{H'} }{ H \lt{\aoneact} H' }$
    we denote the set of `action \onederivatives' of a stacked star expression $H$.     
  In the following argument we avoid list representations of transitions
    (as in Def.~\ref{def:provable:solution}) in favor of arguing with sums over sets of action derivatives
    that represent \ACI-equivalence classes of star expressions. 
    This shorthand is possible due to $\ACI \subsystem \milnersysmin$.
  \begin{alignat*}{2}
    \asol{f^* \prod g}
      & \;\,\parbox[t]{\widthof{$\eqin{\thplus{\milnersysmin}{\aseteqs}}$}}{$\synteq$}\:
        e \phantom{ {} \cdot e + g} 
        \qquad\quad \text{(by the definition of $\sasol$)}
      \\
      & \;\,\parbox[t]{\widthof{$\eqin{\thplus{\milnersysmin}{\aseteqs}}$}}{$\eqin{\thplus{\milnersysmin}{\aseteqs}}$}\:
        f \prod e + g
        \qquad\quad \text{(since $\aseteqs = \setexp{ e = f \prod e + g }$)}  
      \displaybreak[0]\\ 
      & \;\,\parbox[t]{\widthof{$\eqin{\thplus{\milnersysmin}{\aseteqs}}$}}{$\milnersysmineq$}\:
        \Big(
          \terminatesconstof{\onechartof{f}}{f}
            + \hspace*{-3.75ex}
          \sum_{\pair{\aoneact}{F}\in\oneactderivs{f}} \hspace*{-3.5ex}
                 \aoneact \prod \proj{F}
             \Big) \prod e
        +
        \Big(
          \terminatesconstof{\onechartof{g}}{g}
            + \hspace*{-3.75ex}
          \sum_{\pair{\aoneact}{G}\in\oneactderivs{g}} \hspace*{-3.5ex}
                 \aoneact \prod \proj{G}
             \Big)   
          \quad \parbox{\widthof{Lem.~\ref{lem:FT:onechart-int})}}
                       {(by using
                        \\[-0.5ex]\phantom{(}%
                          Lem.~\ref{lem:FT:onechart-int})} 
      \displaybreak[0]\\[0.5ex] 
      & \;\,\parbox[t]{\widthof{$\eqin{\thplus{\milnersysmin}{\aseteqs}}$}}{$\milnersysmineq$}\:
        \Big(
            \hspace*{-3.75ex}
          \sum_{\pair{\aoneact}{F}\in\oneactderivs{f}} \hspace*{-3.5ex}
                 \aoneact \prod ( \proj{F} \prod e )
             \Big) 
        + \Bigl(
            \terminatesconstof{\onechartof{f^* \prod g}}{f^* \prod g}
              + \hspace*{-3.75ex}
            \sum_{\pair{\aoneact}{G}\in\oneactderivs{g}} \hspace*{-3.5ex}
                   \aoneact \prod \proj{G} 
          \Bigr)  
          \;\; \parbox{\widthof{(by ($\assocstexpprod$)}}
                       {(by ($\assocstexpprod$), 
                        \\[-0.5ex]\phantom{(}%
                        ($\rdistr$),
                        \\[-0.5ex]\phantom{(}%
                        ($\stexpzerostexpprod$))}  
        \\[-0.5ex]
        & \;\,\parbox[t]{\widthof{$\eqin{\thplus{\milnersysmin}{\aseteqs}}$\hspace*{3ex}}}{\mbox{}}\: 
          \text{(using $\terminatesconstof{\onechartof{f}}{f} \synteq 0$ due to $\alert{\notterminates{f}}$,
                 and $\terminatesconstof{\onechartof{f^* \prod g}}{f^* \prod g} \synteq \terminatesconstof{\onechartof{g}}{g}$)}     
      \displaybreak[0]\\[0.5ex] 
      & \;\,\parbox[t]{\widthof{$\eqin{\thplus{\milnersysmin}{\aseteqs}}$}}{$\ACIeq$}\:
        \terminatesconstof{\onechartof{f^* \prod g}}{f^* \prod g}
           + 
        \Big(
            \hspace*{-3.75ex}
          \sum_{\pair{\aoneact}{F}\in\oneactderivs{f}} \hspace*{-3.5ex}
                 \aoneact \prod ( \asol{(F \stackprod f^*) \cdot g } )
             \Big) 
        +   \hspace*{-3.75ex}
          \sum_{\pair{\aoneact}{G}\in\oneactderivs{g}} \hspace*{-3.5ex}
                 \aoneact \prod \asol{G}    
        \\[-0.5ex]
        & \;\,\parbox[t]{\widthof{$\eqin{\thplus{\milnersysmin}{\aseteqs}}$\hspace*{3ex}}}{\mbox{}}\: 
          \text{(by definition of $\sasol$, axioms ($\commstexpsum$))}     
      \displaybreak[0]\\[0.5ex] 
      & \;\,\parbox[t]{\widthof{$\eqin{\thplus{\milnersysmin}{\aseteqs}}$}}{$\ACIeq$}\:
        \terminatesconstof{\onechartof{f^* \cdot g}}{f^* \prod g}
           + 
            \hspace*{-3.5ex}
          \sum_{\pair{\aoneact}{E'}\in\oneactderivs{f^* \prod g}} \hspace*{-3.5ex}
                 \aoneact \prod \asol{E'}  
      \\[-0.5ex]
      & \;\,\parbox[t]{\widthof{$\eqin{\thplus{\milnersysmin}{\aseteqs}}$\hspace*{3ex}}}{\mbox{}}\: 
        \parbox{\widthof{{(since
                          $
                          \oneactderivs{f^* \prod g}
                            =
                          \descsetexp{ \pair{\aoneact}{(F \stackprod f^*) \prod g} }
                                     { \pair{\aoneact}{F} \in \oneactderivs{f} } 
                             \cup
                          \oneactderivs{g}$}}}
               {(since
                $
                \oneactderivs{f^* \prod g}
                  =
                \descsetexp{ \pair{\aoneact}{(F \stackprod f^*) \prod g} }
                           { \pair{\aoneact}{F} \in \oneactderivs{f} } 
                   \cup
                \oneactderivs{g} 
                $%
               \\[-0.25ex]\phantom{(}%
                by inspection of the TSS in Def.~\ref{def:onechartof}).}
  \end{alignat*}
  Due to $\ACI \subsystem \milnersysmin \subsystem \thplus{\milnersysmin}{\aseteqs}$,
    the above chain of equalities is provable in $\thplus{\milnersysmin}{\aseteqs}$.
    Therefore it demonstrates, 
      for any given list representation of $\transitionsinfrom{\onechartof{f^*\cdot g}}{f^*\prod g}$
      according to the correctness condition in Def.~\ref{def:provable:solution},
    that $\sasol$ is a \provablein{(\thplus{\milnersysmin}{\aseteqs})} solution of $\onechartof{f^* \cdot g}$~%
      at~$f^* \prod g$.
\end{proof}

\begin{lem}\label{lem:mimic:RSPstar}
  Let $\astexp,\bstexp,\cstexp\in\StExpover{\actions}$ with $\notterminates{\bstexp}$,
  and let $\aseteqs \defdby \setexp{ \astexp   \formeq   \stexpsum{\stexpprod{\bstexp}{\astexp}}{\cstexp} }$. 
  Then it holds that
  $ \astexp \LLEEcoindproofeqin{(\thplus{\milnersysmin}{\aseteqs})} \stexpprod{\stexpit{\bstexp}}{\cstexp} $.
\end{lem}

\begin{proof}
  By Lem.~\ref{lem:lem:mimic:RSPstar} there is a \provablein{\thplus{\milnersysmin}{\aseteqs}} solution 
    $\sasoli{1}$ of $\onechartof{\stexpprod{\stexpit{\bstexp}}{\cstexp}}$ with  
    $\asoli{1}{\stexpprod{\stexpit{\bstexp}}{\cstexp}} \synteq \astexp$. 
  By Lem.~\ref{lem:onechart-int:milnersysmin:solvable} there is a \provablein{\milnersysmin} solution $\sasoli{2}$ of $\onechartof{\stexpprod{\stexpit{\bstexp}}{\cstexp}}$
    with $\asoli{2}{\stexpprod{\stexpit{\bstexp}}{\cstexp}} \synteq \stexpprod{\stexpit{\bstexp}}{\cstexp}$. 
  Then $\pair{\onechartof{f^*{\prod}g}}{\saeqfun}$
    with $\aeqfun{\avert} \defdby \asoli{1}{\avert} \formeq \asoli{2}{\avert}$ 
        for all $\avert\in\vertsof{\onechartof{f^*{\prod}g}}$
    is a \LLEEwitnessed\ coinductive proof 
      of $\astexp \formeq f^*{\prod}g$
        over $\thplus{\milnersysmin}{\aseteqs}$,
        as $\onechartof{\stexpprod{\stexpit{\bstexp}}{\cstexp}}$ 
        has \LLEEwitness\ $\onecharthatof{\stexpprod{\stexpit{\bstexp}}{\cstexp}}$~%
        by~Thm.~\ref{thm:onechart-int:LLEEw}.
\end{proof}

\begin{thm}\label{thm:milnersys:isthmsubsumedby:coindmilnersysone}
  $\milnersys 
     \isthmsubsumedby
   \coindmilnersysone$.
  What is more,
  every derivation in $\milnersys$ with conclusion $\astexp \formeq \bstexp$ 
    can be transformed effectively
  into a derivation with conclusion $\astexp \formeq \bstexp$ in $\coindmilnersysone$.
\end{thm}

\begin{proof}
  Every derivation $\aDeriv$ in $\milnersys$ can be transformed into a derivation $\aDerivacc$ in $\coindmilnersysone$ with the same conclusion as $\aDeriv$
  by replacing every instance of $\RSPstar$ in $\aDeriv$ by a mimicking derivation in $\coindmilnersysone$ as in the following step,
  where $\notterminates{\bstexp}$ holds as the \sidecondition\ of the instance of $\RSPstar\,$:
  \begin{center}
    $
    \begin{aligned}\renewcommand{\fCenter}{\formeq}
      \Axiom$ \astexp   \fCenter   \stexpsum{\stexpprod{\bstexp}{\astexp}}{\cstexp} $
      \RightLabel{$\RSPstar$} 
      \UnaryInf$ \astexp   \fCenter   \stexpprod{\stexpit{\bstexp}}{\cstexp} $
      \DisplayProof
    \end{aligned}
    \qquad\Longmapsto\qquad
    \begin{aligned}
      \AxiomC{$ \astexp   \formeq   \stexpsum{\stexpprod{\bstexp}{\astexp}}{\cstexp} $}
      \AxiomC{$ \aLLEECoProofoverof{\thplus{\milnersysmin}{\aseteqs}}
                                   {\astexp   \formeq   \stexpprod{\stexpit{\bstexp}}{\cstexp}} $}
      \insertBetweenHyps{\hspace*{1ex}}                             
      \RightLabel{\LCoindProofi{1}}
      \BinaryInfC{$ \astexp   \formeq   \stexpprod{\stexpit{\bstexp}}{\cstexp} $}           
      \DisplayProof
    \end{aligned}
    $
  \end{center}
  and where 
  $\aLLEECoProofoverof{\thplus{\milnersysmin}{\aseteqs}}
                      {\astexp   \formeq   \stexpprod{\stexpit{\bstexp}}{\cstexp}}$
  is, for $\aseteqs \defdby \setexp{\astexp   \formeq   \stexpsum{\stexpprod{\bstexp}{\astexp}}{\cstexp} }$,
  a \LLEEwitnessed\ coinductive proof over ${\thplus{\milnersysmin}{\aseteqs}}$
  of $\astexp   \formeq   \stexpprod{\stexpit{\bstexp}}{\cstexp}$
  that is guaranteed by Lem.~\ref{lem:mimic:RSPstar}.                     
\end{proof}

\begin{thm}
  $\milnersys \thmequiv \coindmilnersysone \thmequiv \coindmilnersys \thmequiv \CLC$,
    i.e.\ these proof systems are \theoremequivalent.
\end{thm}

\begin{proof}
  Due to $\milnersys \isthmsubsumedby \coindmilnersysone \isthmsubsumedby \coindmilnersys \: (\thmequiv \CLC) \isthmsubsumedby \milnersys$
  by Thm.~\ref{thm:milnersys:isthmsubsumedby:coindmilnersysone}, 
     Lem.~\ref{lem:easy:rels:coindproofsystems},
     and Thm.~\ref{thm:coindmilnersys:isthmsubsumedby:milnersys}.
\end{proof}

%
%
%
%
%
%
%
%

%
%
%
%
%
%
%
%
%

\section{Conclusion}%
  \label{conclusion}

In order to increase the options for a completeness proof of Milner's system $\milnersys$ for the process semantics of regular expressions
  under bisimilarity, we set out to formulate proof systems of equal strength half-way in between $\milnersys$ 
    and bisimulations between star expressions.
Specifically we aimed at characterizing the derivational power that the fixed-point rule $\RSPstar$ in $\milnersys$
  adds to its purely equational part $\milnersysmin$.     
We based our development on a crucial step from the completeness proof \cite{grab:fokk:2020a} for a tailored restriction of~$\milnersys$
  to `\onefree' star expressions:
  guarded linear specifications that satisfy the loop existence and elimination property (L)LEE \cite{grab:2018,grab:fokk:2020a} 
  are uniquely solvable in $\milnersys$.
We have obtained the following concepts and results:  
\begin{itemize}
  \item
    As \LLEEwitnessed\ coinductive proof we defined
      any \weaklyguarded\ \LLEEonechart\ $\aonechart$ whose vertices are labeled by equations 
      between the values of two provable solutions of $\aonechart$.
  \item
    Based on such proofs, we defined a coinductive version $\coindmilnersys$ of Milner's system $\milnersys$,
    and as its `kernel' a system \CLC\ for merely combining \LLEEwitnessed\ coinductive proofs.
  \item
    Via proof transformations we showed that $\coindmilnersys$ and \CLC\ are \theoremequivalent\ to $\milnersys$.
  \item
    Based on coinductive proofs without \LLEEwitnesses, 
    we formulated systems $\coindmilnersysbar$ and $\CC$ that can be shown to be complete,
     as can a variant $\milnersysaccbar$ of $\milnersys$ with the strong~rule~\USP.
\end{itemize}
\smallskip%
Since the proof systems $\coindmilnersys$ and $\CLC$
  are tied to process graphs via the circular deductions they permit,
  and as they are \theoremequivalent\ with $\milnersys$,
    they may become natural beachheads for a completeness proof of Milner's system.
Indeed, they can be linked to the completeness proof in \cite{grab:fokk:2020a}:
  it namely guarantees that valid equations between `\onefree' star expressions
  can always be mimicked by derivations in \CLC\ of depth~2. This suggests~the~following~question:
\begin{itemize}
  \item[$\triangleright$]
    Can derivations in \CLC\ (in $\coindmilnersys$) always be simplified to some (kind of) normal form
      that is of bounded depth (resp., of bounded nesting depth of LLEE-witn.\ coinductive proofs)? 
\end{itemize}
Investigating workable concepts of `normal form' for derivations in \CLC\ or in $\coindmilnersys$,
  by using simplification steps of process graphs with \LEE\ and \onetransitions\ under \onebisimilarity\ as developed for the completeness proof for `\onefree' star expressions
  in \cite{grab:fokk:2020a}, is our next goal.

\newpage
\bibliography{cMil-arxiv.bib}         

\begin{thebibliography}{10}

\bibitem{amad:card:1993}
Roberto~M. Amadio and Luca Cardelli.
\newblock {Subtyping Recursive Types}.
\newblock {\em ACM Trans. Program. Lang. Syst.}, 15(4):575–631, September
  1993.
\newblock \href {https://doi.org/10.1145/155183.155231}
  {\path{doi:10.1145/155183.155231}}.

\bibitem{anti:1996}
Valentin Antimirov.
\newblock {Partial Derivatives of Regular Expressions and Finite Automaton
  Constructions}.
\newblock {\em Theoretical Computer Science}, 155(2):291--319, 1996.
\newblock \href {https://doi.org/https://doi.org/10.1016/0304-3975(95)00182-4}
  {\path{doi:https://doi.org/10.1016/0304-3975(95)00182-4}}.

\bibitem{blom:2001}
Stefan Blom.
\newblock {\em {Term Graph Rewriting -- Syntax and Semantics}}.
\newblock PhD thesis, Vrije Universiteit Amsterdam, 2001.

\bibitem{boss:1997}
Doeko Bosscher.
\newblock {\em {Grammars Modulo Bisimulation}}.
\newblock PhD thesis, University of Amsterdam, 1997.

\bibitem{bran:heng:1998}
Michael Brandt and Fritz Henglein.
\newblock {Coinductive Axiomatization of Recursive Type Equality and
  Subtyping}.
\newblock {\em Fundamenta Informaticae}, 33(4):309–338, dec 1998.
\newblock \href {https://doi.org/10.1007/3-540-62688-3_29}
  {\path{doi:10.1007/3-540-62688-3_29}}.

\bibitem{grab:2005:prfschr}
Clemens Grabmayer.
\newblock {\em {Relating Proof Systems for Recursive Types}}.
\newblock PhD thesis, Vrije Universiteit Amsterdam, March 2005.
\newblock
  \href{http://www.phil.uu.nl/~clemens/linkedfiles/proefschrift.pdf}{\texttt{www.phil.uu.nl/\~{}clemens/linkedfiles/proefschrift.pdf}}.

\bibitem{grab:2005:calco}
Clemens Grabmayer.
\newblock {Using Proofs by Coinduction to Find ``Traditional'' Proofs}.
\newblock In Jos\'{e}~Luiz Fiadeiro, Neal Harman, Markus Roggenbach, and Jan
  Rutten, editors, {\em Proceedings of CALCO 2005}, volume 3629 of {\em LNCS},
  pages 175--193. Springer, 2005.

\bibitem{grab:2006}
Clemens Grabmayer.
\newblock {A Coinductive Axiomatisation of Regular Expressions under
  Bisimulation}.
\newblock Technical report, University of Nottingham, 2006.
\newblock Short Contribution to CMCS~2006, March 25-27, 2006, Vienna Institute
  of Technology, Austria.

\bibitem{grab:2018}
Clemens Grabmayer.
\newblock {Modeling Terms by Graphs with Structure Constraints (Two
  Illustrations)}.
\newblock In {\em Proc.\ TERMGRAPH@FSCD'18}, volume 288, pages 1--13,
  \url{http://www.eptcs.org/}, 2019.
\newblock \href {https://doi.org/10.4204/EPTCS.288.1}
  {\path{doi:10.4204/EPTCS.288.1}}.

\bibitem{grab:2020:scpgs-arxiv}
Clemens Grabmayer.
\newblock {Structure-Constrained Process Graphs for the Process Semantics of
  Regular Expressions}.
\newblock Technical report, \href{http://arxiv.org}{\texttt{arxiv.org}},
  December 2020.
\newblock \href {http://arxiv.org/abs/2012.10869} {\path{arXiv:2012.10869}}.

\bibitem{grab:2021:CALCO}
Clemens Grabmayer.
\newblock {A Coinductive Version of Milner's Proof System for Regular
  Expressions Modulo Bisimilarity}.
\newblock In {\em Proceedings of CALCO~2021}. LIPIcs, 2021.
\newblock To appear.

\bibitem{grab:2021:TERMGRAPH-postproceedings}
Clemens Grabmayer.
\newblock {Structure-Constrained Process Graphs for the Process Semantics of
  Regular Expressions}.
\newblock In Patrick Bahr, editor, {\em {\rm Proceedings 11th International
  Workshop on} Computing with Terms and Graphs, {\rm Online, 5th July 2020}},
  volume 334 of {\em Electronic Proceedings in Theoretical Computer Science},
  pages 29--45. Open Publishing Association, 2021.
\newblock \href {https://doi.org/10.4204/EPTCS.334.3}
  {\path{doi:10.4204/EPTCS.334.3}}.

\bibitem{grab:fokk:2020a}
Clemens Grabmayer and Wan Fokkink.
\newblock {A Complete Proof System for 1-Free Regular Expressions Modulo
  Bisimilarity}.
\newblock In {\em Proceedings of LICS~2020}, New York, NY, 2020. ACM.
\newblock Extended report see~\cite{grab:fokk:2020b}.

\bibitem{grab:fokk:2020b}
Clemens Grabmayer and Wan Fokkink.
\newblock {A Complete Proof System for 1-Free Regular Expressions Modulo
  Bisimilarity}.
\newblock Technical report, \href{http://arxiv.org}{\texttt{arxiv.org}}, April
  2020.
\newblock \href {http://arxiv.org/abs/2004.12740} {\path{arXiv:2004.12740}}.

\bibitem{miln:1984}
Robin Milner.
\newblock {A Complete Inference System for a Class of Regular Behaviours}.
\newblock {\em Journal of Computer and System Sciences}, 28(3):439--466, 1984.

\bibitem{salo:1966}
Arto Salomaa.
\newblock {Two Complete Axiom Systems for the Algebra of Regular Events}.
\newblock {\em Journal of the ACM}, 13(1):158--169, 1966.
\newblock \href {https://doi.org/10.1145/321312.321326}
  {\path{doi:10.1145/321312.321326}}.

\end{thebibliography}

%
%
%
%
%
%
%
%

\newpage

\appendix
\section{Appendix: Supplements}%
  \label{appendix}

\subsection{Supplements for Section~\ref{intro}}\label{intro::app}        

\begin{rem}[not expressible process graphs]
  As mentioned in Section~\ref{intro}, 
    Milner noticed that guarded systems of recursion equations cannot always be solved by star expressions
    under the process semantics.
  In fact, Milner showed in \cite{miln:1984}
    that the linear specification $\fap{\aspec}{G_2}$ defined by the process graph $G_2$  below 
    does not have a star expression solution modulo bisimilarity. 
  He conjectured that that holds also for the easier specification $\fap{\aspec}{G_1}$ defined by the process graph $G_1$ below,
    which was confirmed, and proved later by Bosscher \cite{boss:1997}.    
  \begin{center}\vspace{-1ex}\label{fig:milner-bosscher-expressible}
    \scalebox{0.9}{\begin{tikzpicture}

%
\matrix[anchor=center,row sep=0cm,column sep=1.5cm,
        every node/.style={draw,very thick,circle,minimum width=2.5pt,fill,inner sep=0pt,outer sep=2pt}] at (0,0.75) {
  \node[color=chocolate](C-1-0){};  &   \node[color=chocolate](C-1-1){};
  \\
};
\draw[<-,very thick,>=latex,chocolate,shorten <=2pt](C-1-0) -- ++ (180:0.58cm);
\path(C-1-0) ++ (0cm,-0.4cm) node{$X_1$};
\path(C-1-0) ++ (-0.325cm,0.5cm) node{\Large $\iap{G}{1}$};

\path (C-1-1) ++ (0cm,-0.45cm) node{$X_2$};

\draw[thick,chocolate] (C-1-1) circle (0.12cm);
\draw[thick,chocolate] (C-1-0) circle (0.12cm);
\draw[->,bend left,distance=0.65cm,shorten <=2pt,shorten >=2pt] (C-1-0) to node[above]{$a$} (C-1-1); 
\draw[->,bend left,distance=0.65cm,shorten <=2pt,shorten >=2pt] (C-1-1) to node[below]{$b$} (C-1-0); 


\path (C-1-0) ++ (-0.75cm,0cm) node[left]{$
  \fap{\aspec}{G_1} = 
  \left\{\,
  \begin{aligned}
    X_1 & {} = 1 + a \prod X_2
    \\ 
    X_2 & {} = 1 + b \prod X_1
  \end{aligned}
  \,\right.
  $};

%
\matrix[anchor=center,row sep=0.924cm,column sep=0.75cm,
        every node/.style={draw,very thick,circle,minimum width=2.5pt,fill,inner sep=0pt,outer sep=2pt}] at (7,-0.75) {
                   &                  &  \node(C-2-2){};
  \\
  \node(C-2-1){};  &                  &                  
  \\
                   &                  &  \node(C-2-3){};  
  \\
};
\draw[<-,very thick,>=latex,color=chocolate](C-2-1) -- ++ (180:0.5cm);  

\path(C-2-1) ++ (-0.325cm,0.5cm) node{\Large $\iap{G}{2}$};
\path(C-2-1) ++ (-0.2cm,-0.4cm) node{$Y_1$};
\draw[->,bend right,distance=0.65cm] (C-2-1) to node[above]{$\aacti{2}$} (C-2-2); 
\draw[->,bend right,distance=0.65cm] (C-2-1) to node[left]{$\aacti{3}$}  (C-2-3);

\path(C-2-2) ++ (0cm,0.4cm) node{$Y_2$};
\draw[->,bend right,distance=0.65cm]  (C-2-2) to node[above]{$\aacti{1}$} (C-2-1); 
\draw[->,bend left,distance=0.65cm]  (C-2-2) to node[right,xshift=-1pt]{$\aacti{3}$} (C-2-3);

\path(C-2-3) ++ (0cm,-0.4cm) node{$Y_3$};
\draw[->,bend right,distance=0.65cm] (C-2-3) to node[left]{$\aacti{1}$}  ($(C-2-1)+(+0.15cm,-0.05cm)$);
\draw[->,bend left,distance=0.65cm]  (C-2-3) to node[right,xshift=-1.5pt]{$\aacti{2}$} (C-2-2);

\path (C-2-1) ++ (-0.75cm,0cm) node[left]{$
  \fap{\aspec}{G_2} = 
  \left\{\,
  \begin{aligned}
    Y_1 & {} = a_2 \prod Y_2  +  a_3 \prod Y_3
    \\ 
    Y_2 & {} = a_1 \prod Y_1  +  a_3 \prod Y_3
    \\
    Y_3 & {} = a_1 \prod Y_1  +  a_2 \prod Y_2
  \end{aligned}
  \,\right.
  $};

\end{tikzpicture}}
  \end{center}\vspace{-2ex}
  Here the start vertex of a process graph is again highlighted by a brown arrow~\picarrowstart,
  and a vertex $\avert$ with immediate termination
  is emphasized in brown as \pictermvert\ including a boldface ring.
  
  $G_1$ and $G_2$ are finite process graphs that are not bisimilar to the process interpretation of any star expression.
  In this sense, $G_1$ and $G_2$ are \emph{not} expressible by a regular expression (under the process semantics). 
  This sets the process semantics of regular expressions apart from the standard language semantics,
  with respect to which every language that is accepted by a finite-state automaton is the interpretation of some regular expression.
\end{rem}

\begin{exa}[LLEE-witnessed coinductive proof on page~\pageref{inf:ex:1:coindproof}]\label{ex:inf:ex:1:coindproof}
  In Section~\ref{intro} on page~\pageref{inf:ex:1:coindproof} we displayed,
  for the statement 
  $g^* \prod 0 \synteq (a + b)^* \prod 0 
     \LLEEcoindproofeqin{\milnersysmin}
   (a \prod (a + b) + b)^* \prod 0 \synteq h^* \prod 0$  
  the coinductive proof $\aCoProof = \pair{\onechartof{h^* \prod 0}}{\saeqfun}$ over $\milnersysmin$
    with underlying \LLEEwitness~$\onecharthatof{h^* \prod 0}$,
  where $\onechartof{h^* \prod 0}$ and $\onecharthatof{h^* \prod 0}$
    are defined according to Def.~\ref{def:onechartof}
    and the equation-labeling function $\saeqfun$ on $\onechartof{h^* \prod 0}$ is defined by the illustration below:
  \begin{flushleft}\vspace*{-2ex}
    \hspace*{-0.5ex}%
    \begin{tikzpicture}
 
\matrix[anchor=center,row sep=1.25cm,column sep=2.65cm,
        every node/.style={draw=none}
        ] at (0,0) {
    \node(v1){}; &             & \node(v2){};
    \\
                 & \node(v){}; 
    \\
  };
  
\path (v1) node(v1-label){$ (1 \cdot g^*) \cdot 0
                              \formeq
                            ((1 \cdot (a + b)) \cdot h^*)  \cdot 0 $};       
\draw[->,shorten <=2.6cm,shorten >=1.75cm] (v1) to node[above,pos=0.5725]{$a, b$} (v2);

\path (v2) node(v2-label){$ (1 \cdot g^*) \cdot 0
                              \formeq
                             (1 \cdot h^*) \cdot 0 $}; 
\draw[->,thick,densely dotted,out=-45,in=0,distance=0.75cm,shorten <= 0.2cm] (v2) to node[right]{$\sone$} ($(v) + (2.775cm,0.3cm)$);

\path (v) node(v-label){$ \underbrace{(a + b)^*}_{g^*} \,\cdot\, 0
                            \formeq
                          \underbrace{(a \cdot (a + b) + b)^*}_{h^*} \,\cdot\, 0 $}; 
\draw[->,thick,darkcyan,shorten <= 1.1cm,shorten >= 0.5cm] (v) to node[above,pos=0.45]{$\loopnsteplab{1}$}
                                                                  node[below,pos=0.7]{$\black{a}$} (v1);
\draw[->,thick,darkcyan,shorten <= 1.1cm,shorten >= 0.5cm] (v) to node[above,pos=0.45]{$\loopnsteplab{1}$}
                                                                  node[below,pos=0.7]{$\black{b}$} (v2);

\matrix[anchor=center,row sep=1.25cm,column sep=1.1cm,
        every node/.style={draw,very thick,circle,minimum width=2.5pt,fill,inner sep=0pt,outer sep=2pt}
        ] at (6.5,0) {
    \node(v1--2){}; &             & \node(v2--2){};
    \\
                 & \node(v--2){}; 
    \\
  };  
  
\draw[->,thick,darkcyan] (v--2) to node[right,pos=0.7]{$\loopnsteplab{1}$}
                                   node[left,pos=0.45]{$\black{a}$} (v1--2);
\draw[->,thick,darkcyan] (v--2) to node[left,pos=0.7]{$\loopnsteplab{1}$}
                                   node[right,pos=0.45]{$\black{b}$} (v2--2);  

\draw[->] (v1--2) to node[above]{$a, b$} (v2--2);

\draw[<-,very thick,>=latex,chocolate,shorten <= 2pt](v--2) -- ++ (180:0.55cm); 
\draw[->,thick,densely dotted,out=-45,in=0,distance=0.75cm] 
  (v2--2) to node[right]{$\sone$} (v--2);

\path (v--2) ++ (-1.5cm,0.2cm) node(v--2-label) {\large $\onecharthatof{h^* \prod 0}$};

\end{tikzpicture}  \vspace*{-3ex}
  \end{flushleft}
  The correctness conditions at the start vertex (at the bottom) can be verified as follows:
  \begin{align*}
    g^* \prod 0 
      & {} \mathrel{\parbox{\widthof{$\milnersysmineq$}}{$\synteq$}}
        (a + b)^* \prod 0 
      \milnersysmineq
        (1 + (a + b) \prod (a + b)^*) \prod 0
      \milnersysmineq
        1 \prod 0 + ((a + b) \prod g^*) \prod 0 
      \displaybreak[0]\\
      & {} \milnersysmineq
        0 + (a \prod g^* + b \prod g^*) \prod 0 
      \milnersysmineq
        (a \prod g^* + b \prod g^*) \prod 0 
      \displaybreak[0]\\
      & {} \milnersysmineq
        (a \prod g^*) \prod 0 + (b \prod g^*) \prod 0
      \milnersysmineq
        a \prod (g^* \prod 0) + b \prod (g^* \prod 0)
      \displaybreak[0]\\
      & {} \milnersysmineq
        a \prod ((1 \prod g^*) \prod 0) + b \prod ((1 \prod g^*) \prod 0) \punc{,}
    \displaybreak[0]\\[0.75ex]   
    h^* \prod 0
    & {} \mathrel{\parbox{\widthof{$\milnersysmineq$}}{$\synteq$}}
    (a \prod (a + b) + b)^* \prod 0
      \milnersysmineq
    (1 + (a \prod (a + b) + b) \prod (a \prod (a + b) + b)^*) \prod 0  
    \displaybreak[0]\\
    & {} \milnersysmineq
    1 \prod 0 + ( (a \prod (a + b) + b) \prod h^*) \prod 0 
      \milnersysmineq
    0 + ( (a \prod (a + b)) \prod h^*
          + b \prod h^* ) \prod 0     
    \displaybreak[0]\\
    & {} \milnersysmineq
    ( a \prod ((a + b) \prod h^*)
          + b \prod h^* ) \prod 0  
      \milnersysmineq
    ( (a \prod ((a + b) \prod h^*)) \prod 0 
          + ( b \prod h^* ) \prod 0        
    \displaybreak[0]\\
    & {} \milnersysmineq
    a \prod ((a + b) \prod h^*) \prod 0) 
          + b \prod (h^* \prod 0) 
    \displaybreak[0]\\
    & {} \milnersysmineq
    a \prod ((1 \prod (a + b)) \prod h^*) \prod 0)    
          + b \prod ((1 \prod h^*) \prod 0) \punc{.}
  \end{align*}
  From the provable equality for $g^* \prod 0$
    the correctness condition for $(1 \prod g^*) \prod 0$ at the left upper vertex of $\onechartof{h^* \prod 0}$ 
      can be obtained by additional uses of the axiom ($\leftidstexpprod$). 
  The correctness condition for $((1 \prod (a + b)) \prod h^*) \prod 0$  at the left upper vertex of $\onechartof{h^* \prod 0}$ 
  can be verified as follows:
  \begin{align*}
    ((1 \prod (a + b)) \prod h^*) \prod 0
      & {} \mathrel{\parbox{\widthof{$\milnersysmineq$}}{$\milnersysmineq$}}
        ((a + b) \prod h^*) \prod 0 
        \mathrel{\parbox{\widthof{$\milnersysmineq$}}{$\milnersysmineq$}}
          (a \prod h^*  +  b \prod h^*) \prod 0 
      \\
      & {} \mathrel{\parbox{\widthof{$\milnersysmineq$}}{$\milnersysmineq$}}
        (a \prod h^*) \prod 0  +  (b \prod h^*) \prod 0 
        \mathrel{\parbox{\widthof{$\milnersysmineq$}}{$\milnersysmineq$}}
        a \prod (h^* \prod 0)  +  b \prod (h^* \prod 0) 
      \\
      & {} \mathrel{\parbox{\widthof{$\milnersysmineq$}}{$\milnersysmineq$}}
        a \prod ((1 \prod h^*) \prod 0)  +  b \prod ((1 \prod h^*) \prod 0) \punc{.}
  \end{align*}
  Finally, the correctness conditions at the right upper vertex of $\onechartof{h^* \prod 0}$ 
    can be obtained by applications of the axiom~($\leftidstexpprod$)~only.  
\end{exa}

\begin{exa}[LLEE-witnessed coinductive proof in Fig.~\ref{fig:ex:1:RSPstar:to:coindproof}]\label{ex:fig:ex:1:RSPstar:to:coindproof}
  We provided a first illustration for translating 
    an instance of the fixed-point rule into a coinductive proof in Figure~\ref{fig:ex:1:RSPstar:to:coindproof} on page~\pageref{fig:ex:1:RSPstar:to:coindproof}.
  Specifically, we mimicked the instance $\ainst$ (see below) of the fixed-point rule \RSPstar\ in Milner's system $\milnersys = \thplus{\milnersysmin}{\RSPstar}$
    by a coinductive proof (see also below)
    over $\thplus{\milnersysmin}{\setexp{\text{premise of $\ainst$}}}$ with \LLEEwitness~$\onecharthatof{\alert{f^*}\cdot \forestgreen{0}}$:
    \begin{center}
    \AxiomC{$ \overbrace{(a + b)^*}^{\chocolate{{e _0 ^*}}} \chocolate{\mathrel{\cdot} 0} 
                \:\formeq\:
              \overbrace{(a \cdot (a + b) + b)}^{\alert{f}} \cdot (\overbrace{(a + b)^*}^{\chocolate{{e _0 ^*}}} \chocolate{{} \cdot 0}) + \forestgreen{0} $}
    \RightLabel{$\ainst$, $\RSPstar$}
    \UnaryInfC{$ (a + b)^* \cdot 0
                   \:\formeq\:
                 (a \cdot (a + b) + b)^* \,\cdot\, \forestgreen{0} $}
    \DisplayProof
    %
    %
    \AxiomC{$ \chocolate{{e _0 ^*} \cdot 0} \:\formeq\: \alert{f} \cdot (\chocolate{{e _0 ^*} \cdot 0}) + \forestgreen{0} \rule{0pt}{23.5pt}$}
    \UnaryInfC{$ \chocolate{{e _0 ^*} \cdot 0} \:\formeq\: \alert{f^*} \,\cdot\, \forestgreen{0}$}
    \DisplayProof 
  \end{center}
  \vspace*{-2.5ex}
  \begin{center}  
    \begin{tikzpicture}
      \matrix[anchor=center,row sep=1.25cm,column sep=3.5cm,
              every node/.style={draw=none}
              ] at (0,0) {
          \node(v1){}; &[1.5cm]             &[-1.5cm] \node(v2){};
          \\
                       & \node(v){}; 
          \\
        };
        
      \path (v1) node(v1-label){$ (1 \cdot (a + b)) \cdot (\chocolate{{e _0 ^*} \cdot 0}) 
                                    \formeq
                                  ((1 \cdot (a + b)) \cdot \alert{f^*})  \cdot \forestgreen{0} $};       
      \draw[->,shorten <=3.35cm,shorten >=1.75cm] (v1) to node[above,pos=0.6]{$a, b$} (v2);

      \path (v2) node(v2-label){$ 1 \cdot (\chocolate{{e _0 ^*} \cdot 0})
                                    \formeq
                                   (1 \cdot \alert{f^*}) \cdot \forestgreen{0} $}; 
      \draw[->,thick,densely dotted,out=-45,in=0,distance=0.75cm,shorten <= 0.2cm] (v2) to node[right]{$\sone$} ($(v) + (2.35cm,0.125cm)$);

      \path (v) node[xshift=-3.5cm,yshift=-0.15cm](v-label)
                             {$ \underbrace{\chocolate{{e _0 ^*} \cdot 0}}
                                           _{\text{(by the premise of $\ainst$)} \;\;
                                               (a \cdot (a + b) + b) \cdot (\chocolate{{e _0 ^*} \cdot 0}) + \forestgreen{0}
                                                 \:\synteq\:
                                               \alert{f} \cdot (\chocolate{{e _0 ^*} \cdot 0}) + \forestgreen{0} 
                                                 \: = \:
                                             \phantom{
                                                 \: = \:
                                               \alert{f} \cdot (\chocolate{{e _0 ^*} \cdot 0}) + \forestgreen{0}   
                                                 \:\synteq\: 
                                               (a \cdot (a + b) + b) \cdot (\chocolate{{e _0 ^*} \cdot 0}) + \forestgreen{0} 
                                               \text{(by the premise of $\ainst$)} }}
                                    \hspace*{-45ex}
                                  \formeq
                                \underbrace{(a \cdot (a + b) + b)^*}_{\alert{f^*}} \,\cdot\, \forestgreen{0} $}; 
      \draw[->,thick,darkcyan,shorten <= 1.1cm,shorten >= 1cm] (v) to node[above,pos=0.45]{$\loopnsteplab{1}$}
                                                                        node[below,pos=0.525]{$\black{a}$} (v1);
      \draw[->,thick,darkcyan,shorten <= 0.6cm,shorten >= 0.5cm] (v) to node[above,pos=0.45]{$\loopnsteplab{1}$}
                                                                        node[below,pos=0.575]{$\black{b}$} (v2); 
                                                                        
      \path (v) ++ (-5.5cm,0.5cm) node{\Large $ {\onecharthatof{\rule{0pt}{8pt}\alert{f^*}\cdot \forestgreen{0}}} $};                                                                     
    \end{tikzpicture}
  \end{center}
  The correctness conditions for the right-hand sides of this prooftree
  to be a \LLEEwitnessed\ coinductive proof $\thplus{\milnersysmin}{\setexp{\text{premise of $\ainst$}}}$
    are the same as those 
      that we have verified for the right-hand sides of the coinductive proof over $\milnersysmin$ with the same \LLEEwitness\ in Ex.~\ref{ex:inf:ex:1:coindproof}.
    Note that the premise of $\ainst$ is not used for the correctness conditions of the right-hand sides.
  The correctness condition for the left-hand side $\chocolate{{e _0 ^*} \prod 0}$ at the bottom vertex of $\onechartof{\alert{f^*}\cdot \forestgreen{0}}$
  can be verified as follows, now making use of the premise of the considered instance $\ainst$ of $\RSPstar\,$:
  \begin{align*}
    \chocolate{{e _0 ^*} \prod 0}
      & \;\;{\eqin{\setexp{\text{premise of $\ainst$}}}}\;\;
        \alert{f} \prod (\chocolate{{e _0 ^*} \prod 0}) + \forestgreen{0}
      \\
      & \;\;\parbox{\widthof{$\milnersysmineq$}}{$\milnersysmineq$}\;\;
        (a \prod (a + b) + b) \prod (\chocolate{{e _0 ^*} \prod 0})
      \\  
      & \;\;\parbox{\widthof{$\milnersysmineq$}}{$\milnersysmineq$}\;\;
        (a \prod (a + b)) \prod (\chocolate{{e _0 ^*} \prod 0})  +  b \prod (\chocolate{{e _0 ^*} \prod 0})
      \\  
      & \;\;\parbox{\widthof{$\milnersysmineq$}}{$\milnersysmineq$}\;\;
        a \prod ((a + b) \prod (\chocolate{{e _0 ^*} \prod 0})) + b \prod (1 \prod (\chocolate{{e _0 ^*} \prod 0}))
      \\  
      & \;\;\parbox{\widthof{$\milnersysmineq$}}{$\milnersysmineq$}\;\;
        a \prod ((1 \prod (a + b)) \prod (\chocolate{{e _0 ^*} \prod 0})) + b \prod (1 \prod (\chocolate{{e _0 ^*} \prod 0}))
  \end{align*}
  Together this yields the provable equation:
  \begin{align*}
    \chocolate{{e _0 ^*} \prod 0}
      & \;\;{\eqin{\thplus{\milnersysmin}{\setexp{\text{premise of $\ainst$}}}}}\;\;
        a \prod ((1 \prod (a + b)) \prod (\chocolate{{e _0 ^*} \prod 0})) + b \prod (1 \prod (\chocolate{{e _0 ^*} \prod 0})) \punc{,}
  \end{align*}
  which demonstrates the correctness condition for
    the left-hand side $\chocolate{{e _0 ^*} \prod 0}$ at the bottom vertex of $\onechartof{\alert{f^*}\cdot \forestgreen{0}}$.
  The correctness condition for the left-hand side $a \prod ((1 \prod (a + b))$ at the top left vertex of $\onechartof{\alert{f^*}\cdot \forestgreen{0}}$
  can be verified without using the premise of $\ainst$ as follows:
  \begin{align*}
    ((1 \prod (a + b)) \prod (\chocolate{{e _0 ^*} \prod 0})
      & {} \mathrel{\parbox{\widthof{$\milnersysmineq$}}{$\milnersysmineq$}}
        ((a + b) \prod (\chocolate{{e _0 ^*} \prod 0}) 
        \mathrel{\parbox{\widthof{$\milnersysmineq$}}{$\milnersysmineq$}}
          a \prod (\chocolate{{e _0 ^*} \prod 0}) +  b \prod (\chocolate{{e _0 ^*} \prod 0}) 
      \\
      & {} \mathrel{\parbox{\widthof{$\milnersysmineq$}}{$\milnersysmineq$}}
          a \prod (1 \prod (\chocolate{{e _0 ^*} \prod 0})) +  b \prod (1 \prod (\chocolate{{e _0 ^*} \prod 0})) \punc{.}
  \end{align*}
  Finally, the correctness condition of the left-hand side $1 \prod (\chocolate{{e _0 ^*} \prod 0})$ at the right upper vertex of $\onechartof{\alert{f^*} \prod \forestgreen{0}}$ 
    can be obtained by an application of the axiom~($\leftidstexpprod$) only.
\end{exa}

\newpage
\subsection{Supplements for Section~\ref{prelims}}\label{prelims::app}        
   
Antimirov defined partial derivatives of regular expressions in \cite{anti:1996} 
  in a way that facilitates 
  the following equivalent formulation,
  as guaranteed by Proposition~3.1 in \cite{anti:1996}.

\begin{defi}[partial derivatives of star expressions]
  For every star expression $\astexp\in\StExpsover{\actions}$, and $\aact\in\actions$,
    the set $\partderivs{\aact}{\astexp}$ of \emph{partial \derivativesact{\aact} of $\astexp$}
  is defined by means of the function
  $\spartderivs{(\cdot)}{(\cdot)} \funin \actions \times \StExpover{\actions} \to \StExpover{\actions}$, 
  $\astexp \mapsto \partderivs{\aact}{\astexp}$ 
  via by the following recursive definition:\vspace*{-0.5ex}
  \begin{align*}
    \begin{aligned}
      \partderivs{\aact}{\stexpzero} 
        & {} \defdby
            \emptyset \punc{,}
      \\
      \partderivs{\aact}{\stexpone}  
        & {} \defdby 
            \emptyset \punc{,}
      \\
      \partderivs{\aact}{\bact}  
        & {} \defdby 
        \begin{cases} 
          \setexp{\stexpone} 
           & \ldots\; \bact = \aact \punc{,}
          \\
          \emptyset    
           & \ldots\; \bact \neq \aact \punc{,}
        \end{cases}
    \end{aligned}               
    \hspace*{3ex}
    \begin{aligned}
      \partderivs{\aact}{\stexpsum{\astexpi{1}}{\astexpi{2}}}
        & {} \defdby 
      \partderivs{\aact}{\astexpi{1}}
        \cup
      \partderivs{\aact}{\astexpi{2}} \punc{,}
      \\ 
      \partderivs{\aact}{\stexpprod{\astexpi{1}}{\astexpi{2}}}
        & {} \defdby 
        \begin{cases}
           \descsetexp{ \stexpprod{\astexpacci{1}}{\astexpi{2}} }
                      { \astexpacci{1} \in \partderivs{\aact}{\astexpi{1}} }
              & \ldots\; \text{if $\notterminates{\astexpi{1}}$}
           \\
           \descsetexp{ \stexpprod{\astexpacci{1}}{\astexpi{2}} }
                      { \astexpacci{1} \in \partderivs{\aact}{\astexpi{1}} }
              \cup \partderivs{\aact}{\astexpi{2}}
              & \ldots\; \text{if $\terminates{\astexpi{1}}$} \punc{,}
        \end{cases} 
      \\
      \partderivs{\aact}{\stexpit{\astexp}}
        & {} \defdby
      \descsetexp{ \stexpprod{\astexpacc}{\stexpit{\astexp}} }
                 { \astexpacc \in \partderivs{\aact}{\astexp} } \punc{.}
    \end{aligned}    
  \end{align*} 
  By  
  $\actderivs{\astexp} \defdby \descsetexp{ \pair{\aact}{\astexpacc} }{ \aact\in\actions,\, \astexpacc\in\partderivs{\aact}{\astexp} }$
  we denote the set of \emph{action derivatives} of $\astexp\in\StExpover{\actions}$.
\end{defi}

Partial derivatives correspond directly to derivatives of star expressions
  as defined by the transition system specification in Def.~\ref{def:chartof}.

\begin{lem}\label{lem:chartof:vs:partderivs}
  $\partderivs{\aact}{\astexp}
     =
   \descsetexpnormalsize{ \astexpacc }{ \astexp \lt{\aact} \astexpacc }$, 
  for all $\aact\in\actions$ and $\astexp\in\StExpover{\actions}$.  
\end{lem}

\begin{proof}
  By a straightforward induction on the structure of star expressions.
\end{proof}

\begin{lem}\label{lem:actderivs}
  The action derivatives $\actderivs{\astexp}$ of a star expression $\astexp\in\StExpover{\actions}$
    satisfy the following recursive equations, for all $\aact\in\actions$, $\astexpi{1},\astexpi{2},\astexp\in\StExpover{\actions}$:
  \begin{align*}
    \actderivs{0} 
      & {} \defdby \actderivs{1} 
           \defdby \emptyset \punc{,}
    \\
    \actderivs{a}
       & {} \defdby \setexp{ \pair{a}{1} } \punc{,}
    \\
    \actderivs{e_1 + e_2} 
      & {} \defdby  
             \actderivs{e_1} \cup \actderivs{e_2} \punc{,}
    \\
    \actderivs{e_1 \prod e_2}
      & {} \defdby
           \begin{cases}
             \descsetexpbig{ \pair{a}{e'_1 \prod e_2} }{ \pair{a}{e'_1}\in\actderivs{e_1} }
               & \text{if $\notterminates{e_1}$,}
             \\
             \descsetexpbig{ \pair{a}{e'_1 \prod e_2} }{ \pair{a}{e'_1}\in\actderivs{e_1} }
               \cup
             \actderivs{e_2} 
               & \text{if $\terminates{e_1}$,}  
           \end{cases}  
    \\
    \actderivs{e^*}
      & {} \defdby
             \descsetexpbig{ \pair{a}{e' \prod e^*} }
                           { \pair{a}{e'} \in\actderivs{e} } \punc{.}  
  \end{align*}
\end{lem}

\begin{proof}
  By case-wise inspection of the definition of partial derivatives.
\end{proof}

\begin{lem}\label{lem:milnersys:thmequiv:milnersysacc}
  $\milnersys \thmequiv \milnersysacc$.
\end{lem}

\begin{proof}
  It suffices to show that $\RSPstar$ is a derivable rule in $\milnersysacc$, and that $\USPstar$ is a derivable rule in $\milnersys$.
  For the former, we assume the premise 
  $\astexp = \stexpsum{\stexpprod{\stexpit{\bstexp}}{\astexp}}
                      {\cstexp}$
  with $\notterminates{\bstexp}$ of an instance of $\RSPstar$. 
  We have to show that from this premise the conclusion  
  $\astexp = \stexpprod{\stexpit{\bstexp}}{\cstexp}$ of $\RSPstar$ can be obtained by a derivation in $\milnersysacc$.
  By stepwise use of axioms of $\milnersysmin$ we get
  $ \stexpprod{\stexpit{\bstexp}}{\cstexp}
      \milnersysmineq
    \stexpprod{(\stexpsum{\stexpone}
                         {\stexpprod{\bstexp}{\stexpit{\bstexp}}})}{\cstexp}
      \milnersysmineq
    \stexpsum{\stexpprod{\stexpone}{\cstexp}}
             {\stexpprod{(\stexpprod{\bstexp}{\stexpit{\bstexp}})}{\cstexp}}
      \milnersysmineq
    \stexpsum{\stexpprod{\stexpone}{\cstexp}}
             {\stexpprod{\bstexp}{(\stexpprod{\stexpit{\bstexp}}{\cstexp})}}
      \milnersysmineq
    \stexpsum{\cstexp}
             {\stexpprod{\bstexp}{(\stexpprod{\stexpit{\bstexp}}{\cstexp})}}
      \milnersysmineq
    \stexpsum{\stexpprod{\bstexp}{(\stexpprod{\stexpit{\bstexp}}{\cstexp})}}
             {\cstexp} $.
  Hence there is a derivation of
  $ \stexpprod{\stexpit{\bstexp}}{\cstexp}
      = 
    \stexpsum{\stexpprod{\bstexp}{(\stexpprod{\stexpit{\bstexp}}{\cstexp})}}
             {\cstexp} $
  in $\milnersysacc$.
  By applying $\USPstar$ to $\astexp = \stexpprod{\stexpit{\bstexp}}{\cstexp}$ and the formula just reached
  we get a derivation of
  $\astexp = \stexpprod{\stexpit{\bstexp}}{\cstexp}$ in $\milnersysacc$.
  
  For the latter, 
  we consider premise equations
  $\astexpi{1} = \stexpsum{\stexpprod{\stexpit{\bstexp}}{\astexpi{1}}}
                          {\cstexp}$,
  and
  $\astexpi{2} = \stexpsum{\stexpprod{\stexpit{\bstexp}}{\astexpi{2}}}
                          {\cstexp}$
  of an application of $\USPstar$, supposing that $\notterminates{\bstexp}$ holds. 
  By applications of $\RSPstar$ we obtain
  $\astexpi{1} = \stexpprod{\stexpit{\bstexp}}{\cstexp}$,
  and
  $\astexpi{2} = \stexpprod{\stexpit{\bstexp}}{\cstexp}$,
  respectively. 
  By adding an application of $\SYMM$ below the second formula, we obtain
  $\stexpprod{\stexpit{\bstexp}}{\cstexp} = \astexpi{2}$,
  and then by applying $\TRANS$ to $\astexpi{1} = \stexpprod{\stexpit{\bstexp}}{\cstexp}$ and $\stexpprod{\stexpit{\bstexp}}{\cstexp} = \astexpi{2}$
  we obtain the conclusion $\astexpi{1} = \astexpi{2}$ of $\USPstar$. 
\end{proof}

\subsection{Supplements for Section~\ref{LEE}}\label{LEE::app}      
       
\begin{exa}[failure of \LEE]
  We consider again the charts $G_1$ and $G_2$ in Section~\ref{intro::app} 
    that are not expressible by regular expressions modulo bisimilarity:
  
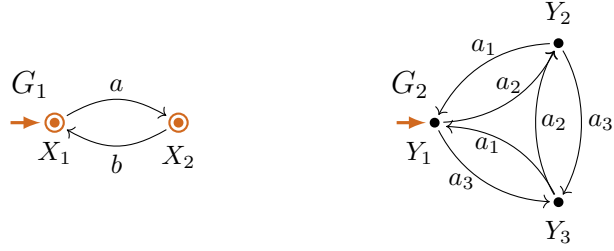
\begin{center}\vspace*{-1.5ex}
    \begin{tikzpicture}

%
\matrix[anchor=center,row sep=0cm,column sep=1.5cm,
        every node/.style={draw,very thick,circle,minimum width=2.5pt,fill,inner sep=0pt,outer sep=2pt}] at (0,0) {
  \node[color=chocolate](C-1-0){};  &   \node[color=chocolate](C-1-1){};
  \\
};
\draw[<-,very thick,>=latex,chocolate,shorten <=2pt](C-1-0) -- ++ (180:0.58cm);
\path(C-1-0) ++ (0cm,-0.4cm) node{$X_1$};
\path(C-1-0) ++ (-0.325cm,0.5cm) node{\Large $\iap{G}{1}$};

\path (C-1-1) ++ (0cm,-0.45cm) node{$X_2$};

\draw[thick,chocolate] (C-1-1) circle (0.12cm);
\draw[thick,chocolate] (C-1-0) circle (0.12cm);
\draw[->,bend left,distance=0.65cm,shorten <=2pt,shorten >=2pt] (C-1-0) to node[above]{$a$} (C-1-1); 
\draw[->,bend left,distance=0.65cm,shorten <=2pt,shorten >=2pt] (C-1-1) to node[below]{$b$} (C-1-0); 


%
\matrix[anchor=center,row sep=0.924cm,column sep=0.75cm,
        every node/.style={draw,very thick,circle,minimum width=2.5pt,fill,inner sep=0pt,outer sep=2pt}] at (5,0) {
                   &                  &  \node(C-2-2){};
  \\
  \node(C-2-1){};  &                  &                  
  \\
                   &                  &  \node(C-2-3){};  
  \\
};
\draw[<-,very thick,>=latex,color=chocolate](C-2-1) -- ++ (180:0.5cm);  

\path(C-2-1) ++ (-0.325cm,0.5cm) node{\Large $\iap{G}{2}$};
\path(C-2-1) ++ (-0.2cm,-0.4cm) node{$Y_1$};
\draw[->,bend right,distance=0.65cm] (C-2-1) to node[above]{$\aacti{2}$} (C-2-2); 
\draw[->,bend right,distance=0.65cm] (C-2-1) to node[left]{$\aacti{3}$}  (C-2-3);

\path(C-2-2) ++ (0cm,0.4cm) node{$Y_2$};
\draw[->,bend right,distance=0.65cm]  (C-2-2) to node[above]{$\aacti{1}$} (C-2-1); 
\draw[->,bend left,distance=0.65cm]  (C-2-2) to node[right,xshift=-1pt]{$\aacti{3}$} (C-2-3);

\path(C-2-3) ++ (0cm,-0.4cm) node{$Y_3$};
\draw[->,bend right,distance=0.65cm] (C-2-3) to node[left]{$\aacti{1}$}  ($(C-2-1)+(+0.15cm,-0.05cm)$);
\draw[->,bend left,distance=0.65cm]  (C-2-3) to node[right,xshift=-1.5pt]{$\aacti{2}$} (C-2-2);

\end{tikzpicture}\vspace*{-1.5ex}
  \end{center}
  Neither $G_1$ nor $G_2$ are loop \onecharts:
  $G_1$ violates (L3), and $G_2$ violates (L2), because it facilitates infinite paths that do not return to the start vertex.
  
  Moreover, neither $G_1$ nor $G_2$ contains a loop \subonechart. 
  Therefore the loop elimination procedure stops immediately on $G_1$ and on $G_2$,
    yet unsuccessfully, because both process graphs still have infinite paths. 
  Consequently neither $G_1$ nor $G_2$ satisfies \LEE\ nor \LLEE.
\end{exa}

By $\oneactderivs{\asstexp} 
      \defdby
        \descsetexpbig{ \pair{\aoneact}{\asstexpacc} }{ \asstexp \lt{\aoneact} \asstexpacc }$
we define the set of \emph{action \onederivatives\ of} a star expression, or stacked star expression, over set $\actions$ of actions,
  where the transitions are defined by the TSS in Def.~\ref{def:onechartof}.                

\begin{lem}\label{lem:actonederivs}
  The action \onederivatives\ $\oneactderivs{\asstexp}$ of a stacked star expression $\asstexp$ over actions in $\actions$
    satisfy the following recursive equations, for all $\aact\in\actions$, $\astexp,\astexpi{1},\astexpi{2}\in\StExpover{\actions}$,
    and stacked star expressions~$E_1$ over actions in $\actions$:
  \begin{align*}
    \oneactderivs{0} 
      & {} \defdby \oneactderivs{1} 
           \defdby \emptyset \punc{,}
    \\
    \oneactderivs{\aact}
       & {} \defdby \setexp{ \pair{\aact}{1} } \punc{,}
    \\
    \oneactderivs{e_1 + e_2} 
      & {} \defdby  
             \oneactderivs{e_1} \cup \oneactderivs{e_2} \punc{,}
    \\
    \oneactderivs{E_1 \prod e_2}
      & {} \defdby
           \begin{cases}
             \descsetexpbig{ \pair{\aoneact}{E'_1 \prod e_2} }{ \pair{\aoneact}{E'_1}\in\oneactderivs{E_1} }
               & \text{if $\notterminates{E_1}$,}
             \\
             \descsetexpbig{ \pair{\aoneact}{E'_1 \prod e_2} }{ \pair{\aoneact}{E'_1}\in\oneactderivs{E_1} }
               \cup
             \oneactderivs{E_2} 
               & \text{if $\terminates{E_1}$,}  
           \end{cases}  
    \\
    \oneactderivs{E_1 \stackprod e^*_2}
      & {} \defdby
           \begin{cases}
             \descsetexpbig{ \pair{\aoneact}{E'_1 \prod e^*_2} }{ \pair{\aoneact}{E'_1}\in\oneactderivs{E_1} }
               & \text{if $\notterminates{E_1}$,}
             \\
             \descsetexpbig{ \pair{\aoneact}{E'_1 \prod e^*_2} }{ \pair{\aoneact}{E'_1}\in\oneactderivs{E_1} }
               \cup
             \setexp{ \pair{\sone}{e^*_2} } 
               & \text{if $\terminates{E_1}$,}  
           \end{cases}  
    \\
    \oneactderivs{e^*}
      & {} \defdby
             \descsetexpbig{ \pair{\aact}{E' \stackprod e^*} }
                           { \pair{\aact}{E'} \in\oneactderivs{e} } \punc{.}  
  \end{align*}
\end{lem}

\begin{proof}
  By case-wise inspection of the definition of the TSS in Def.~\ref{def:onechartof}.
\end{proof}

For the proof of Lem.~\ref{lem:FT:onechart-int} we will need the following auxiliary statement.

\begin{lem}\label{lem:terminates:2:notterminates}
  If $\terminates{e}$ for a star expression $e\in\StExpover{\actions}$,
    then there is a star expression $f\in\StExpover{\actions}$
    with $\notterminates{f}$, $e \milnersysmineq 1 + f$, $\sth{f} = \sth{e}$, 
    and $\compfuns{((\sidfun\times\sproj)}{\soneactderivs)}{f} = \compfuns{((\sidfun\times\sproj)}{\soneactderivs)}{e}$.
\end{lem}     
        
\begin{proof}
  By a proof by induction on structure of $e$,
    in which all axioms of $\milnersysmin$ are used.
\end{proof}

\begin{repeatedlem}[= Lem.~\ref{lem:FT:onechart-int}]
  $ \proj{\asstexp}
     \milnersysmineq
   \terminatesconstof{\onechartof{\asstexp}}{\asstexp}  
     +  
   \sum_{i=1}^{n} \stexpprod{\aoneacti{i}}{\proj{\asstexpacci{i}}} $,
  given a list representation\vspace*{-1pt}
  $\transitionsinfrom{\onechartof{\asstexp}}{\bvert}
     =
   \descsetexpbig{ \asstexp \lt{\aoneacti{i}} \asstexpacci{i} }{ i\in\setexp{1,\ldots,n} }$
  of the transitions from $\asstexp$ in $\onechartof{E}$.
\end{repeatedlem}

\begin{proof}
  The lemma can be established by induction on the structure of stacked star expressions~$E$.
  All cases of stacked star expressions can be dealt with in a quite straightforward manner, 
    except for the case of star expressions with an outermost iteration. 
  There, an appeal to Lem.~\ref{lem:terminates:2:notterminates} is crucial.
  We detail this case below. 
 
  \smallskip
  Suppose that $E \synteq e^*$ for some star expression $e\in\StExpover{\actions}$ (without occurrences of stacked product $\sstexpstackprod$).
  For showing the representation of $E$ as stated by the lemma, we assume 
    that the transitions from $e$ in $\onechartof{e}$ as defined in Def.~\ref{def:onechartof} are as follows:
  \begin{equation}\label{eq:1:lem:FT:onechart-int}
    \transitionsinfrom{\onechartof{e}}{e}
      = \descsetexpbig{ e \lt{\aacti{i}} E'_i }{ i = 1, \ldots n } \punc{,}
  \end{equation}
  for some stacked star expression $E'_1,\ldots,E'_n$, which are \onederivatives\ of $e$. 
  Note that according to the TSS in Def.~\ref{def:onechartof} only proper transitions (those with proper action labels in $\actions$) can depart from $e$.
  Then it follows, again from the TSS in Def.~\ref{def:onechartof} that:
  \begin{equation}\label{eq:3:lem:FT:onechart-int}
    \transitionsinfrom{\onechartof{E}}{E} =
    \transitionsinfrom{\onechartof{e^*}}{e^*}
       = \descsetexpbig{ e^* \lt{\aacti{i}} E'_i \stackprod e^* }{ i = 1, \ldots n } \punc{.}
  \end{equation} 
  We assume now that $\terminates{e}$ holds. (We will see that if $\notterminates{e}$ holds, the argumentation below becomes easier).
  Then by Lem.~\ref{lem:terminates:2:notterminates} there is
    a star expression $f\in\StExpover{\actions}$ 
    with $\alert{\notterminates{f}}$, and such that
  $ f \milnersysmineq e$, $\sth{f} = \sth{e}$, and $\oneactderivs{f} = \oneactderivs{e}$ hold.
  From the latter it follows with \eqref{eq:1:lem:FT:onechart-int}: 
  \begin{align}
    \transitionsinfrom{\onechartof{f}}{f}
      & {} = \descsetexpbig{ f \lt{\aacti{i}} F'_i }{ i = 1, \ldots n } \punc{,}
        \label{eq:2:lem:FT:onechart-int}
    \\    
    \proj{F'_i}
      & {} = \proj{E'_i} \qquad \text{(for all $i = 1, \ldots n$)} \punc{,}
        \label{eq:4:lem:FT:onechart-int}
  \end{align}
  for some stacked star expression $F'_1,\ldots,F'_n$, which are \onederivatives\ of $f$. 
  Note again that only proper transitions can depart from the star expression $f$ according to Def.~\ref{def:onechartof}.
  On the basis of these assumptions we can now argue as follows: 
  \begin{alignat*}{2}
    \proj{E}
      & \;\,\parbox[t]{\widthof{$\eqin{\milnersysmin}$}}{$\synteq$}\:
        \proj{e^*}
        & & \quad\text{(in this case)} 
      \displaybreak[0]\\ 
      & \;\,\parbox[t]{\widthof{$\eqin{\milnersysmin}$}}{$\synteq$}\:
        e^*
        & & \quad\text{(since $e^*$ does not contain $\sstackprod$)}  
      \displaybreak[0]\\ 
      & \;\,\parbox[t]{\widthof{$\eqin{\milnersysmin}$}}{$\milnersysmineq$}\:
        (1 + f)^*
        & & \quad\text{(by the choice of $f$ with $f \milnersysmineq e$)}   
      \displaybreak[0]\\ 
      & \;\,\parbox[t]{\widthof{$\eqin{\milnersysmin}$}}{$\milnersysmineq$}\:
        (f)^*
        & & \quad\text{(by axiom ($\termstexpit$))} 
      \displaybreak[0]\\ 
      & \;\,\parbox[t]{\widthof{$\eqin{\milnersysmin}$}}{$\milnersysmineq$}\:
        1 + f \prod f^*
        & & \quad\text{(by axiom ($\recdefstexpit$))}  
      \displaybreak[0]\\ 
      & \;\,\parbox[t]{\widthof{$\eqin{\milnersysmin}$}}{$\milnersysmineq$}\:
        1 + f \prod e^*
        & & \quad\text{(by the choice of $f$ with $f \milnersysmineq e$)}  
      \displaybreak[0]\\ 
      & \;\,\parbox[t]{\widthof{$\eqin{\milnersysmin}$}}{$\milnersysmineq$}\:
        1 + \bigl( \terminatesconstof{\chartof{f}}{f}
                     +
                   \sum_{i=1}^n
                     \aacti{i} \prod \proj{ F'_i } 
            \bigr) \prod e^*         
        & & \quad\parbox{\widthof{($\sth{f} = \sth{e} < \sth{e} + 1 = \sth{e^*} = \sth{E} $,)}}
                        {(by the induction hypothesis, due to
                         \\\phantom{(}%
                         $\sth{f} = \sth{e} < \sth{e} + 1 = \sth{e^*} = \sth{E}$,
                         \\\phantom{(}
                         and in view of \eqref{eq:2:lem:FT:onechart-int})}  
      \displaybreak[0]\\ 
      & \;\,\parbox[t]{\widthof{$\eqin{\milnersysmin}$}}{$\synteq$}\:
        1 + \bigl( 0
                     +
                   \sum_{i=1}^n
                     \aacti{i} \prod \proj{ E'_i } 
            \bigr) \prod e^*         
        & & \quad\parbox{\widthof{(by $\terminatesconstof{\onechartof{f}}{f} \synteq 0$ due to $\alert{\notterminates{f}}$}}
                        {(by $\terminatesconstof{\onechartof{f}}{f} \synteq 0$ due to $\alert{\notterminates{f}}$
                         \\
                         and by using \eqref{eq:4:lem:FT:onechart-int})}
      \displaybreak[0]\\ 
      & \;\,\parbox[t]{\widthof{$\eqin{\milnersysmin}$}}{$\milnersysmineq$}\:
        \terminatesconstof{\onechartof{e^*}}{e^*}
          + \sum_{i=1}^n
              \aacti{i} \prod ( \proj{ E'_i } \prod e^* ) 
        & & \quad\parbox{\widthof{(by axioms ($\neutralstexpsum$), ($\rdistr$),}}
                        {(by axioms ($\neutralstexpsum$), ($\rdistr$),
                         \\
                          ($\assocstexpprod$), $\terminatesconstof{\onechartof{e^*}}{e^*}\synteq 1$ due to $\terminates{(e^*)}$}  
      \displaybreak[0]\\ 
      & \;\,\parbox[t]{\widthof{$\eqin{\milnersysmin}$}}{$\synteq$}\:
        \terminatesconstof{\onechartof{E}}{E}
          + \sum_{i=1}^n
              \aacti{i} \prod  \proj{ E'_i \stackprod e^* }  
        & & \quad\text{(by the definition of the projection $\sproj$).}  
  \end{alignat*}
  In view of \eqref{eq:3:lem:FT:onechart-int}, 
  this chain of \provablein{\milnersysmin} equalities verifies the statement in the lemma
    in this case $E = e^*$ with $\terminates{e}$. 
  If $\notterminates{e  }$ holds, then the detour via $f$ is not necessary,
    and the argument is much simpler. The statement of the lemma holds true then as well. 
\end{proof}

\begin{repeatedlem}[= Lem.~\ref{lem:onechart-int:milnersysmin:solvable}]
  For every star expression $\astexp\in\StExpover{\actions}$ 
    with \onechart\ interpretation $\onechartof{\astexp} = \tuple{\vertsof{\astexp},\actions,\sone,\astexp,\transs,\exts}$
  the \starexpression\ function $\sasol \funin \vertsof{\astexp} \to \StExpover{\actions}$, $\asstexp \mapsto \proj{\asstexp}$
    is a \provablyin{\milnersysmin} solution of $\onechartof{\astexp}$ with principal value $\astexp$.  
\end{repeatedlem}

\begin{proof}
  Immediate consequence of Lem.~\ref{lem:FT:onechart-int}.
\end{proof}

\subsection{Supplements for Section~\ref{coindmilnersys}}\label{coindmilnersys::app}        

\begin{rem}[the choice for proof-graphs]
  By using `proof-graphs' for defining a coinductive version of Milner's system
  we consciously depart from classical prooftrees.
  While a formulation as a theorem-equivalent natural-deduction system
    similar to that of Brandt and Henglein \cite{bran:heng:1998} for recursive type equality and subtyping
    is possible as well, prooftrees in such a system correspond to graphs with only vertical sharing (defined by $\mu$\nb-terms \cite{blom:2001}).
  Consequently, natural-deduction system cannot have a direct correspondence with all general sharing graphs.
  The property \LEE, which is important for our purpose, does, however, not restrict graphs to those with only vertical sharing;
    for instance, see the \LLEEonechart\ in the example below.
  In order to link the new systems as closely as possible to process graphs 
    whose shapes are only restricted by the property \LEE,
  we therefore decided to use proofs of graph form.  
\end{rem}

\begin{exanonum}[= Ex.~\ref{ex:2:LLEEcoindproof}, solution condition at $\averti{2}$]  
  The statement 
  $\stexpit{(\stexpprod{\stexpit{\aact}}
                       {\stexpit{\bact}})}
     \LLEEcoindproofeqin{\milnersysmin}
   \stexpit{(\stexpsum{\aact}{\bact})}$ 
  can be established by the following \LLEEwitnessed\ coinductive proof $\aCoProof = \pair{\aonechart}{\saeqfun}$ over $\milnersysmin$
    that is witnessed by the \LLEEwitness~$\aonecharthat$ as indicated
    (vertex names are in framed boxes):
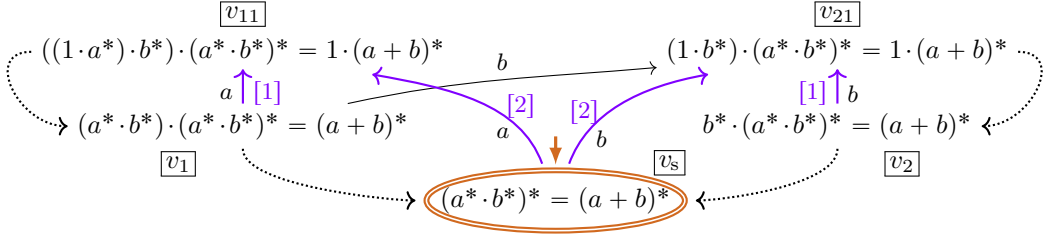
\vspace*{-3.75ex}
  \begin{flushleft}
    \hspace*{-0.8em}
    \begin{tikzpicture}\renewcommand{\stexpprod}[2]{{#1}\hspace*{1pt}{\sstexpprod}\hspace*{1pt}{#2}}
  
\matrix[anchor=center,row sep=0.4cm,column sep=-0.3cm] {
  \node(v11){$ \stexpprod{(\stexpprod{(\stexpprod{\stexpone}{\stexpit{\aact}})}{\stexpit{\bact}})}
                         {\stexpit{(\stexpprod{\stexpit{\aact}}{\stexpit{\bact}})}} 
                 \formeq
               \stexpprod{\stexpone}{\stexpit{(\stexpsum{\aact}{\bact})}} $};
    & & \node(v21){$ \stexpprod{(\stexpprod{\stexpone}{\stexpit{\bact}})}
                               {\stexpit{(\stexpprod{\stexpit{\aact}}{\stexpit{\bact}})}} 
                       \formeq
                     \stexpprod{\stexpone}{\stexpit{(\stexpsum{\aact}{\bact})}} $};
  \\                 
  \node(v1){$ \stexpprod{(\stexpprod{\stexpit{\aact}}{\stexpit{\bact}})}{\stexpit{(\stexpprod{\stexpit{\aact}}{\stexpit{\bact}})}} 
                \formeq
              \stexpit{(\stexpsum{\aact}{\bact})} $};
    & & \node(v2){$ \stexpprod{\stexpit{\bact}}{\stexpit{(\stexpprod{\stexpit{\aact}}{\stexpit{\bact}})}} 
                      \formeq
                    \stexpit{(\stexpsum{\aact}{\bact})} $};
  \\
    & \node(v){$\stexpit{(\stexpprod{\stexpit{\aact}}{\stexpit{\bact}})}
                  \formeq
                \stexpit{(\stexpsum{\aact}{\bact})}$};
  \\
  };   
 
\path (v11) ++ (0cm,0.5cm) node{\tightfbox{$\averti{11}$}};
\draw[->,thick,densely dotted,out=180,in=180,distance=0.65cm](v11) to (v1);

\path (v1) ++ (-0.85cm,-0.5cm) node{\tightfbox{$\averti{1}$}}; 
\draw[->,thick,darkcyan,shorten >= -2pt]
  (v1) to node[left,pos=0.35]{\small $\black{\aact}$} node[right,pos=0.35]{\small $\loopnsteplab{1}$} (v11);
\draw[->,thick,densely dotted,out=-90,in=180,distance=0.5cm,shorten >= 0.2cm](v1) to (v);
\draw[->,shorten <= 0pt,shorten >= 0pt,out=12.5,in=185] (v1) to node[above]{\small $\bact$} (v21); 

\path (v21) ++ (0cm,0.5cm) node{\tightfbox{$\averti{21}$}};
\draw[->,thick,densely dotted,out=0,in=0,distance=0.65cm](v21) to (v2);  
  
\path (v2) ++ (0.85cm,-0.5cm) node{\tightfbox{$\averti{2}$}}; 
\draw[->,thick,darkcyan,shorten >= -2pt]
  (v2) to node[right,pos=0.35]{\small $\black{\bact}$} node[left,pos=0.35]{\small $\loopnsteplab{1}$} (v21);
\draw[->,thick,densely dotted,out=-90,in=0,distance=0.5cm,shorten >= 0.2cm](v2) to (v);

\path (v) ++ (1.5cm,0.5cm) node{\tightfbox{$\start$}}; 
\draw[thick,chocolate,double] (v) ellipse (1.7cm and 0.4cm); 
\draw[<-,very thick,>=latex,chocolate,shorten <= 5pt](v) -- ++ (90:0.85cm); 
\draw[->,thick,darkcyan,out=110,in=-10,shorten <= 0.2cm] 
  (v) to node[below,pos=0.36]{\small $\black{\aact}$} node[above,pos=0.26,xshift=1pt,yshift=1pt]{\small $\loopnsteplab{2}$}  (v11);
\draw[->,thick,darkcyan,out=70,in=190,shorten <= 0.2cm] 
  (v) to node[below,pos=0.37]{\small $\black{\bact}$} node[above,pos=0.25,xshift=-1pt,yshift=1pt]{\small $\loopnsteplab{2}$} (v21);

\end{tikzpicture}    

  \end{flushleft}
  The solution condition at $\averti{2}$ can be verified as follows:
  \begin{align*}
    (a + b)^* 
      & {} \milnersysmineq
    1 + (a + b) \cdot (a + b)^* 
      \milnersysmineq
    1 + a \cdot (a + b)^* + b \cdot (a + b)^*
      \\  
      & {} \milnersysmineq
    1 + a \cdot (1 \cdot (a + b)^*) + b \cdot (a + b)^*) 
    \\
    (a^* \cdot b^*)^*
      & {} \milnersysmineq
    ((1 + a \cdot a^*) \cdot (1 + b \cdot b^*))^*
      \milnersysmineq
    (1 + a \cdot a^* + b \cdot b^* + a \cdot a^* \cdot b \cdot b^*)^*
    \\
      & {} \milnersysmineq
    (a \cdot a^* + b \cdot b^* + a \cdot a^* \cdot b \cdot b^*)^*
      \milnersysmineq
    (a \cdot a^* + a \cdot a^* \cdot b \cdot b^* + b \cdot b^*)^*
    \\
      & {} \milnersysmineq
    (a \cdot a^* \cdot (1 + b \cdot b^*) + b \cdot b^*)^*
      \milnersysmineq
    (a \cdot a^* \cdot b^* + b \cdot b^*)^*
  \end{align*}    
\end{exanonum}

\subsection{Proofs in Section~\ref{coindmilnersys:2:milnersys}}\label{coindmilnersys:2:milnersys::app}        

The statements of Lemma~\ref{lem:lem:extrsol:is:sol}, Lemma~\ref{lem:extrsol:is:sol}, 
  Lemma~\ref{lem:lem:sols:provably:equal:LLEE}, and Lemma~\ref{lem:sols:provably:equal:LLEE}
in Section~\ref{coindmilnersys:2:milnersys}
  are adaptations to \LLEEonecharts\ of the statements about \LLEEcharts\ 
that are expressed by Lemma~5.4, Proposition~5.5, Lemma~5.7, and Proposition~5.8 in \cite{grab:fokk:2020a,grab:fokk:2020b}, respectively.
While the proofs of the adaptations are similar to the proofs in \cite{grab:fokk:2020a,grab:fokk:2020b},
  we provide them here.

\begin{repeatedlem}[=~Lem.~\ref{lem:lem:extrsol:is:sol}]
  Let $\aonechart$ be a weakly guarded \onechart\ $\aonechart$ with guarded \LLEEwitness~$\aonecharthat$. 
  
  Then 
  $\extrsolof{\aonecharthat}{\bvert}
     \milnersysmineq 
   \stexpprod{\extrsoluntilof{\aonecharthat}{\bvert}{\avert}}
             {\extrsolof{\aonecharthat}{\avert}}$
  holds 
  for all vertices $\bvert,\avert\in\vertsof{\aonechart}$ such that $\bvert \convdescendsinlooptosc \avert$.
\end{repeatedlem}

\begin{proof}
  Let $\aonechart = \tuple{\verts,\actions,\sone,\start,\transs,\exts}$ be a \LLEEonechart\
    with \LLEEwitness\ $\aonecharthat$.
  
  We have to show that \mbox{}
  $\extrsol{\bvert}
     \eqin{\milnersysmin}
   \stexpprod{\extrsoluntilof{\aonecharthat}{\bvert}{\avert}}
             {\extrsol{\avert}}$ \mbox{}
  holds for all $\bvert,\avert\in\verts$ with $\bvert \convdescendsinlooptosc \avert$.
  We first notice that this statement holds obviously for $\bvert = \avert$,
  due $ \extrsoluntilof{\aonecharthat}{\bvert}{\avert} \synteq \extrsoluntilof{\aonecharthat}{\avert}{\avert} \synteq \stexpone$,
  and the presence of the axiom $(\leftidstexpprod)$ in $\milnersysmin$. 
  Therefore it suffices to show, by also using this fact, that
  $\extrsol{\bvert}
     \eqin{\milnersysmin}
   \stexpprod{\extrsoluntilof{\aonecharthat}{\bvert}{\avert}}
             {\extrsol{\avert}}$ \mbox{}
  holds for all $\bvert,\avert\in\verts$ with $\bvert \convdescendsinloopto \avert$.
  We will show this by using the same induction as for the definition of the relative extraction function $\sextrsoluntilof{\aonecharthat}$ in Def.~\ref{def:extrsoluntil:extrsol}, that is,
    by complete induction on the (converse) lexicographic partial order $\slexspo$ of $\sconvdescendsinlooptotc$ and $\sconvredtci{\bodylab}$
    on $\verts\times\verts$ defined by:
    $\pair{\bverti{1}}{\averti{1}}
          \lexspo
        \pair{\bverti{2}}{\averti{2}}
          \;\funin\: \Longleftrightarrow\:
            \averti{1} \convdescendsinlooptotc \averti{2}
              \logor
            (\, \averti{1} = \averti{2}
                  \logand 
                \bverti{1} \convredtci{\bodylab} \bverti{2} \,) $,
  which is \wellfounded\ due to Lem.~\ref{lem:descsteps:bodysteps:wf}.               
  For our argument we assume to have given, 
    underlying the definition of the relative extraction function $\sextrsoluntilof{\aonecharthat}$ 
    and the extraction function $\sextrsolof{\aonecharthat}$,
  list\vspace*{-2pt} representations $\transitionsinfrom{\aonecharthat}{\bvert}$ of the transitions from $\bvert$ in $\aonecharthat$ 
  as written in Def.~\ref{def:extrsoluntil:extrsol},
  for all $\bvert\in\verts$.\vspace{1pt}
  
  In order to carry out the induction step, 
  we let $\bvert,\avert\in\verts$ be arbitrary, but such that $\bvert \convdescendsinloopto \avert$. 
  On the basis of the form of $\transitionsinfrom{\aonecharthat}{\bvert}$ 
  as in Def.~\ref{def:extrsoluntil:extrsol}
  we argue as follows,
  starting with a step in which we use the definition of $\sextrsolof{\aonecharthat}$,
  and followed by a second step in which we use that $\terminatesconstof{\aonechart}{\bvert} \synteq \stexpzero$ holds,
    because $\bvert$ cannot have immediate termination as due to $\bvert \convdescendsinloopto \avert$ it is
    in the body of the loop at $\avert$ (cf. condition~(L3) for loop \onecharts\ in Section~\ref{LEE}):
  \begin{align*}
    \extrsolof{\aonecharthat}{\bvert}
      & \;\,\parbox[t]{\widthof{$\eqin{\milnersysmin}$}}{$\synteq$}\:
      \stexpprod{\Bigl(
                   \sum_{i=1}^{n} \stexpprod{\aacti{i}}{\extrsoluntilof{\aonecharthat}{\bverti{i}}{\bvert}}
                 \Bigr)^{\sstar} 
                 }{\Bigr(\stexpsum{\Bigr(\displaystyle
                                     \sum_{i=1}^{m} 
                                       \stexpprod{\boneacti{i}}{\extrsolof{\aonecharthat}{\cverti{i}}}
                                   \Bigl)}
                                  {\terminatesconstof{\aonechart}{\bvert}}
                   \Bigl)} 
    \displaybreak[0]\\
      & \;\,\parbox[t]{\widthof{$\eqin{\milnersysmin}$}}{$\synteq$}\:
      \stexpprod{\Bigl(
                   \sum_{i=1}^{n} \stexpprod{\aacti{i}}{\extrsoluntilof{\aonecharthat}{\bverti{i}}{\bvert}}
                 \Bigr)^{\sstar} 
                 }{\Bigr(\stexpsum{\Bigr(\displaystyle
                                     \sum_{i=1}^{m} 
                                       \stexpprod{\boneacti{i}}{\extrsolof{\aonecharthat}{\cverti{i}}}
                                   \Bigl)}
                                  {\stexpzero}
                   \Bigl)} 
    \displaybreak[0]\\
      & \;\,\parbox[t]{\widthof{$\eqin{\milnersysmin}$}}{$\:\milnersysmineq$}\:
      \stexpprod{\Bigl(
                   \sum_{i=1}^{n} \stexpprod{\aacti{i}}{\extrsoluntilof{\aonecharthat}{\bverti{i}}{\bvert}}
                 \Bigr)^{\sstar} 
                 }{\Bigr(\displaystyle
                     \sum_{i=1}^{m} 
                       \stexpprod{\boneacti{i}}{\extrsolof{\aonecharthat}{\cverti{i}}}
                   \Bigl)} 
        \\
        & \;\,\parbox[t]{\widthof{$\eqin{\milnersysmin}$\hspace*{3ex}}}{\mbox{}}\:
          \text{(by axiom $(\neutralstexpsum)$)}
    \displaybreak[0]\\
      & \;\,\parbox[t]{\widthof{$\eqin{\milnersysmin}$}}{$\eqin{\milnersysmin}$}\:
      \stexpprod{\Bigl(
                   \sum_{i=1}^{n} \stexpprod{\aacti{i}}{\extrsoluntilof{\aonecharthat}{\bverti{i}}{\bvert}}
                 \Bigr)^{\sstar}
                 }{\sum_{i=1}^{m} 
                     \stexpprod{\boneacti{i}}
                               {\stexpprod{(\extrsoluntilof{\aonecharthat}{\cverti{i}}{\avert}}
                                          {\extrsolof{\aonecharthat}{\avert})}}}
        \notag\\
        & \;\,\parbox[t]{\widthof{$\eqin{\milnersysmin}$\hspace*{3ex}}}{\mbox{}}\:
          \parbox[t]{\widthof{\phantom{(if }%
                              imply $\cverti{i} \convdescendsinloopto \avert$ and $\cverti{i} \convredi{\bodylab} \bvert$,
                              and hence $\pair{\cverti{i}}{\avert}  \lexspo  \pair{\bvert}{\avert}$) holds)}}
                         {$\bigl($if $\cverti{i} = \avert$, then 
                             $\extrsolof{\aonecharthat}{\cverti{i}}
                                \milnersysmineq
                              \stexpprod{\extrsoluntilof{\aonecharthat}{\cverti{i}}{\avert}}
                                        {\extrsolof{\aonechart}{\avert})}$
                             due to $\extrsoluntilof{\aonecharthat}{\avert}{\avert} = 1$;   
                          \\
                          \phantom{(}if $\cverti{i} \neq \avert$, we can apply the induction hypothesis
                          to $\extrsolof{\aonecharthat}{\cverti{i}}$,
                          \\
                          \phantom{(if }%
                          as $\bvert \redi{\bodylab} \cverti{i}$ (see $\transitionsinfrom{\aonecharthat}{\bvert}$ 
                                                                      as in Def.~\ref{def:extrsoluntil:extrsol})
                           and $\cverti{i} \neq \avert$
                          \\
                          \phantom{(if }%
                          imply $\cverti{i} \convdescendsinloopto \avert$,
                          and $\cverti{i} \convredi{\bodylab} \bvert$
                            entails $\pair{\cverti{i}}{\avert}  \lexspo  \pair{\bvert}{\avert}$)}
    \displaybreak[0]\\  
      & \;\,\parbox[t]{\widthof{$\eqin{\milnersysmin}$}}{$\eqin{\milnersysmin}$}\:
      \stexpprod{\Bigl(\Bigl(
                   \sum_{i=1}^{n} \stexpprod{\aacti{i}}{\extrsoluntilof{\aonecharthat}{\bverti{i}}{\bvert}}
                 \Bigr)^{\sstar} 
                 }{\stexpprod{\Bigr(
                                \sum_{i=1}^{m} 
                                  \stexpprod{\boneacti{i}}
                                            {\extrsoluntilof{\aonecharthat}{\cverti{i}}{\avert}}
                              \Bigl)\Bigr)}
                             {\extrsolof{\aonecharthat}{\avert}}
                   }
        \\
        & \;\,\parbox[t]{\widthof{$\eqin{\milnersysmin}$\hspace*{3ex}}}{\mbox{}}\:
          \text{(by axioms $(\rdistr)$, and $(\assocstexpprod)$)} 
    \displaybreak[0]\\[1ex]
      & \;\,\parbox[t]{\widthof{$\eqin{\milnersysmin}$}}{$\synteq$}\:
      \stexpprod{ \extrsoluntilof{\aonecharthat}{\bvert}{\avert} }
                  { \extrsolof{\aonecharthat}{\avert} }
        \\
        & \;\,\parbox[t]{\widthof{$\eqin{\milnersysmin}$\hspace*{3ex}}}{\mbox{}}\: 
          \text{(by $\bvert \convdescendsinloopto \avert$, 
                    and the definition of $\extrsoluntilof{\aonecharthat}{\bvert}{\avert}$ in Def.~\ref{def:extrsoluntil:extrsol})}
  \end{align*}
  We note that this reasoning also applies for the special cases $n=0$,
   and with a slight change also for $m=0$, 
    where
    $\sum_{i=1}^{m} 
       \stexpprod{\boneacti{i}}{\extrsolof{\aonecharthat}{\cverti{i}}} 
       \synteq
     \sum_{i=1}^{m} 
       \stexpprod{\boneacti{i}}
                 {\stexpprod{(\extrsoluntilof{\aonecharthat}{\cverti{i}}{\avert}}
                            {\extrsolof{\aonecharthat}{\avert})}}
       \synteq
     \sum_{i=1}^{m} 
     \stexpprod{\boneacti{i}}
               {\extrsoluntilof{\aonecharthat}{\cverti{i}}{\avert}}  
       \synteq
     \stexpzero$,
    and then an axiom ($\stexpzerostexpprod$) has to be used.   
  In this way we have shown, due to $\ACI \subsystem \milnersysmin$, the desired \provablein{\milnersysmin} equality 
  $\asol{\bvert}
     \eqin{\milnersysmin}
   \stexpprod{\extrsoluntilof{\aonecharthat}{\bvert}{\avert}}
             {\asol{\avert}}$ \mbox{}
  for the vertices $\avert$ and $\bvert$ that we picked with the property $\bvert \convdescendsinloopto \avert$.  
  
  Since $\bvert,\avert\in\verts$ with $\bvert \convdescendsinloopto \avert$ were arbitrary for this argument,
  we have successfully carried out the proof by induction that
  \mbox{}
  $\extrsolof{\aonecharthat}{\bvert}
     \eqin{\milnersysmin}
   \stexpprod{\extrsoluntilof{\aonecharthat}{\bvert}{\avert}}
             {\extrsolof{\aonechart}{\avert}}$ \mbox{}
  holds for all $\bvert,\avert\in\verts$ with $\bvert \convdescendsinloopto \avert$.
  As we have argued that the statement also holds for $\bvert = \avert$, we have proved the lemma.
\end{proof}

\begin{repeatedlem}[=~Lem.~\ref{lem:extrsol:is:sol}]
  Let $\aonechart$ 
      be a \LLEEonechart\
  with \LLEEwitness~$\aonecharthat$.
  Then the extraction function $\sextrsolof{\aonecharthat}$ of $\aonecharthat$
  is a \provablein{\milnersysmin} solution~of~$\aonechart$. 
\end{repeatedlem}

\begin{proof}
  Let $\aonechart = \tuple{\verts,\actions,\sone,\start,\transs,\exts}$ be a \LLEEonechart, 
  and let $\aonecharthat$ be a \LLEEwitness\ of $\aonechart$. 
  We show that the extraction function $\sextrsolof{\aonecharthat}$ of $\aonecharthat$ 
  is a \provablein{\milnersysmin} solution of $\aonechart$
  by verifying the \provablein{\milnersysmin} correctness conditions for $\sextrsolof{\aonecharthat}$
  at every vertex $\bvert\in\verts$. 
  For the argument we assume to have given, 
    underlying the definition of the relative extraction function $\sextrsoluntilof{\aonecharthat}$ 
    and\vspace*{-2pt} the extraction function $\sextrsolof{\aonecharthat}$,
  list representations $\transitionsinfrom{\aonecharthat}{\bvert}$ of the transitions from $\bvert$ in $\aonecharthat$ 
  as written in Def.~\ref{def:extrsoluntil:extrsol},
  for all vertices $\bvert\in\verts$. 
 
  We let $\bvert\in\verts$ be arbitrary. 
  Starting from the definition of $\sextrsolof{\aonecharthat}$ in Def.~\ref{def:extrsoluntil:extrsol} 
    on the basis of the form of $\transitionsinfrom{\aonecharthat}{\bvert}$,
  we argue by the following steps:  
  \begin{align*}
    \extrsolof{\aonecharthat}{\bvert}
      & \;\,\parbox[t]{\widthof{$\eqin{\milnersysmin}$}}{$\:\synteq\:$}\:
      \stexpprod{\Bigl(
                   \sum_{i=1}^{n} \stexpprod{\aacti{i}}{\extrsoluntilof{\aonecharthat}{\bverti{i}}{\bvert}}
                 \Bigr)^{\sstar}
                         }{ 
                     \:       
                     \Bigr(\displaystyle
                     \stexpsum{\Bigl(
                                 \sum_{i=1}^{m} \stexpprod{\boneacti{i}}{\extrsolof{\aonecharthat}{\cverti{i}}}
                               \Bigr)}
                              {\terminatesconstof{\aonechart}{\bvert}}
                     \Bigl)}
    \displaybreak[0]\\
    & \;\,\parbox[t]{\widthof{$\eqin{\milnersysmin}$}}{$\eqin{\milnersysmin}$}\:
      \stexpprod{\Bigl(
                 \stexpsum{\sone}
                          {\stexpprod{\Bigl(
                                        \sum_{i=1}^{n} \stexpprod{\aacti{i}}{\extrsoluntilof{\aonecharthat}{\bverti{i}}{\bvert}}
                                      \Bigr)\hspace*{-2pt}}
                                     {\hspace*{-2pt}
                                      \Bigl(
                                        \sum_{i=1}^{n} \stexpprod{\aacti{i}}{\extrsoluntilof{\aonecharthat}{\bverti{i}}{\bvert}}
                                      \Bigr)^{\sstar}}}
                 \Bigr)\hspace*{-2pt}}
                {\hspace*{-2pt}
                 \Bigr(\displaystyle
                 \stexpsum{\Bigl(
                             \sum_{i=1}^{m} \stexpprod{\boneacti{i}}{\extrsolof{\aonecharthat}{\cverti{i}}}
                           \Bigr)}
                          {\terminatesconstof{\aonechart}{\bvert}}
                 \Bigl)}
        \\
        & \;\,\parbox[t]{\widthof{$\eqin{\milnersysmin}$\hspace*{3ex}}}{\mbox{}}\:
          \text{(by axiom $(\recdefstexpit)$)}
    \displaybreak[0]\\
    & \;\,\parbox[t]{\widthof{$\eqin{\milnersysmin}$}}{$\eqin{\milnersysmin}$}\:
      \begin{aligned}[t]
        & 
        \Bigr( 
          \stexpsum{\Bigl(
                      \sum_{i=1}^{m} \stexpprod{\boneacti{i}}{\extrsolof{\aonecharthat}{\cverti{i}}}
                    \Bigr)}
                   {\terminatesconstof{\aonechart}{\bvert}}
        \Bigl)                     
        \\[-1ex]
        &
        \stexpsum{}
                 {\stexpprod{\Bigl(
                               \sum_{i=1}^{n} \stexpprod{\aacti{i}}{\extrsoluntilof{\aonecharthat}{\bverti{i}}{\bvert}}
                             \Bigr)}
                            {\stexpprod{\Bigl( 
                                        \Bigl(
                                          \sum_{i=1}^{n} \stexpprod{\aacti{i}}{\extrsoluntilof{\aonecharthat}{\bverti{i}}{\bvert}}
                                        \Bigr)^{\sstar}}
                                       {\Bigl(
                                          \stexpsum{\Bigl(
                                                      \sum_{i=1}^{m} \stexpprod{\boneacti{i}}{\extrsolof{\aonecharthat}{\cverti{i}}}
                                                    \Bigr)}
                                                   {\terminatesconstof{\aonechart}{\bvert}}
                                        \Bigr)}
                                        \Bigr)}}
      \end{aligned}
        \\
        & \;\,\parbox[t]{\widthof{$\eqin{\milnersysmin}$\hspace*{3ex}}}{\mbox{}}\:
          \text{(by axioms $(\rdistr)$, $(\leftidstexpprod)$, and $(\assocstexpprod)$)}
    \displaybreak[0]\\
    & \;\,\parbox[t]{\widthof{$\eqin{\milnersysmin}$}}{$\synteq$}\:
      \stexpsum{\Bigr( 
                  \stexpsum{\Bigl(
                              \sum_{i=1}^{m} \stexpprod{\boneacti{i}}{\extrsolof{\aonecharthat}{\cverti{i}}}
                            \Bigr)}
                           {\terminatesconstof{\aonechart}{\bvert}}
                \Bigl)}
               {\stexpprod{\Bigl(
                             \sum_{i=1}^{n} \stexpprod{\aacti{i}}{\extrsoluntilof{\aonecharthat}{\bverti{i}}{\bvert}}
                           \Bigr)}
                          {\extrsolof{\aonecharthat}{\bvert}}}
        \\
        & \;\,\parbox[t]{\widthof{$\eqin{\milnersysmin}$\hspace*{3ex}}}{\mbox{}}\:
          \text{(by definition of $\extrsolof{\aonecharthat}{\bvert}$ in Def.~\ref{def:extrsoluntil:extrsol})} 
        \displaybreak[0]\\
    & \;\,\parbox[t]{\widthof{$\eqin{\milnersysmin}$}}{$\eqin{\milnersysmin}$}\:
      \stexpsum{\Bigl(
                  \sum_{i=1}^{n} 
                    \stexpprod{\aacti{i}}
                              {\stexpprod{\bigl(\extrsoluntilof{\aonecharthat}{\bverti{i}}{\bvert}}
                                         {\extrsolof{\aonecharthat}{\bvert}\bigr)}}
                \Bigr)}
               {\Bigr( 
                  \stexpsum{\Bigl(
                              \sum_{i=1}^{m} \stexpprod{\boneacti{i}}{\extrsolof{\aonecharthat}{\cverti{i}}}
                            \Bigr)}
                           {\terminatesconstof{\aonechart}{\bvert}}
                \Bigl)}
        \\
        & \;\,\parbox[t]{\widthof{$\eqin{\milnersysmin}$\hspace*{3ex}}}{\mbox{}}\:
          \text{(by axioms $(\commstexpsum)$, $(\rdistr)$, and $(\assocstexpprod)$)} 
        \displaybreak[0]\\
    & \;\,\parbox[t]{\widthof{$\eqin{\milnersysmin}$}}{$\eqin{\milnersysmin}$}\:
      \stexpsum{\Bigl(
                  \sum_{i=1}^{n} 
                    \stexpprod{\aacti{i}}{\extrsolof{\aonecharthat}{\bverti{i}}}
                \Bigr)}
               {\stexpsum{\Bigl(
                            \sum_{i=1}^{m} \stexpprod{\boneacti{i}}{\extrsolof{\aonecharthat}{\cverti{i}}}
                          \Bigr)} 
                         {\terminatesconstof{\aonechart}{\bvert}}}
        \\
        & \;\,\parbox[t]{\widthof{$\eqin{\milnersysmin}$\hspace*{3ex}}}{\mbox{}}\:
          \parbox[t]{\widthof{(by Lemma~\ref{lem:lem:extrsol:is:sol}, due to $\bverti{i} \convdescendsinlooptosc \bvert$, 
                           which follows from $\bvert \loopnstepto{\alabi{i}} \bverti{i}$}}
                    {(by Lemma~\ref{lem:lem:extrsol:is:sol}, due to $\bverti{i} \convdescendsinlooptosc \bvert$, 
                      which follows from $\bvert \loopnstepto{\alabi{i}} \bverti{i}$%
                     \\
                     \phantom{(}(see $\transitionsinfrom{\aonecharthat}{\bvert}$
                                 as in Def.~\ref{def:extrsoluntil:extrsol}),
                                 and by axioms $(\assocstexpsum)\,$).}
    \\
    & \;\,\parbox[t]{\widthof{$\eqin{\milnersysmin}$}}{$\eqin{\ACI}$}\:
      \stexpsum{\terminatesconstof{\aonechart}{\bvert}}
               {\stexpsum{\Bigl(
                            \sum_{i=1}^{n} 
                              \stexpprod{\aacti{i}}{\extrsolof{\aonecharthat}{\bverti{i}}}
                          \Bigr)}
                         {\Bigl(
                            \sum_{i=1}^{m} \stexpprod{\boneacti{i}}{\extrsolof{\aonecharthat}{\cverti{i}}}
                          \Bigr)}} \punc{.}
  \end{align*}
  Since $\ACI \subsystem \milnersysmin$,
  this chain of equalities yields a \provablein{\milnersysmin} equality
  that establishes,
  in view of $\transitionsinfrom{\aonecharthat}{\bvert}$ 
             as in Def.~\ref{def:extrsoluntil:extrsol},
  the correctness condition for $\sextrsolof{\aonecharthat}$ to be a \provablein{\milnersysmin} solution 
  at the vertex $\bvert$ that we picked. 
  
  Since $\bvert\in\verts$ was arbitrary,
  we have established that the extraction function $\sextrsolof{\aonecharthat}$ of $\aonecharthat$
  is a \provablein{\milnersysmin} solution of $\aonechart$.
\end{proof}

\begin{repeatedlem}[=Lem.~\ref{lem:lem:sols:provably:equal:LLEE}]
  Let $\aonechart$ be a \onechart\ with guarded \LLEEwitness~$\aonecharthat$.
  Let $\asys$ be an \eqlogicbased\ proof system over $\StExpover{\actions}$ 
    such that $\ACI \subsystem \asys \isthmsubsumedby \milnersys$.
  
  Let $\sasol \funin \vertsof{\aonechart} \to \StExpover{\actions}$
      be an \provablein{\asys} solution of $\aonechart$.
  Then 
  $\asol{\bvert}
     \eqin{\milnersys}
   \stexpprod{\extrsoluntilof{\aonecharthat}{\bvert}{\avert}}
             {\asol{\avert}}$ \mbox{}
  holds for all vertices $\bvert,\avert\in\vertsof{\aonechart}$ with $\bvert \convdescendsinlooptosc \avert$.
\end{repeatedlem}

\begin{proof}
  Let $\aonecharthat$ be a guarded \LLEEwitness\ of a \wg\ \onechart~$\aonechart = \tuple{\verts,\actions,\sone,\start,\transs,\exts}$.
  Let $\sasol \funin \verts \to \StExpsover{\actions}$ be an \provablein{\asys} solution of $\aonechart$.
  
  We have to show that \mbox{}
  $\asol{\bvert}
     \eqin{\milnersys}
   \stexpprod{\extrsoluntilof{\aonecharthat}{\bvert}{\avert}}
             {\asol{\avert}}$ \mbox{}
  holds for all $\bvert,\avert\in\verts$ with $\bvert \convdescendsinlooptosc \avert$.
  We first notice that this statement holds obviously for $\bvert = \avert$,
  due $ \extrsoluntilof{\aonecharthat}{\bvert}{\avert} \synteq \extrsoluntilof{\aonecharthat}{\avert}{\avert} \synteq \stexpone$. 
  Therefore it suffices to show, by also using this fact, that
  $\asol{\bvert}
     \eqin{\milnersys}
   \stexpprod{\extrsoluntilof{\aonecharthat}{\bvert}{\avert}}
             {\asol{\avert}}$ \mbox{}
  holds for all $\bvert,\avert\in\verts$ with $\bvert \convdescendsinloopto \avert$.
  We will show this by using the same induction as for the definition of the relative extraction function in Def.~\ref{def:extrsoluntil:extrsol}, that is,
    by complete induction on the (converse) lexicographic partial order $\slexspo$ of $\sconvdescendsinlooptotc$ and $\sconvredtci{\bodylab}$
    on $\verts\times\verts$ defined by:
    $\pair{\bverti{1}}{\averti{1}}
          \lexspo
        \pair{\bverti{2}}{\averti{2}}
          \;\funin\: \Longleftrightarrow\:
            \averti{1} \convdescendsinlooptotc \averti{2}
              \logor
            (\, \averti{1} = \averti{2}
                  \logand 
                \bverti{1} \convredtci{\bodylab} \bverti{2} \,) $,
  which is \wellfounded\ due to Lem.~\ref{lem:descsteps:bodysteps:wf}.               
  For our argument we suppose to have given, 
    underlying the definition of the relative extraction function $\sextrsoluntilof{\aonecharthat}$ 
    and the extraction function $\sextrsolof{\aonecharthat}$,
  list representations $\transitionsinfrom{\aonecharthat}{\bvert}$ of the transitions from $\bvert$ in $\aonecharthat$ 
  as written in Def.~\ref{def:extrsoluntil:extrsol},
  for all $\bvert\in\verts$. 
   
  In order to carry out the induction step, 
  we let $\bvert,\avert\in\verts$ be arbitrary, but such that $\bvert \convdescendsinloopto \avert$. 
  On the basis of the form of $\transitionsinfrom{\aonecharthat}{\bvert}$ 
                              as in Def.~\ref{def:extrsoluntil:extrsol}
  we argue as follows,
  starting with a step in which we use that $\sasol$ is an \provablein{\asys} solution of $\aonechart$,
  and followed by a second step in which we use that $\terminatesconstof{\aonechart}{\bvert} \synteq \stexpzero$ holds,
    because $\bvert$ cannot have immediate termination as due to $\bvert \convdescendsinloopto \avert$ it is
    in the body of the loop at $\avert$ (cf. condition~(L3) for loop \onecharts\ in Section~\ref{LEE}):
  \begin{align*}
    \asol{\bvert}
        &
        \;\,\parbox[t]{\widthof{$\eqin{\milnersysmin}$}}{$\eqin{\asys}$}\:
      \stexpsum{\terminatesconstof{\aonechart}{\bvert}}
               {\stexpsum{\Bigl(\Bigl(
                            \sum_{i=1}^{n}
                              \stexpprod{\aacti{i}}{ \asol{\bverti{i}} }
                          \Bigr)}
                         {\Bigl(
                            \sum_{i=1}^{m}
                              \stexpprod{\boneacti{i}}{ \asol{\cverti{i}} }
                          \Bigr)\Bigr)}}
    \\
        &
        \;\,\parbox[t]{\widthof{$\eqin{\milnersysmin}$}}{$\eqin{\milnersysmin}$}\:
      \stexpsum{\Bigl(
                  \sum_{i=1}^{n}
                    \stexpprod{\aacti{i}}{ \asol{\bverti{i}} }
                \Bigr)}
               {\Bigl(
                  \sum_{i=1}^{m}
                    \stexpprod{\boneacti{i}}{ \asol{\cverti{i}} }
                \Bigr)}
    \displaybreak[0]\\                      
      &
      \;\,\parbox[t]{\widthof{$\eqin{\milnersysmin}$}}{$\eqin{\milnersys}$}\:  
    \stexpsum{\Bigl(
                \sum_{i=1}^{n}
                  \stexpprod{\aacti{i}}
                            {\bigl(\stexpprod{\extrsoluntilof{\aonecharthat}{\bverti{i}}{\bvert}}
                                             {\asol{\bvert}}\bigr)}
              \Bigr)}
             {\Bigl(
                \sum_{i=1}^{m}
                  \stexpprod{\boneacti{i}}
                            {\asol{\cverti{i}}}
              \Bigr)}
        \\
        & \;\,\parbox[t]{\widthof{$\eqin{\milnersysmin}$\hspace*{3ex}}}{\mbox{}}\:
          \parbox[t]{\widthof{\phantom{(if }%
                              because then $\bvert \loopnstepto{\alabi{i}} \bverti{i}$  
                              (see $\transitionsinfrom{\aonecharthat}{\bvert}$ 
                                   as in Def.~\ref{def:extrsoluntil:extrsol}) implies
                              $\bverti{i} \convdescendsinloopto \bvert$,}}
                         {$\bigl($if $\bverti{i} = \bvert$, then 
                             $\asol{\bverti{i}}
                                \milnersyseq
                              \stexpprod{\extrsoluntilof{\aonecharthat}{\bverti{i}}{\bvert}}
                                        {\asol{\bvert})}$
                           due to $\extrsoluntilof{\aonecharthat}{\bvert}{\bvert} = 1$; 
                          \\
                          \phantom{(}if $\bverti{i} \neq \bvert$, we can apply the induction hypothesis
                          to $\asol{\bverti{i}}$,
                          \\
                          \phantom{(if }%
                          because then $\bvert \loopnstepto{\alabi{i}} \bverti{i}$  
                           (see $\transitionsinfrom{\aonecharthat}{\bvert}$ 
                            as in Def.~\ref{def:extrsoluntil:extrsol}) implies
                          $\bverti{i} \convdescendsinloopto \bvert$,
                          \\
                          \phantom{(if }%
                          and due to $\bvert \convdescendsinloopto \avert$ we get $\pair{\bverti{i}}{\bvert} \lexspo \pair{\bvert}{\avert}$)}
    \displaybreak[0]\\[0.75ex]
      &
      \;\,\parbox[t]{\widthof{$\eqin{\milnersysmin}$}}{$\eqin{\milnersys}$}\:  
    \stexpsum{\Bigl(
                \sum_{i=1}^{n}
                  \stexpprod{\aacti{i}}
                            {\bigl(\stexpprod{\extrsoluntilof{\aonecharthat}{\bverti{i}}{\bvert}}
                                             {\asol{\bvert}}\bigr)}
              \Bigr)}
             {\Bigl(
                \sum_{i=1}^{m}
                  \stexpprod{\boneacti{i}}
                            {\bigl(\stexpprod{\extrsoluntilof{\aonecharthat}{\cverti{i}}{\avert}}
                                             {\asol{\avert}}\bigr)}
              \Bigr)}
        \\
        & \;\,\parbox[t]{\widthof{$\eqin{\milnersysmin}$\hspace*{3ex}}}{\mbox{}}\:
          \parbox[t]{\widthof{\phantom{(if }%
                              imply $\cverti{i} \convdescendsinloopto \avert$ and $\cverti{i} \convredi{\bodylab} \bvert$,
                              and hence $\pair{\cverti{i}}{\avert}  \lexspo  \pair{\bvert}{\avert}$) holds)}}
                         {$\bigl($if $\cverti{i} = \avert$, then 
                             $\asol{\cverti{i}}
                                \milnersyseq
                              \stexpprod{\extrsoluntilof{\aonecharthat}{\cverti{i}}{\avert}}
                                        {\asol{\avert})}$
                             due to $\extrsoluntilof{\aonecharthat}{\avert}{\avert} = 1$;   
                          \\
                          \phantom{(}if $\cverti{i} \neq \avert$, we can apply the induction hypothesis
                          to $\asol{\cverti{i}}$,
                          \\
                          \phantom{(if }%
                          as $\bvert \redi{\bodylab} \cverti{i}$ (see $\transitionsinfrom{\aonecharthat}{\bvert}$ 
                                                                  as in Def.~\ref{def:extrsoluntil:extrsol})
                           and $\cverti{i} \neq \avert$
                          \\
                          \phantom{(if }%
                          imply $\cverti{i} \convdescendsinloopto \avert$,
                          and $\cverti{i} \convredi{\bodylab} \bvert$
                            entails $\pair{\cverti{i}}{\avert}  \lexspo  \pair{\bvert}{\avert}$)}
    \displaybreak[0]\\[0.5ex]
      &
      \;\,\parbox[t]{\widthof{$\eqin{\milnersysmin}$}}{$\eqin{\milnersysmin}$}\:  
    \stexpsum{\stexpprod{\Bigl(
                           \sum_{i=1}^{n}
                             \stexpprod{\aacti{i}}
                                       {\extrsoluntilof{\aonecharthat}{\bverti{i}}{\bvert}}
                         \Bigr)}                 
                        {\asol{\bvert}}}
             {\stexpprod{\Bigl(
                           \sum_{i=1}^{m}
                             \stexpprod{\boneacti{i}}
                                       {\extrsoluntilof{\aonecharthat}{\cverti{i}}{\avert}}
                         \Bigr)}              
                        {\asol{\avert}}}
        \\
        & \;\,\parbox[t]{\widthof{$\eqin{\milnersysmin}$\hspace*{3ex}}}{\mbox{}}\:
          \text{(by axioms $(\assocstexpprod)$, and \chocolate{$(\rdistr)$}).}
  \end{align*}       
  We note that these equalities also hold for the special cases in which $n=0$ or/and $m=0$,
  where in the case $m = 0$ an axiom ($\stexpzerostexpprod$) needs to be used in the last step.
  Since $\ACI \subsystem \asys \isthmsubsumedby \milnersys$, and $\milnersysmin \subsystem \milnersys$, 
  we have obtained the following provable equality:  
  \begin{align*}     
    \asol{\bvert}
        &
        \;\,\parbox[t]{\widthof{$\eqin{\milnersysmin}$}}{$\eqin{\milnersys}$}\: 
    \stexpsum{\stexpprod{\Bigl(
                           \sum_{i=1}^{n}
                             \stexpprod{\aacti{i}}
                                       {\extrsoluntilof{\aonecharthat}{\bverti{i}}{\bvert}}
                         \Bigr)}                 
                        {\asol{\bvert}}}
             {\stexpprod{\Bigl(
                           \sum_{i=1}^{m}
                             \stexpprod{\boneacti{i}}
                                       {\extrsoluntilof{\aonecharthat}{\cverti{i}}{\avert}}
                         \Bigr)}              
                        {\asol{\avert}}} \punc{,}
  \end{align*}
  Since 
  $\notterminates{\bigl(
                     \sum_{i=1}^{n}
                       \stexpprod{\aacti{i}}
                                 {\extrsoluntilof{\aonecharthat}{\bverti{i}}{\bvert}}
                   \bigr)}$
  holds, we can apply $\RSPstar$ in order to obtain, and reason further: 
  \begin{align*}
    \asol{\bvert}
      &
      \;\,\parbox[t]{\widthof{$\eqin{\milnersysmin}$}}{$\eqin{\milnersys}$}\:
    \stexpprod{
                        {\Bigl(
                           \sum_{i=1}^{n}
                             \stexpprod{\aacti{i}}
                                       {\extrsoluntilof{\aonecharthat}{\bverti{i}}{\bvert}}
                         \Bigr)^{\!\sstar}\!}
               }           
              {\!
               \Bigl(
               \stexpprod{\Bigl(
                             \sum_{i=1}^{m}
                               \stexpprod{\boneacti{i}}
                                         {\extrsoluntilof{\aonecharthat}{\cverti{i}}{\avert}}
                          \Bigr)}              
                         {\asol{\avert}}
               \Bigr)}
    \displaybreak[0]\\[0.75ex]
      &
      \;\,\parbox[t]{\widthof{$\eqin{\milnersysmin}$}}{$\eqin{\milnersysmin}$}\:
    \stexpprod{\Bigl(
                 \stexpprod{\Bigl(
                              \sum_{i=1}^{n}
                                \stexpprod{\aacti{i}}
                                          {\extrsoluntilof{\aonecharthat}{\bverti{i}}{\bvert}}
                            \Bigr)^{\!\sstar}\!}
                           {\!
                            \Bigl(
                              \sum_{i=1}^{m}
                                \stexpprod{\boneacti{i}}
                                          {\extrsoluntilof{\aonecharthat}{\cverti{i}}{\avert}}
                            \Bigr)}  
               \Bigr)}                                        
              {\asol{\avert}}  
        \\
        & \;\,\parbox[t]{\widthof{$\eqin{\milnersysmin}$\hspace*{3ex}}}{\mbox{}}\:
          \text{(by axiom $(\assocstexpprod)$)}
    \displaybreak[0]\\[0.75ex]
      &
      \;\,\parbox[t]{\widthof{$\eqin{\milnersysmin}$}}{$\synteq$}\:
    \stexpprod{\extrsoluntilof{\aonecharthat}{\bvert}{\avert}}
              {\asol{\avert}}
        \\
        & \;\,\parbox[t]{\widthof{$\eqin{\milnersysmin}$\hspace*{3ex}}}{\mbox{}}\:
          \text{(by $\bvert \convdescendsinloopto \avert$, and the definition of $\extrsoluntilof{\aonecharthat}{\bvert}{\avert}$
                 in Def.~\ref{def:extrsoluntil:extrsol})}
  \end{align*} 
  In this way we have shown, due to $\milnersysmin \subsystem \milnersys$, the desired \provablein{\milnersys} equality 
  $\asol{\bvert}
     \eqin{\milnersys}
   \stexpprod{\extrsoluntilof{\aonecharthat}{\bvert}{\avert}}
             {\asol{\avert}}$ \mbox{}
  for the vertices $\avert$ and $\bvert$ that we picked with the property $\bvert \convdescendsinloopto \avert$.  
  
  Since $\bvert,\avert\in\verts$ with $\bvert \convdescendsinloopto \avert$ were arbitrary for this argument,
  we have successfully carried out the proof by induction 
  \mbox{}
  $\asol{\bvert}
     \eqin{\milnersys}
   \stexpprod{\extrsoluntilof{\aonecharthat}{\bvert}{\avert}}
             {\asol{\avert}}$ \mbox{}
  holds for all $\bvert,\avert\in\verts$ with $\bvert \convdescendsinloopto \avert$.
  As we have argued that the statement also holds for $\bvert = \avert$, we have proved the lemma.
\end{proof}

\begin{repeatedlem}[=Lem.~\ref{lem:sols:provably:equal:LLEE}]
  Let $\aonechart$ be a guarded \LLEEonechart,
  and let $\asys$ be an \eqlogicbased\ proof system over $\StExpover{\actions}$ 
    such that $\ACI \subsystem \asys \isthmsubsumedby \milnersys$.
  
  Then any two \provablein{\asys} solutions of $\aonechart$ are \provablyin{\milnersys} equal.
\end{repeatedlem}
  
\begin{proof}
  Let $\aonecharthat$ be a guarded \LLEEwitness\ of a \wg\ \onechart~$\aonechart = \tuple{\verts,\actions,\sone,\start,\transs,\exts}$.
  In order to show that any two \provablein{\asys} solutions of $\aonechart$
  are \provablyin{\milnersys} equal, 
  it suffices to show that every \provablein{\asys} solution of $\aonechart$ is \provablyin{\milnersys} equal
  to the extraction function $\sextrsolof{\aonecharthat}$ of $\aonecharthat$. 

  For demonstrating this,   
  let $\sasol \funin \verts \to \StExpsover{\actions}$ be an \provablein{\asys} solution of $\aonechart$.
  We have to show that
  $\asol{\bvert}
     \milnersyseq
   \extrsolof{\aonecharthat}{\bvert}$ 
  holds for all $\bvert\in\verts$. 
  We proceed by complete induction on $\sconvredtci{\bodylab}$ (which is well-founded due to Lem.~\ref{lem:descsteps:bodysteps:wf})
    that does not require separate treatment of base cases. 
  For our argument we assume to have given, 
    underlying the definition of the relative extraction function $\sextrsoluntilof{\aonecharthat}$ 
    and the extraction function $\sextrsolof{\aonechart}$,
  list representations $\transitionsinfrom{\aonecharthat}{\bvert}$ of the transitions from $\bvert$ in $\aonecharthat$ 
  as written in Def.~\ref{def:extrsoluntil:extrsol},
  for all $\bvert\in\verts$. 
  
  Let $\bvert\in\verts$ be arbitrary.
  On the basis of $\transitionsinfrom{\aonecharthat}{\bvert}$ 
                              as in Def.~\ref{def:extrsoluntil:extrsol}
  we argue as follows,
  starting with a step in which we use that $\sasol$ is an \provablein{\asys} solution of $\aonechart$
    in view of the form of $\transitionsinfrom{\aonecharthat}{\bvert}\,$: 
  \begin{align}
    \asol{\bvert}
      &
      \;\,\parbox[t]{\widthof{$\eqin{\milnersysmin}$}}{$\eqin{\asys}$}\:
    \stexpsum{\terminatesconstof{\aonechart}{\bvert}}
             {\Bigl(
                \stexpsum{\Bigl(
                            \sum_{i=1}^{n}
                              \stexpprod{\aacti{i}}{ \asol{\bverti{i}} }
                          \Bigr)}
                         {\Bigl(
                            \sum_{i=1}^{m}
                              \stexpprod{\boneacti{i}}{ \asol{\cverti{i}} }
                          \Bigr)}
              \Bigr)}
        \displaybreak[0]
    \notag\\
      &
      \;\,\parbox[t]{\widthof{$\eqin{\milnersysmin}$}}{$\eqin{\ACI}$}\:
    \stexpsum{\Bigl(
                \sum_{i=1}^{n}
                  \stexpprod{\aacti{i}}{ \asol{\bverti{i}} }
              \Bigr)}
             {\Bigl(
                \stexpsum{\Bigl(
                            \sum_{i=1}^{m}
                              \stexpprod{\boneacti{i}}{ \asol{\cverti{i}} }
                          \Bigr)}
                         {\terminatesconstof{\aonechart}{\bvert}
                          \Bigr)}
              \Bigr)}
        \displaybreak[0]
    \notag\\
      &
      \;\,\parbox[t]{\widthof{$\eqin{\milnersysmin}$}}{$\eqin{\milnersys}$}\:  
    \stexpsum{\Bigl(
                \sum_{i=1}^{n}
                  \stexpprod{\aacti{i}}
                            {(\stexpprod{\extrsoluntilof{\aonecharthat}{\bverti{i}}{\bvert}}
                                        {\asol{\bvert}})}
              \Bigr)}
             {\Bigl(
                \stexpsum{\Bigl(
                            \sum_{i=1}^{m}
                              \stexpprod{\boneacti{i}}{ \asol{\cverti{i}} }
                          \Bigr)}
                         {\terminatesconstof{\aonechart}{\bvert}}
              \Bigr)}
        \notag\\[-0.5ex]
        & \;\,\parbox[t]{\widthof{$\eqin{\milnersysmin}$\hspace*{3ex}}}{\mbox{}}\:
          \parbox{\widthof{%
                     ($\bvert \loopnstepto{\alabi{i}} \bverti{i}$ 
                     (due to $\transitionsinfrom{\aonecharthat}{\bvert}$ 
                             as in Def.~\ref{def:extrsoluntil:extrsol}) implies $\bverti{i} \convdescendsinlooptosc \bvert$,}}
                 {($\bvert \loopnstepto{\alabi{i}} \bverti{i}$ 
                    (due to $\transitionsinfrom{\aonecharthat}{\bvert}$ 
                            as in Def.~\ref{def:extrsoluntil:extrsol}) implies $\bverti{i} \convdescendsinlooptosc \bvert$,
                  \\
                  \phantom{(}%
                  from which Lem.~\ref{lem:lem:sols:provably:equal:LLEE} yields
                  $\asol{\bverti{i}} \milnersyseq \stexpprod{\extrsoluntilof{\aonecharthat}{\bverti{i}}{\bvert}}{\asol{\bvert}}\,$)}
      \displaybreak[0]              
    \notag\\
      &
      \;\,\parbox[t]{\widthof{$\eqin{\milnersysmin}$}}{$\:\eqin{\milnersysmin}$}\:  
    \stexpsum{\stexpprod{\Bigl(
                           \sum_{i=1}^{n}
                             \stexpprod{\aacti{i}}
                                       {\extrsoluntilof{\aonecharthat}{\bverti{i}}{\bvert}}
                         \Bigr)}
                        {\asol{\bvert}}}
             {\Bigl(
               \stexpsum{\Bigl(
                           \sum_{i=1}^{m}
                             \stexpprod{\boneacti{i}}{ \asol{\cverti{i}} }
                         \Bigr)}
                        {\terminatesconstof{\aonechart}{\bvert}}
              \Bigr)}
        \notag\\
        & \;\,\parbox[t]{\widthof{$\eqin{\milnersysmin}$\hspace*{3ex}}}{\mbox{}}\:
          \text{(by axioms $(\assocstexpprod)$ and \chocolate{$(\rdistr)$})}
      \displaybreak[0]      
    \notag\\
      &
      \;\,\parbox[t]{\widthof{$\eqin{\milnersysmin}$}}{$\:\eqin{\milnersys}$}\:  
    \stexpsum{\stexpprod{\Bigl(
                           \sum_{i=1}^{n}
                             \stexpprod{\aacti{i}}
                                       {\extrsoluntilof{\aonecharthat}{\bverti{i}}{\bvert}}
                         \Bigr)}
                        {\asol{\bvert}}}
             {\Bigl(
               \stexpsum{\Bigl(
                           \sum_{i=1}^{m}
                             \stexpprod{\boneacti{i}}{ \extrsolof{\aonecharthat}{\cverti{i}} }
                         \Bigr)}
                        {\terminatesconstof{\aonechart}{\bvert}}
              \Bigr)}
        \notag\\
        & \;\,\parbox[t]{\widthof{$\eqin{\milnersysmin}$\hspace*{3ex}}}{\mbox{}}\:
          \parbox{\widthof{(due to $\bvert \redi{\bodylab} \cverti{i}$ (see $\atsiof{\aonecharthat}{\bvert}$ as in Def.~\ref{def:extrsoluntil:extrsol}),
                   and hence $\cverti{i} \convredi{\bodylab} \bvert$,}}
                 {(due to $\bvert \redi{\bodylab} \cverti{i}$ (see $\transitionsinfrom{\aonecharthat}{\bvert}$ 
                                                                   as in Def.~\ref{def:extrsoluntil:extrsol}),
                   and hence $\cverti{i} \convredi{\bodylab} \bvert$,
                  \\[0.25ex]
                  \phantom{(}%
                  $\asol{\cverti{i}} \milnersyseq \extrsolof{\aonecharthat}{\cverti{i}}$ follows from the induction hypothesis).}        
    \notag\displaybreak[0]
    \intertext{Since $\ACI \subsystem \asys \isthmsubsumedby \milnersys$, and $\milnersysmin \subsystem \milnersys$, 
               we have obtained the following provable equality:}
    \displaybreak[0]              
    \asol{\bvert}
      &
      \;\,\parbox[t]{\widthof{$\eqin{\milnersysmin}$}}{$\eqin{\milnersys}$}\:  
    \stexpsum{\stexpprod{\Bigl(
                           \sum_{i=1}^{n}
                             \stexpprod{\aacti{i}}
                                       {\extrsoluntilof{\aonecharthat}{\bverti{i}}{\bvert}}
                         \Bigr)}
                        {\asol{\bvert}}}
             {\Bigl(
               \stexpsum{\Bigl(
                           \sum_{i=1}^{m}
                             \stexpprod{\boneacti{i}}{ \extrsolof{\aonecharthat}{\cverti{i}} }
                         \Bigr)}
                        {\terminatesconstof{\aonechart}{\bvert}}
              \Bigr)}
              \notag
  \end{align}
  Now since 
  $\notterminates{\bigl(
                     \sum_{i=1}^{n}
                       \stexpprod{\aacti{i}}
                                 {\extrsoluntilof{\aonecharthat}{\bverti{i}}{\bvert}}
                   \bigr)}$
  holds, we can apply the rule $\RSPstar$ to this in order to obtain:
  \begin{align*}                                       
    \asol{\bvert}
      &
      \;\,\parbox[t]{\widthof{$\eqin{\milnersysmin}$}}{$\:\eqin{\milnersys}$}\:  
    \stexpprod{
                        \Bigl(
                          \sum_{i=1}^{n}
                            \stexpprod{\aacti{i}}
                                      {\extrsoluntilof{\aonecharthat}{\bverti{i}}{\bvert}}
                        \Bigr)^{\!\sstar}\!
               }{\Bigl(
                   \stexpsum{\Bigl(
                               \sum_{i=1}^{m}
                                 \stexpprod{\boneacti{i}}{ \extrsolof{\aonecharthat}{\cverti{i}} }
                             \Bigr)}
                            {\terminatesconstof{\aonechart}{\bvert}}
                \Bigr)}
      \\
      &
      \;\,\parbox[t]{\widthof{$\eqin{\milnersysmin}$}}{$\synteq$}\: 
    \extrsolof{\aonecharthat}{\bvert}
      \;\,
      \parbox[t]{\widthof{$\eqin{\milnersysmin}$\hspace*{3ex}}}{\mbox{}}\:
        \text{(by the definition of $\sextrsolof{\aonecharthat}$ in Def.~\ref{def:extrsoluntil:extrsol})} 
  \end{align*}
  Thus we have verified the proof obligation
  $\asol{\bvert}
     \milnersyseq
   \extrsolof{\aonecharthat}{\bvert}$
  for the induction step, for the vertex $\bvert$ as picked.
 
  By having performed the induction step, we have successfully carried out the proof by induction on $\sconvredi{\bodylab}$
  that 
  $\asol{\bvert}
     \milnersyseq
   \extrsolof{\aonecharthat}{\bvert}$ holds for all $\bvert\in\verts$,
  and for an arbitrary \provablein{\asys} solution $\sasol$ of $\aonechart$. 
  This implies the statement of the lemma, that any two \provablein{\asys} solutions of $\aonechart$ are \provablyin{\milnersys} equal.  
\end{proof}

%


\subsection{Supplements for Section~\ref{milnersys:2:coindmilnersys}}\label{milnersys:2:coindmilnersys::app}            

\begin{nonexa}\label{ex:RSPstar:2:LLEEcoindproof}
  In order to see where the construction of a \LLEEwitnessed\ coinductive proof goes wrong
    for a non-instance of the fixed-point rule \RSPstar, in which the \sidecondition\ is violated, we consider:\vspace{-0.5ex}
  \begin{center}\renewcommand{\fCenter}{\formeq}
    \Axiom$ \overbrace{(a + c)^*}^{\chocolate{e}}   \fCenter   \overbrace{(a + 1)}^{\alert{f}} \cdot \overbrace{(a + c)^*}^{\chocolate{e}} + \overbrace{1}^{\forestgreen{g}} $
    \RightLabel{$\notRSPstar \;\; \text{since $f \synteq \alert{\terminates{\black{(a + 1)}}}$}$}
    \UnaryInf$ (a + c)^*   \fCenter   \underbrace{(a + 1)^*}_{\alert{f^*}} \cdot \underbrace{1}_{\forestgreen{g}} $
    \DisplayProof
  \end{center}\vspace*{-0.5ex}
  Here
  $ (a + c)^*   
      \milnersysmineq
    (a + 1) \cdot (a + c)^* + 1 $
  holds, 
  and hence by Prop.~\ref{prop:milnersys:sound} this equation is semantically valid.
  But the conclusion is obviously not valid semantically, 
  because the left-hand side can iterated \transitionsact{\cact}, while the right-hand side does not permit \transitionsact{\cact}. 
  When we use the same construction as in Ex.~\ref{ex:1:RSPstar:2:LLEEcoindproof} we obtain:
  
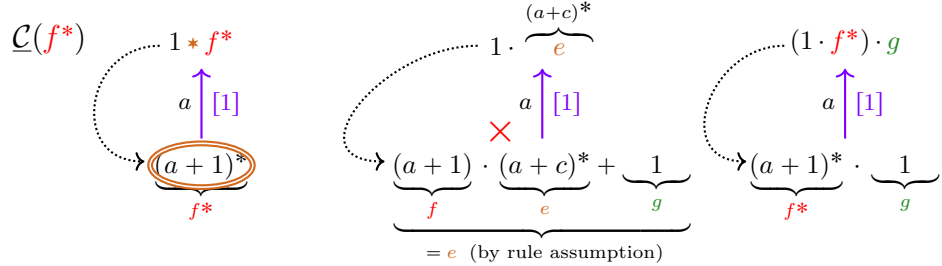
\begin{center}\vspace*{-1ex}
    \begin{tikzpicture}
  \renewcommand{\prod}{\,{\cdot}\,}
  \renewcommand{\stackprod}{\,\chocolate{\varstar}\,}
 
\matrix[anchor=center,row sep=0.8cm,column sep=3.25cm,
        every node/.style={draw=none}
        ] {
    \node(C--1){};      &[1cm]   \node(nonsol--1){};      &[0.5cm]   \node(solfstarg--1){};      
    \\ 
    \node(C--root){};   &          \node(nonsol--root){};   &        \node(solfstarg--root){};         
    \\
  };

\path (C--1) ++ (-2cm,0cm) node[above]{\Large $\onechartof{\alert{f^*}}$};
\path (C--1) ++ (0cm,0cm) node[above]{$ 1 \stackprod \alert{f^*} $};
\draw[->,thick,densely dotted,out=180,in=180,distance=1cm] ($(C--1) + (-0.55cm,0.25cm)$) to ($(C--root) + (-0.75cm,-0.25cm)$); 

\path (C--root) ++ (0cm,0cm) node[below]{$ \underbrace{(a + 1)^*}_{\alert{f^*}} $};
\draw[->,thick,darkcyan,shorten <= -2pt, shorten >= -2pt] (C--root) to node[left]{\small \black{$\aact$}} node[right]{\small $\loopsteplab{1}$} (C--1);

\draw[thick,chocolate,double] ($(C--root) + (0cm,-0.3cm)$) ellipse (0.675cm and 0.29cm);

\path (nonsol--1) ++ (0cm,0cm) node[above]{$ 1 \prod \overbrace{\chocolate{e}}^{(a + c)^*} $};
\draw[->,thick,densely dotted,out=180,in=180,distance=1.25cm] ($(nonsol--1) + (-0.85cm,0.25cm)$) to ($(nonsol--root) + (-2.075cm,-0.25cm)$);
\path (nonsol--root) ++ (0cm,0cm) node[below]
   {$ \underbrace{
         \underbrace{(a + 1)}_{\alert{f}} \prod \underbrace{(a + c)^*}_{\chocolate{e}} +  \underbrace{1}_{\forestgreen{g}}
                  }_{\:=\: \chocolate{e} \;\; \text{(by rule assumption)}}  $};
\draw[->,thick,darkcyan,shorten <= -2pt, shorten >= -2pt] (nonsol--root) to node[left]{\small \black{$\aact$}} node[right]{\small $\loopsteplab{1}$} (nonsol--1);
\path (nonsol--root) ++ (-0.55cm,0.15cm) node{\LARGE \alert{$\boldsymbol{\times}$}};

\path (solfstarg--1) ++ (0cm,0cm) node[above]{$ (1 \prod \alert{f^*}) \prod \forestgreen{g} $};
\draw[->,thick,densely dotted,out=180,in=180,distance=1cm] ($(solfstarg--1) + (-0.8cm,0.25cm)$) to ($(solfstarg--root) + (-1.3cm,-0.25cm)$); 
\path (solfstarg--root) ++ (0cm,0cm) node[below]{$ \underbrace{(a + 1)^*}_{\alert{f^*}} \prod \underbrace{1}_{\forestgreen{g}} $};
\draw[->,thick,darkcyan,shorten <= -2pt, shorten >= -2pt] (solfstarg--root) to node[left]{\small \black{$\aact$}} node[right]{\small $\loopsteplab{1}$} (solfstarg--1);

\end{tikzpicture}
  \end{center}\vspace*{-1ex}
  We recognize that, while 
    $\alert{f^*}\prod\forestgreen{g}$
    is the principal value of a \provablein{\milnersysmin} solution of $\onechartof{\alert{f^*}\prod\forestgreen{g}}$ (see on the right),
  we have not obtained a \provablein{\thplus{\milnersys}{\setexp{\chocolate{e} = \alert{f} \prod \chocolate{e} + \forestgreen{g}}}} 
                         solution of $\onechartof{\alert{f^*}\prod\forestgreen{g}}$
  with principal value $\alert{f} \prod \chocolate{e} + \forestgreen{g}$ (in the middle).
  This is because the correctness condition is violated at the bottom vertex,
    as $(a + 1) \prod (a + c)^* + 1 = 1 + a \prod (1 \prod (a + c)^*)$
      is not provable in $\milnersysmin$
      (using Prop.~\ref{prop:milnersys:sound} and that this equation is not semantically valid,
        since the right-hand side does not permit an initial \transitionact{\cact}). 
  Therefore the construction does not give rise to \LLEEwitnessed\ coinductive proof of the formal equation,
  $ (a + c)^*   
      \formeq
    (a + 1) \cdot (a + c)^* + 1 $, which is not valid semantically.
\end{nonexa}

\begin{repeatedlem}[= Lem.~\ref{lem:lem:mimic:RSPstar}]
  Let $\astexp,\bstexp,\cstexp\in\StExpover{\actions}$ with $\notterminates{\bstexp}$,
  and let $\aseteqs \defdby \setexp{ \astexp = \stexpsum{\stexpprod{\bstexp}{\astexp}}{\cstexp} }$. 
  Then the star expression $\astexp$ is the principal value of a
  \provablein{(\thplus{\milnersysmin}{\aseteqs})} solution of
    the \onechart\ interpretation $\onechartof{\stexpprod{\stexpit{\bstexp}}{\cstexp}}$ of~$\stexpprod{\stexpit{\bstexp}}{\cstexp}$.
\end{repeatedlem}

\begin{proof}
  First, it can be verified that the vertices of $\onechartof{f^* \cdot g}$ are of either of three forms:
  \begin{equation}\label{eq:1:lem:lem:mimic:RSPstar}
    \vertsof{\onechartof{f^* \cdot g}}
      =
    \setexp{ f^* \cdot g }
      \cup
    \descsetexpbig{ (F \stackprod f^*) \cdot g }{ F \in \itpartonederivs{f} }
      \cup
    \descsetexp{ G }{ G \in \itpartonederivs{g} } \punc{,}  
  \end{equation}
  where $\itpartonederivs{f}$ means the set of iterated \onederivatives\ of $f$ according to the TSS in Def.~\ref{def:onechartof}.
  
  This facilitates to define a function $\sasol \funin \vertsof{\onechartof{f^* \cdot g}} \to \StExpsover{\actions}$
    on $\onechartof{f^* \cdot g}$ by:
  \begin{alignat*}{2}
    \asol{f^* \prod g}
      & {} \defdby
             e \punc{,}
    \\
    \asol{(F \stackprod f^*) \prod g}
      & {} \defdby
             \proj{F} \prod e \punc{,}
        & & \qquad \text{(for $F \in \itpartonederivs{f}$),}       
    \\
    \asol{G}
      & {} \defdby
             \proj{G}
        & & \qquad \text{(for $G \in \itpartonederivs{g}$),}   
  \end{alignat*}
  We will show that $\sasol$ is a \provablein{(\thplus{\milnersysmin}{\aseteqs})} solution of $\onechartof{f^* \cdot g}$.
  Instead of verifying the correctness conditions for $\sasol$ for list representations of transitions,
    we will argue more loosely with sums over action \onederivatives\ sets $\oneactderivs{H}$ of stacked star expressions $H$
    where such sums are only \welldefined\ up to \ACI. Due to $\ACI \subsystem \milnersysmin$ such an argumentation 
    is possible. Specifically we will demonstrate, for all $E\in\vertsof{\onechartof{f^* \cdot g}}$, 
    that $\sasol$ is a \provablein{(\thplus{\milnersysmin}{\aseteqs})} solution at $E$, that is, that it holds:
  \begin{equation}\label{eq:2:lem:lem:mimic:RSPstar}
    \asol{E}
      \eqin{\thplus{\milnersysmin}{\aseteqs}}
    \terminatesconstof{\onechartof{E}}{E}
      + \hspace*{-3.75ex}
    \sum_{\pair{\aoneact}{E'}\in\oneactderivs{E}} \hspace*{-3.75ex}
      \aoneact \prod \asol{E'} \punc{,}
  \end{equation}
  where by the sum on the right-hand side we mean an arbitrary representative 
    of the \ACI\ equivalence class of star expressions that can be obtained by the sum expression of this form.  
  
  \smallskip
  
  For showing \eqref{eq:2:lem:lem:mimic:RSPstar}, 
    we distinguish the three cases of vertices $E\in\vertsof{\onechartof{f^* \prod g}}$  
  according to \eqref{eq:1:lem:lem:mimic:RSPstar},
    that is, $E \synteq f^* \prod g$,
             $E \synteq (F \stackprod f^*) \prod g$ for some $F\in\itpartonederivs{f}$,
             and $E \synteq G$ for~some~$G\in\itpartonederivs{g}$. 
  
  In the first case, $E \synteq f^* \prod g$, we find by Lem.~\ref{lem:actonederivs} 
    (or by inspecting the TSS in Def.~\ref{def:onechartof}):
  \begin{equation}\label{eq:2a:lem:lem:mimic:RSPstar}
    \oneactderivs{f^* \prod g}
      =
    \descsetexp{ \pair{\aoneact}{(F \stackprod f^*) \prod g} }
               { \pair{\aoneact}{F} \in \oneactderivs{f} } 
       \cup
    \oneactderivs{g} \punc{.}
  \end{equation}%
  Then we argue as follows:
  \begin{alignat*}{2}
    \asol{E}
      & \;\,\parbox[t]{\widthof{$\eqin{\thplus{\milnersysmin}{\aseteqs}}$}}{$\synteq$}\:
        \asol{f^* \prod g}
        \qquad\quad \text{(in this case)}
      \\  
      & \;\,\parbox[t]{\widthof{$\eqin{\thplus{\milnersysmin}{\aseteqs}}$}}{$\synteq$}\:
        e \phantom{ {} \cdot e + g} 
        \qquad\quad \text{(by the definition of $\sasol$)}
      \\
      & \;\,\parbox[t]{\widthof{$\eqin{\thplus{\milnersysmin}{\aseteqs}}$}}{$\eqin{\thplus{\milnersysmin}{\aseteqs}}$}\:
        f \prod e + g
        \qquad\quad \text{(since $\aseteqs = \setexp{ e = f \prod e + g }$)}  
      \displaybreak[0]\\ 
      & \;\,\parbox[t]{\widthof{$\eqin{\thplus{\milnersysmin}{\aseteqs}}$}}{$\milnersysmineq$}\:
        \Big(
          \terminatesconstof{\onechartof{f}}{f}
            + \hspace*{-3.75ex}
          \sum_{\pair{\aoneact}{F}\in\oneactderivs{f}} \hspace*{-3.5ex}
                 \aoneact \prod \proj{F}
             \Big) \prod e
        +
        \Big(
          \terminatesconstof{\onechartof{g}}{g}
            + \hspace*{-3.75ex}
          \sum_{\pair{\aoneact}{G}\in\oneactderivs{g}} \hspace*{-3.5ex}
                 \aoneact \prod \proj{G}
             \Big)   
        \\[-0.5ex]
        & \;\,\parbox[t]{\widthof{$\eqin{\thplus{\milnersysmin}{\aseteqs}}$\hspace*{3ex}}}{\mbox{}}\: 
          \text{(by using Lem.~\ref{lem:FT:onechart-int})}   
      \displaybreak[0]\\ 
      & \;\,\parbox[t]{\widthof{$\eqin{\thplus{\milnersysmin}{\aseteqs}}$}}{$\milnersysmineq$}\:
        \Big(
          \terminatesconstof{\onechartof{f}}{f}
            \prod e
            + \hspace*{-3.75ex}
          \sum_{\pair{\aoneact}{F}\in\oneactderivs{f}} \hspace*{-3.5ex}
                 \aoneact \prod (\proj{F} \prod e)
             \Big) \prod e
        +
        \Big(
          \terminatesconstof{\onechartof{g}}{g}
            + \hspace*{-3.75ex}
          \sum_{\pair{\aoneact}{G}\in\oneactderivs{g}} \hspace*{-3.5ex}
                 \aoneact \prod \proj{G}
             \Big)   
        \\[-0.5ex]
        & \;\,\parbox[t]{\widthof{$\eqin{\thplus{\milnersysmin}{\aseteqs}}$\hspace*{3ex}}}{\mbox{}}\: 
          \text{(by using ($\rdistr$) and ($\assocstexpprod$))} 
      \displaybreak[0]\\[0.5ex] 
      & \;\,\parbox[t]{\widthof{$\eqin{\thplus{\milnersysmin}{\aseteqs}}$}}{$\milnersysmineq$}\:
        \Big(
            \hspace*{-3.75ex}
          \sum_{\pair{\aoneact}{F}\in\oneactderivs{f}} \hspace*{-3.5ex}
                 \aoneact \prod ( \proj{F} \prod e )
             \Big) 
        + \Bigl(
            \terminatesconstof{\onechartof{f^* \prod g}}{f^* \prod g}
              + \hspace*{-3.75ex}
            \sum_{\pair{\aoneact}{G}\in\oneactderivs{g}} \hspace*{-3.5ex}
                   \aoneact \prod \proj{G} 
          \Bigr)   
        \\[-0.5ex]
        & \;\,\parbox[t]{\widthof{$\eqin{\thplus{\milnersysmin}{\aseteqs}}$\hspace*{3ex}}}{\mbox{}}\: 
          \parbox{\widthof{(since $\terminatesconstof{\onechartof{f}}{f} \synteq 0$ due to $\alert{\notterminates{f}}$,
                            by using axiom ($\neutralstexpsum$)}}
            {(since $\terminatesconstof{\onechartof{f}}{f} \synteq 0$ due to $\alert{\notterminates{f}}$,
                  by using axiom ($\stexpzerostexpprod$),
             \\[-0.25ex]\phantom{(}%
             and $\terminatesconstof{\onechartof{f^* \prod g}}{f^* \prod g} \synteq \terminatesconstof{\onechartof{g}}{g}$)}     
      \displaybreak[0]\\[0.5ex] 
      & \;\,\parbox[t]{\widthof{$\eqin{\thplus{\milnersysmin}{\aseteqs}}$}}{$\ACIeq$}\:
        \terminatesconstof{\onechartof{f^* \prod g}}{f^* \prod g}
           + 
        \Big(
            \hspace*{-3.75ex}
          \sum_{\pair{\aoneact}{F}\in\oneactderivs{f}} \hspace*{-3.5ex}
                 \aoneact \prod ( \asol{(F \stackprod f^*) \cdot g } )
             \Big) 
        +   \hspace*{-3.75ex}
          \sum_{\pair{\aoneact}{G}\in\oneactderivs{g}} \hspace*{-3.5ex}
                 \aoneact \prod \asol{G}    
        \\[-0.5ex]
        & \;\,\parbox[t]{\widthof{$\eqin{\thplus{\milnersysmin}{\aseteqs}}$\hspace*{3ex}}}{\mbox{}}\: 
          \text{(by definition of $\sasol$, axioms ($\commstexpsum$))}     
      \displaybreak[0]\\[0.5ex] 
      & \;\,\parbox[t]{\widthof{$\eqin{\thplus{\milnersysmin}{\aseteqs}}$}}{$\ACIeq$}\:
        \terminatesconstof{\onechartof{E}}{E}
           + 
            \hspace*{-5.75ex}
          \sum_{\pair{\aoneact}{E'}\in\oneactderivs{f^* \cdot g}=\oneactderivs{E}} \hspace*{-5.5ex}
                 \aoneact \prod \asol{E'}  
        \\[-0.5ex]
        & \;\,\parbox[t]{\widthof{$\eqin{\thplus{\milnersysmin}{\aseteqs}}$\hspace*{3ex}}}{\mbox{}}\: 
          \text{(due to \eqref{eq:2a:lem:lem:mimic:RSPstar}).}
  \end{alignat*}  
  Due to $\ACI \subsystem \milnersysmin \subsystem \thplus{\milnersysmin}{\aseteqs}$
    this chain of equalities is provable in $\thplus{\milnersysmin}{\aseteqs}$, 
    which verifies \eqref{eq:2:lem:lem:mimic:RSPstar} for $E$ as considered here,
  or in other words,    
    $\sasol$ is a \provablein{(\thplus{\milnersysmin}{\aseteqs})} solution of at $E$.
  
  \smallskip
  
  In the second case we consider $E \synteq (F \stackprod f^*) \prod g \in\vertsof{\onechartof{f^* \prod g}}$. 
  Then 
  $\terminatesconstof{\onechartof{E}}{E} 
     \synteq 
   \terminatesconstof{\onechartof{(F \stackprod f^*) \prod g}}{(F \stackprod f^*) \prod g}
     \synteq 0$ holds,
  because expressions with stacked product occurring do not have immediate termination by Def.~\ref{def:onechartof}.   
  We distinguish the subcases $\terminates{F}$ and $\notterminates{F}$.
  
  For the first subcase we assume $\notterminates{F}$. 
  Then $\terminatesconstof{\onechartof{F}}{F} \synteq 0$ holds,
    and we find by Lem.~\ref{lem:actonederivs} (or by inspecting the TSS in Def.~\ref{def:onechartof}):
  \begin{equation}\label{eq:2b:lem:lem:mimic:RSPstar}
    \oneactderivs{(F \stackprod f^*) \prod g}
      =
    \descsetexp{ \pair{\aoneact}{(F' \stackprod f^*) \prod g} }
               { \pair{\aoneact}{F'} \in \oneactderivs{F} } \punc{.}
  \end{equation}%
  Now we argue as follows:
  \begin{alignat*}{2}
    \asol{E}
      & \;\,\parbox[t]{\widthof{$\eqin{\milnersysmin}$}}{$\synteq$}\:
        \asol{(F \stackprod f^*) \prod g}
        & & \quad\text{(in this case)}
      \\  
      & \;\,\parbox[t]{\widthof{$\eqin{\milnersysmin}$}}{$\synteq$}\:
        \proj{F} \prod  e
        & & \quad\text{(by the definition of $\sasol$)} 
      \displaybreak[0]\\ 
      & \;\,\parbox[t]{\widthof{$\eqin{\milnersysmin}$}}{$\milnersysmineq$}\:
        \Big(
          \terminatesconstof{\onechartof{F}}{F}
            + \hspace*{-3.75ex}
          \sum_{\pair{\aoneact}{F'}\in\oneactderivs{F}} \hspace*{-3.5ex}
                 \aoneact \prod \proj{F'}
             \Big) \prod e
        & & \quad\text{(by using Lem.~\ref{lem:FT:onechart-int})} 
      \displaybreak[0]\\ 
      & \;\,\parbox[t]{\widthof{$\eqin{\milnersysmin}$}}{$\milnersysmineq$}\:
        0 \prod e 
          +
            \hspace*{-3.75ex}
          \sum_{\pair{\aoneact}{F'}\in\oneactderivs{F}} \hspace*{-3.5ex}
                 \aoneact \prod ( \proj{F'} \prod e )
        & & \quad\parbox{\widthof{(by $\terminatesconstof{\onechartof{F}}{F} \synteq 0$, due to $\notterminates{F}$, and}}%
                        {(by $\terminatesconstof{\onechartof{F}}{F} \synteq 0$, due to $\notterminates{F}$, and
                         \\[-0.5ex]\phantom{(}%
                         axioms ($\rdistr$), ($\assocstexpprod$))}
      \displaybreak[0]\\ 
      & \;\,\parbox[t]{\widthof{$\eqin{\milnersysmin}$}}{$\milnersysmineq$}\:
        0 +
            \hspace*{-3.75ex}
          \sum_{\pair{\aoneact}{F'}\in\oneactderivs{F}} \hspace*{-3.5ex}
                 \aoneact \prod \asol{ (F' \stackprod f^*) \prod g } 
        & & \quad\text{(by ax.\ ($\stexpzerostexpprod$) and def. of $\sasol$)} 
      \\ 
      & \;\,\parbox[t]{\widthof{$\eqin{\milnersysmin}$}}{$\ACIeq$}\:
        \terminatesconstof{\onechartof{(F \stackprod f^*) \prod g}}{ (F \stackprod f^*) \prod g }
          + 
            \hspace*{-5ex}
          \sum_{\pair{\aoneact}{E'}\in\oneactderivs{(F \stackprod f^*) \prod g}} \hspace*{-8ex}
                 \aoneact \prod \asol{ E' } 
        & & \quad\text{(due to \eqref{eq:2b:lem:lem:mimic:RSPstar}, and $\terminatesconstof{\onechartof{E}}{E} \synteq 0)$}    
      \displaybreak[0]\\[0.5ex] 
      & \;\,\parbox[t]{\widthof{$\eqin{\milnersysmin}$}}{$\synteq$}\:
        \terminatesconstof{\onechartof{E}}{E}
           + 
            \hspace*{-3.75ex}
          \sum_{\pair{\aoneact}{E'}\in\oneactderivs{E}} \hspace*{-3.5ex}
                 \aoneact \prod \asol{E'}  
        & & \quad \text{(in this case).}
  \end{alignat*}  
  For the second subcase we assume $\terminates{F}$.
  Then $F\in\StExpover{\actions}$ (that is, $F$ does not contain a stacked product symbol),
    and $\terminatesconstof{\onechartof{F}}{F} \synteq 1$ holds.
  Furthermore, we find, again by inspecting the TSS in Def.~\ref{def:onechartof}:     
  \begin{equation}\label{eq:3:lem:lem:mimic:RSPstar}
    \oneactderivs{(F \stackprod f^*) \prod g}
      =
    \setexp{ \pair{\sone}{f^* \prod g} }
      \cup
    \descsetexp{ \pair{\aoneact}{(F' \stackprod f^*) \prod g} }
               { \pair{\aoneact}{F'} \in \oneactderivs{F} } \punc{.}
  \end{equation}%
  Now we argue as follows:
  \begin{alignat*}{2}
    \asol{E}
      & \;\,\parbox[t]{\widthof{$\eqin{\milnersysmin}$}}{$\synteq$}\:
        \asol{(F \stackprod f^*) \prod e}
        & & \quad\text{(in this case)}
      \\  
      & \;\,\parbox[t]{\widthof{$\eqin{\milnersysmin}$}}{$\synteq$}\:
        \proj{F} \prod  e
        & & \quad\text{(by the definition of $\sasol$)}  
      \displaybreak[0]\\ 
      & \;\,\parbox[t]{\widthof{$\eqin{\milnersysmin}$}}{$\milnersysmineq$}\:
        \Big(
          \terminatesconstof{\onechartof{F}}{F}
            + \hspace*{-3.75ex}
          \sum_{\pair{\aoneact}{F'}\in\oneactderivs{F}} \hspace*{-3.5ex}
                 \aoneact \prod \proj{F'}
             \Big) \prod e
        & & \quad\text{(by using Lem.~\ref{lem:FT:onechart-int})}
      \displaybreak[0]\\ 
      & \;\,\parbox[t]{\widthof{$\eqin{\milnersysmin}$}}{$\milnersysmineq$}\:
        1 \prod e 
          +
            \hspace*{-3.75ex}
          \sum_{\pair{\aoneact}{F'}\in\oneactderivs{f}} \hspace*{-3.5ex}
                 \aoneact \prod ( \proj{F'} \prod e )
        & & \quad\parbox{\widthof{(axioms ($\rdistr$), ($\assocstexpprod$))}}%
                        {(by $\terminatesconstof{\onechartof{F}}{F} \synteq 1$, and
                         \\[0ex]\phantom{(}%
                         axioms ($\rdistr$), ($\assocstexpprod$))}
      \displaybreak[0]\\ 
      & \;\,\parbox[t]{\widthof{$\eqin{\milnersysmin}$}}{$\synteq$}\:
        1 \prod \asol{f^* \prod g} +
            \hspace*{-3.75ex}
          \sum_{\pair{\aoneact}{F'}\in\oneactderivs{f}} \hspace*{-3.5ex}
                 \aoneact \prod \asol{ (F' \stackprod f^*) \prod g } 
        & & \quad\text{(by the definition of $\sasol$)}
      \displaybreak[0]\\ 
      & \;\,\parbox[t]{\widthof{$\eqin{\milnersysmin}$}}{$\ACIeq$}\:
        0 +
            \hspace*{-3.75ex}
          \sum_{\pair{\aoneact}{E'}\in\oneactderivs{(F \stackprod f^*) \prod g}} \hspace*{-7.5ex}
                 \aoneact \prod \asol{ E' } 
        & & \quad \parbox{\widthof{(($\commstexpsum$), and ($\assocstexpsum$))}}
                         {(by \eqref{eq:3:lem:lem:mimic:RSPstar}, using axioms
                          \\\phantom{(}
                          ($\commstexpsum$), and ($\assocstexpsum$))}
      \displaybreak[0]\\ 
      & \;\,\parbox[t]{\widthof{$\eqin{\milnersysmin}$}}{$\synteq$}\:
        \terminatesconstof{\onechartof{E}}{E} 
          + \hspace*{-3.75ex}
        \sum_{\pair{\aoneact}{E'}\in\oneactderivs{E}} \hspace*{-3.5ex}
          \aoneact \prod \asol{E'} 
        & & \quad \text{(in this case, due to $\terminatesconstof{\onechartof{E}}{E} \synteq 0$).}
  \end{alignat*}  
  Due to $\ACI \subsystem \milnersysmin \subsystem \thplus{\milnersysmin}{\aseteqs}$
    the chains of equalities in both subcases are provable in $\thplus{\milnersysmin}{\aseteqs}$,
  and therefore we have now verified \eqref{eq:2:lem:lem:mimic:RSPstar} also in the (entire) second case,
    that is, that $\sasol$ is a \provablein{(\thplus{\milnersysmin}{\aseteqs})} solution of $\onechartof{f^* \prod g}$ at $E$ as in this case. 
  
  \smallskip
  
  In the final case, $E = G$, we argue as follows:
  \begin{alignat*}{2}
    \asol{E}
      & \;\,\parbox[t]{\widthof{$\eqin{\milnersysmin}$}}{$\synteq$}\:
        \asol{G}
        & & \qquad \text{(in this case)}
      \\   
      & \;\,\parbox[t]{\widthof{$\eqin{\milnersysmin}$}}{$\synteq$}\:
        \proj{G} 
        & & \qquad \text{(by the definition of $\sasol$)}
      \displaybreak[0]\\
      & \;\,\parbox[t]{\widthof{$\eqin{\milnersysmin}$}}{$\milnersysmineq$}\:
        \terminatesconstof{\onechartof{G}}{G}
          + \hspace*{-3.75ex}
        \sum_{\pair{\aoneact}{G'}\in\oneactderivs{G}} \hspace*{-3.5ex}
          \aoneact \prod \proj{G'}
        & & \qquad \text{(by using Lem.~\ref{lem:FT:onechart-int})}  
      \displaybreak[0]\\ 
      & \;\,\parbox[t]{\widthof{$\eqin{\milnersysmin}$}}{$\eqin{\ACI}$}\:
        \terminatesconstof{\onechartof{G}}{G} 
          + \hspace*{-3.75ex}
        \sum_{\pair{\aoneact}{G'}\in\oneactderivs{G}} \hspace*{-3.5ex}
          \aoneact \prod \asol{G'} 
        & & \qquad \text{(by the definition\ of $\sasol$)}  
      \displaybreak[0]\\ 
      & \;\,\parbox[t]{\widthof{$\eqin{\milnersysmin}$}}{$\eqin{\ACI}$}\:
        \terminatesconstof{\onechartof{E}}{E} 
          + \hspace*{-3.75ex}
        \sum_{\pair{\aoneact}{E'}\in\oneactderivs{E}} \hspace*{-3.5ex}
          \aoneact \prod \asol{E'} 
        & & \qquad \text{(in this case).}
  \end{alignat*}  
  Due to $\ACI \subsystem \milnersysmin \subsystem \thplus{\milnersysmin}{\aseteqs}$
    this chain of equalities verifies \eqref{eq:2:lem:lem:mimic:RSPstar} also in this case.
    
  By having established \eqref{eq:2:lem:lem:mimic:RSPstar} for the, according to \eqref{eq:1:lem:lem:mimic:RSPstar}, three possible forms of stacked star expressions
    that are vertices of $\onechartof{f^* \prod g}$, we have established
    that $\sasol$ is indeed a \provablein{(\thplus{\milnersysmin}{\aseteqs})} solution of $\onechartof{f^* \prod g}$.
\end{proof}


\end{document}